%% file: paper.tex
\documentclass[11pt,notitlepage,tightenlines,nofootinbib,superscriptaddress]{revtex4-2}

\usepackage[english]{babel}

\usepackage[colorlinks=true, allcolors=blue]{hyperref}

\usepackage{preamble}
\usepackage{latex_notations}

\makeatletter
\@for\cloning@env:=theorem,proposition,lemma,claim,fact,corollary,definition,problem,question,example,remark,notation,conjecture\do{%
  \expandafter\def\csname theH\cloning@env\endcsname{\theHsection.\arabic{theorem}}
}
\makeatother

\begin{document}
\title{Approximate cloning of structured pure states}
\author{Yaroslav Herasymenko}
    \email{yherasymenko@pitp.ca}
    \affiliation{Perimeter Institute for Theoretical Physics, Waterloo, Ontario N2L 2Y5, Canada}
\author{Junqiao Lin}
    \affiliation{QuSoft, CWI, Amsterdam, The Netherlands}
\author{Philip Verduyn Lunel}
    \affiliation{Sorbonne Universit\'e, CNRS, LIP6, France}
\author{Dmitry Grinko}    
    \affiliation{QuSoft, University of Amsterdam, The Netherlands}

\begin{abstract}
The no-cloning theorem is a cornerstone result in quantum mechanics that forbids copying of general quantum information. More quantitatively, even given $N$ copies of an unknown state, any channel producing an $(N+1)$-copy state incurs a nonzero trace-distance error on some inputs. Turning this around, one can ask about the sample complexity of approximate $N\rightarrow N+1$ cloning: for a desired small error $\epsilon$, what $N$ suffices? For arbitrary pure states on a Hilbert space $\cH$, the answer is $N\simeq \dim \cH/\epsilon$ -- astronomically large for most $\cH$ of interest. What if the input is promised to lie in a structured family? We examine a range of families fundamental to many-body physics and quantum information: $n$-mode fermionic Gaussian and Slater states, $n$-mode bosonic Gaussian states, and qudit phase states. For fermions, we construct optimal cloning channels, reducing the complexity from exponential to polynomial in $n$. For bosonic Gaussian states, we construct an explicit cloner that certifies polynomial complexity, without any constraint on the energy of the state; this channel, however, is not optimal in general. The fermionic and bosonic results follow from a single representation-theoretic framework, which generalizes previous work by Werner and by Chiribella and Yang; it also yields the cloners' unitary implementation (a Stinespring dilation). Furthermore, we explain why for many families the cloning complexity is linear in the dimension of the family manifold, via a controlled saddle-point evaluation of the frame potential at high $N$ --- the same calculation accounts for the analogous scaling of tomography complexity. The above framework does not cover phase states; for these we give an independent construction of a channel with optimal cloning fidelity. For all considered families, our results extend to $N\rightarrow N+K$ cloning for $K>1$. Our work complements recent results on cloning stabilizer states, where cloning via learning yields asymptotically optimal complexity; for our families, we show that such a measure-and-prepare cloner cannot reach small $\epsilon$ at any $N$.
\end{abstract}

\maketitle
\tableofcontents

\input{MainText}

\input{NotationsAndDefinitions}

\input{CloningVsLearning}

\input{OrbitsAndCloners}

\input{FreeFermions}

\input{ActiveBosonicCloning}
\input{DimensionSaddlePoint}

\input{ComplementaryChannelsAndCircuits}

\input{PhaseStateCloner}

\section{Acknowledgements}

We thank Alex B.\ Grilo and Jonas Helsen for initial discussions on structured optimal cloning, and we thank Aram Harrow for a discussion regarding the phase states.
We also acknowledge useful discussions with Zongbo Bao, Chris Jackson, and Mykola Semenyakin.
We used large language models for literature searches, proofreading, text editing, and brainstorming at several stages of the project.
In particular, we are indebted to Claude Opus 4.7 for pointing out the degree series representations and the possible use of saddle-point approximation as an explanation for dimension law; ChatGPT 5.6 has helped find the proof of the optimal cloning strategy for the phase states.
We also used Claude Opus 5.5 and GPT-6 Astra for proofreading and text editing, and for assistance with proofs of technical statements, substantially including those in the appendices on bosonic Gaussian cloning, the dimension law, and phase-state optimality (\cref{app:bosonic-cloning-proofs,app:dimension-law-proofs,app:phase-smoothing}).

PVL is supported by France 2030 under the French National Research Agency award number ANR-22-PETQ-0007.
DG is supported by NWO grant NGF.1623.23.025 (“Qudits in theory and experiment”). JL is supported by ERC
STG grant 101040624-ASC-Q. Research of YH at Perimeter Institute is supported in part by the Government of Canada through the Department of Innovation, Science and Economic Development and by the Province of Ontario through the Ministry of Colleges and Universities. YH is also supported by the Applied Quantum Computing Challenge Program at the National Research Council of Canada.

\bibliographystyle{apsrev4-2}
\bibliography{Bibfile}

\appendix
\crefalias{section}{appendix}
\crefalias{subsection}{appendix}
\crefalias{subsubsection}{appendix}

\input{SimpleOrbits}

\input{AppCartanProduct}
\clearpage
\input{AppBosonicCloningProofs}
\clearpage
\input{AppDimensionLawProofs}
\clearpage
\input{AppGaussianHowe}
\input{AppPhaseStates}
\input{CloningBounds}

\end{document}

%% file: MainText.tex
\section{Introduction}

The no-cloning theorem is a foundational result in quantum mechanics, which demonstrates that general quantum information cannot be copied \cite{wootters1982single,dieks1982communication}. For any Hilbert space $\cH$, this statement forbids channels $T$ such that $T(\psi)=\psi^{\otimes2}$ for all pure $\psi$ on $\cH$. More generally, cloning channels from $N$ input copies to $N+K$ output copies --- those always producing $T(\psi^{\otimes N})=\psi^{\otimes N+K}$ --- are forbidden for any $N$ and $K$. Quantum mechanics not only bans \textit{exact} copying of information, but also puts limitations on \textit{approximate} cloning.
Indeed, one may ask how close can one get to achieving $T(\psi^{\otimes N})=\psi^{\otimes N+K}$ --- for example, in terms of the trace distance to the desired outcome?
The pioneering results in \cite{buvzek1996quantum, gisin1997optimal, werner1998optimal} (see also \cite{harrowChurchSymmetricSubspace2013} for a pedagogical introduction) give quantitative bounds on cloning --- in the example of $1\rarr2$ qubit state cloning, the best guaranteed trace distance is $1/3$ \cite{werner1998optimal}. This shows that the no-cloning theorem is fundamentally robust. From a practical perspective, these no-go results delimit the capabilities of quantum computers, by disallowing the `copy' operation on quantum memories \cite{nielsen2001quantum,briegel1998quantum}. 
At the same time, no-cloning theorems are the basis for various applications of quantum information, particularly in quantum communications and cryptography \cite{wiesner1983conjugate,bennett1984quantum,scarani2005quantum, aaronson2009quantum}.

Our work is motivated by the fundamental question: in what settings are quantum states \textit{effectively clonable}? 
Firstly, if the input state $\psi$ is promised to come from a particular family of states $\cS$, standard no-go theorems need not apply~\cite{bruss1998optimal}. As an extreme example \cite{nielsen2001quantum}, if $\cS$ consists of mutually orthogonal states, an exact cloning channel can be constructed (reconciling no-cloning with classical ability to copy information). Furthermore, one can consider near-exact rather than exact cloning, and allow multiple input copies to assist the cloning precision. 
To be more specific, let us call $T$ an $\epsilon$-approximate cloning channel (or cloner) from $N$ to $N+K$ copies, if it extends an input $\psi^{\otimes N}$ to $\psi^{\otimes N+K}$ for any $\psi\in\cS$, up to a deviation $\epsilon$ in the trace distance:
\begin{align}
\sup_{\psi \in \cS} \f{1}{2} \|T(\psi^{\otimes N})-\psi^{\otimes N+K}\|_{1}\leq\epsilon.
\end{align}
The cloners which achieve the best possible cloning precision $\epsilon$ we will call \textit{optimal} \cite{werner1998optimal}. To characterize `clonability of $\cS$' quantitatively, we will consider \textit{cloning sample complexity} $N_{\mr{min},\cS}(\epsilon)$. It is a recently introduced notion of complexity \cite{bansalCloningHardLearning2026}, defined as the answer to the question: for what $N$ do there exist $\epsilon$-approximate cloners from $N$ to $N+1$ copies? (The above notions are introduced formally in \cref{sec:basic_definitions}, including the difference between using fidelity or trace distance as a precision measure.)

How does the cloning sample complexity $N_{\mr{min},\cS}(\epsilon)$ depend on the family $\cS$ and on $\epsilon$? 
By pioneering results on optimal approximate cloning \cite{werner1998optimal}, the unstructured family of all pure states $\cS_{\cH}$ has complexity $N_{\mr{min},\cS_{\cH}}(\epsilon)=\Theta\lrp{\f{\dim \cH}{\epsilon}}$. This is prohibitively large in many cases: if $\cH$ is a many-body Hilbert space of $n$ qudits or fermionic modes, $\dim \cH$ is exponentially large in $n$; for bosons, it is even infinite. In search of more feasible cloning complexities, one can consider smaller, structured families $\cS\subsetneq\cS_{\cH}$. Indeed, results from prior work already establish reduced $N_{\mr{min},\cS}(\epsilon)$ for coherent and squeezed states of a single bosonic mode \cite{braunstein2001coherentCloning, chiribella2014optimal}, and recently also for stabilizer states \cite{bansalCloningHardLearning2026}. Do other families allow for improved cloning complexities?

Motivated by this question, our work examines a range of state families fundamental to many-body physics and quantum information. In particular, we consider fermionic Gaussian and Slater states, bosonic (multi-mode) Gaussian and condensate states, and qudit phase states. As shown below, the structures of these families allow for substantial reductions in cloning sample complexities, certified by explicit cloning channels. Moreover, for all of these families except bosonic Gaussian states, our cloning channel constructions are optimal.

\section{Overview of results}

\subsection{Obstruction to near-exact cloning by learning}

To achieve improved cloning for structured states, a natural instinct is to try and leverage results from learning theory, which was recently developed for many families listed above \cite{chen2026gaussian,arunachalam2022optimal}.
In particular: can one use the initial $N$ copies of $\psi$ to learn its estimate $\tilde{\psi}$ using established tools, and then return $\tilde{\psi}^{\otimes N+K}$? 
It was recently shown that for stabilizer states, this measure-and-prepare approach has the optimal sample complexity $N_{\mr{min},\cS}=\Theta(n)$ \cite{bansalCloningHardLearning2026,montanaroLearningStabilizerStates2017} (for constant $\epsilon$).
Does this also yield optimal complexity for other $\cS$? We find that the answer depends on whether $\cS$ is \textit{discrete} (like the stabilizer states) or \textit{continuous} (like the Gaussian states and the phase states). For a continuous family $\cS$, we find that measure-and-prepare cloning cannot achieve an arbitrary precision $\epsilon$ regardless of $N$ and $K$ (see \cref{sec:cloning_and_learning} for a complete explanation). The fundamental issue is that for such $\cS$, the best estimate $\tilde{\psi}$ has an error $\eta\simeq N^{-1/2}$,
which for $\tilde{\psi}^{\otimes N}$ is amplified to a constant $\epsilon$. This obstruction holds even for the simplest case of pure states on a qubit; by contrast, the standard $N\rightarrow N+1$ cloner of \cite{werner1998optimal} is readily $\epsilon$-approximate for $N=\Theta(1/\epsilon)$ --- but does not have the measure-and-prepare form.
This phenomenon pushes us to look for cloning approaches that do not rely on learning.

\subsection{Gaussian states and resolving orbits}

Our Gaussian-state results distinguish two bosonic families: centered pure Gaussian states $\cS_{\mathrm{sq}}$, whose first moments vanish but whose squeezing and correlations between modes are arbitrary, and all pure Gaussian states $\Sbg$, which also have unknown displacement. For fermions, $\Sffq$ is the $q$-particle subfamily of fermionic Gaussian states $\Sff$ ($\Sffq$ are also known as Slater states). See \cref{sec:optimal_fermionic_gaussian} and \cref{sec:active-bosonic-cloning} for precise definitions of these state families.
\begin{theorem}[Cloning of Gaussian states, informal]
\label{thm:gaussian_cloning_informal}
    For fermionic Gaussian states $\Sff$ and $q$-particle Slater states $\Sffq$ on $n$ modes, explicit optimal cloners give
    \begin{align}
    \label{eq:fermionic_gaussian_complexities}
        N_{\mr{min},\Sff}(\epsilon)=\frac{n(n-1)}{2\epsilon}\bigl[1+O(\epsilon)\bigr],
        \qquad
        N_{\mr{min},\Sffq}(\epsilon)=\frac{q(n-q)}{\epsilon}\bigl[1+O(\epsilon)\bigr],
    \end{align}
    as $\epsilon\to0$, uniformly for $n\geq2$ and $1\leq q\leq n-1$.
    For centered pure bosonic Gaussian states on $n\geq1$ modes, explicit cloners and a matching lower bound give
    \begin{align}
    \label{eq:centered_gaussian_complexity_informal}
        N_{\mr{min},\cS_{\mathrm{sq}}}(\epsilon)=\Theta_\epsilon(n)
        \qquad\text{for every fixed }0<\epsilon<1.
    \end{align}  
    For general bosonic Gaussian states, we obtain
    \begin{align}
    \label{eq:bosonic_gaussian_complexity}
        \frac{n}{-\log(1-\epsilon)}-1
        \leq N_{\mr{min},\Sbg}(\epsilon)
        \leq\left\lceil\frac{n(n+3)}{2\epsilon}+2n\right\rceil,
        \qquad 0<\epsilon<1.
    \end{align}
    The upper bounds are constructive; the lower bounds apply to arbitrary channels (\cref{cor:ab-centered-sample-complexity,cor:ab-gaussian-sample-complexity}).
\end{theorem}

These results give an exponential improvement over unstructured cloning for fermions and a finite sample complexity for bosons, even when the unknown state has arbitrarily large energy.
For centered bosonic states, the linear dependence on $n$ is optimal at each fixed error; the sharp joint dependence on $n$ and $\epsilon$ remains open.
For the full Gaussian family, whether the quadratic upper bound can be improved to linear in $n$ is also open.

The optimal fermionic cloners and the quadratic bosonic upper bound follow from a common construction for families we call \emph{resolving orbits}.
The linear centered-state bound uses a different channel, obtained by restricting to a smaller symmetry group, as explained below.
Before we state the theorem for resolving orbits families $\cS$, let us explain this terminology (cf. the more formal \cref{def:resolving_orbit}).
First, we demand that $\cS$ is an orbit of a group $G$ under a representation $\phi$; for $\cS$ to be `resolving', we demand that for any $N$, representation $\phi^{\otimes N}$ restricted to a \emph{particular subspace} of $\cH^{\otimes N}$ has nice properties.
The subspace is the `$N$-copy span', i.e. the span of $\{\ket{\psi}^{\otimes N},\psi\in\cS\}$.
If thus restricted representation is irreducible and square-integrable for $N\geq L$, we say that $\cS$ is an $L$-resolving orbit; or simply `resolving' if $L=1$. Square-integrability is relevant for infinite-dimensional cases. The reason for the name `resolving' is that for such $\cS$, the integral $\int_{\psi\in\cS} \psi^{\otimes N}$ is well-defined and gives the resolution of identity on the $N$-copy span. A bit more precisely (see \cref{sec:orbits_cloners} for complete formalities), for $N$-resolving $\cS$ we have
\begin{align}
\label{eq:resolution_informal}
   \int_{\psi\in\cS} \psi^{\otimes N} = \frac{1}{d_N}\pi_N,
\end{align}
where $\pi_N$ is the projector onto $N$-copy span and $d_N<\infty$ is a number called formal degree. If $\dim \cH <\infty$, $d_N$ is simply the dimension of that span! But when this dimension is infinite, a finite formal degree $d_N<\infty$ is still guaranteed by our square-integrability condition.
The promised general theorem is

\begin{theorem}[Cloning of states from resolving orbits, informal]
\label{thm:resolving_cloning_informal}
If a family $\cS$ is an $N$-resolving orbit, a map $T$ defined as
\begin{align}
\label{eq:optimal_cloner_informal}
    T(\rho_{N})\coloneq\f{d_{N}}{d_{N+K}}\pi_{N+K} \lrp{\rho_{N}\otimes \bI_{K}} \pi_{N+K}
\end{align}
is an $\epsilon$-precise $N\rarr N+K$ cloning channel. Accuracy of $T$ is given by $\epsilon=1-f_{T,\cS}$, where $f_{T,\cS}=d_{N}/d_{N+K}$ has the meaning of cloning fidelity.
For compact $G$ and $\dim \cH<\infty$, the cloning channel $T$ is also \emph{optimal}.
\end{theorem}

Note that the tight relation $\epsilon=1-f_{T,\cS}$ between trace distance and fidelity is specific to this channel.

\cref{thm:resolving_cloning_informal} is formally stated and proven in \cref{sec:orbits_cloners} as \cref{thm:structured_optimal_cloning}.
This generalizes the cloner for unstructured pure states proposed originally by Werner in \cite{werner1998optimal} and the cloner for generalized coherent states of \cite{chiribella2014optimal}.
Indeed, the family of all pure states $\cS_{\cH}$ considered in \cite{werner1998optimal} is a resolving orbit: its $N$-copy span is the symmetric subspace of $\cH^{\otimes N}$, which hosts an irreducible representation of $\rU(\cH)$.
For Werner, its irreducibility was key to computing the precision and proving optimality of the cloning channel given in \cref{eq:optimal_cloner_informal}.
Here, we show that this idea generalizes to any resolving orbit; we refer to the cloner of \cref{eq:optimal_cloner_informal} as `canonical cloner'.
The proof is similar to that of \cite{werner1998optimal} but has to take a couple of alternative routes to allow for infinite-dimensionality, made possible by our square-integrability assumption. The result of \cite{chiribella2014optimal} covers the finite-dimensional resolving orbits generated by highest-weight vectors (such as the Slater states $\Sffq$; see below).

Why does \cref{thm:resolving_cloning_informal} apply to the Gaussian state families --- why do they form resolving orbits? One way to see it is to examine the so-called Howe dualities, focusing on the representations which are trivial in the copy space. 
That approach uses fairly heavy mathematical machinery, so we present more transparent and self-contained proofs using standard notions from representation theory (\cref{sec:optimal_fermionic_gaussian} and \cref{sec:active-bosonic-cloning}). (However, we do resort to Howe dualities to realize a Stinespring dilation, as explained below.)

Let us first sketch the fermionic case. One starts with the basic fact, that Gaussian states $\Sff$ form an orbit of the so-called $\Pin (2n)$ group, and its $q$-particle subset $\Sffq$ --- that of the $U(n)$ group.
In physics terms, these are the Bogoliubov transformations and the single-particle transformations. To obtain the resolving property, one needs to show that these groups are represented irreducibly on the $N$-copy spans.
For Slater states $\Sffq$, one observes that the base of the orbit (the `Fermi sea state') is the highest-weight vector of the respective $U(n)$ representation.
Using the Cartan product theorem, it follows that the $N$th tensor power of these vectors is also highest-weight and generates an irreducible representation on the $N$-copy span; the use of highest-weight theory is analogous to \cite{chiribella2014optimal}.
The resulting representations on $N$-copy spans are standard, and the dimensionality $d_N$ is given by the Weyl product formula.
The fermionic complexities in \cref{thm:gaussian_cloning_informal} then follow from the asymptotics of the cloning fidelity $d_{N}/d_{N+1}$.
The proof for general Gaussian states $\Sff$ is similar but with a couple of extra steps, because the group $\Pin (2n)$ is disconnected and thus the highest-weight theory does not directly apply.
See \cref{sec:optimal_fermionic_gaussian} for more details.

For the bosonic constructions, see \cref{sec:active-bosonic-cloning}. There, we describe Gaussian states $\cS_{\mathrm{sq}}$ as the vacuum orbit of the metaplectic group; adding displacements gives the Jacobi-group orbit $\Sbg$.
The desired proof of irreducibility on the copy spans follows from the vacuum being the unique vector at zero total photon number.
We find the formal degrees and square-integrability threshold $L$ from explicit integrals; these yield canonical cloners with sample complexity $O(n^2/\epsilon)$.
Optimality requires more care here: these groups are not amenable, so the construction is guaranteed to be optimal only among covariant channels.
For coherent states, the displacement group is amenable, and the canonical cloner is globally optimal, with sample complexity $\Theta(n/\epsilon)$.

For the centered states, we instead use a smaller, solvable group that still reaches every state and carries the same representation on all copy spans with at least $n$ copies.
Identifying the input and output spans gives a cloner with sample complexity $O(n/\epsilon^2)$.
A lower bound from the one-mode problem shows that its linear dependence on $n$ is optimal at fixed error, as stated in \cref{thm:gaussian_cloning_informal}; the sharp joint dependence on $n$ and $\epsilon$ remains open.
For the full Gaussian family, the fidelity optimum is also unknown, and the sample complexity remains between linear and quadratic bounds in $n$. (See \cref{sec:active-bosonic-cloning} for the precise guarantees and other details.)

\subsection{Stinespring dilation and complementary channels}

What does a unitary implementation of the cloner $T$ look like?
Every channel can be written as an isometry into the output and an environment, its \emph{Stinespring dilation}, followed by discarding the environment; discarding the output instead gives the \emph{complementary channel}.
In \cref{sec:complementary-channels-circuits}, we construct this isometry for $T$.
For compact resolving orbits, i.e.\ orbits of compact groups such as $\Sff$ and $\Sffq$, its environment is the complex conjugate of the $K$-copy span, the smallest possible.
For fermionic Gaussian and Slater states and for Bose condensates ($q$ bosons in a common unknown mode), symmetry alone fixes the isometry up to a phase.
Our proofs of the resolving property avoid Howe duality, but the circuit is built from it.
Howe duality organizes the multi-copy Hilbert space according to two commuting symmetries: the physical transformations applied identically to every copy, and transformations mixing the copies.
In this basis, each copy span is a single irreducible representation of the physical group, and the isometry becomes one Clebsch--Gordan coupling, as in the addition of angular momenta.
The coupling is sandwiched between the change to this basis on the input, the \emph{Howe transform}, and inverse Howe transforms on both outputs (\cref{subsec:howe-cg-circuit}).
This circuit extends Ref.~\cite{Brzic2026} for arbitrary pure states, where Schur--Weyl duality plays the role of Howe duality.
The circuit specifies the isometry exactly; an efficient decomposition of the Howe transforms and the coupling into elementary gates is left open.

What does the environment receive?
The complementary channel measures the POVM $\{d_N\psi^{\otimes N}d\psi\}_{\psi\in\cS}$ given by \cref{eq:resolution_informal} and prepares $K$ copies of the complex conjugate $\overline\psi$ of the outcome, i.e.\ of its transpose in a fixed product basis.
For compact resolving orbits, this channel is optimal among all channels for the task $\psi^{\otimes N}\mapsto\overline\psi^{\otimes K}$, with the same fidelity $d_N/d_{N+K}$ as the cloner (\cref{thm:structured-cloning-complementarity}).
This result extends the cloning--transposition complementarity found for unstructured pure states in Ref.~\cite{Brzic2026}.
Complex conjugation depends on the basis and is not a physical operation, but for fermionic Gaussian states a fixed unitary turns the environment output into $K$ optimal approximate copies of the particle--hole conjugate state (\cref{prop:particle-hole-decoder}).
For a $q$-particle Slater state, this results in the $(n-q)$-particle Slater state occupying the complementary single-particle subspace.

The same measurement links cloning to learning, a connection going back to Ref.~\cite{BrussEkertMacchiavello1998} (also cf.~\cite{chiribella2010quantum}).
Preparing $\psi^{\otimes K}$ instead of $\overline\psi^{\otimes K}$ gives the optimal way to learn from $N$ copies and prepare $K$ fresh ones, again with fidelity $d_N/d_{N+K}$ (\cref{cor:cloning-estimation}).
Adding $K$ copies is therefore exactly as hard as learning to prepare $K$ fresh ones: at a given fidelity, both need the same number of inputs.
A measure-and-prepare cloner, however, must also re-prepare the $N$ copies its measurement consumed, which caps its fidelity at $d_N/d_{2N+K}$; for pure qubit states and $K=1$, the number is $1/2$ for every $N$.
For the compact families of \cref{sec:dimension_law}, of real dimension $m$, this tends to $2^{-m/2}$ as $N\to\infty$ at $K=1$, whereas the infidelity of the optimal cloner decays as $m/(2N)$.
This gap quantifies the learning obstruction of \cref{sec:cloning_and_learning}.

\subsection{Why the manifold dimension controls canonical cloning}
\label{sec:dimension_law_overview}

The canonical cloners of \cref{thm:resolving_cloning_informal} display a pattern: their sample complexities are proportional to $n(n-1)$, $2q(n-q)$, $n(n+1)$, and $n(n+3)$ for fermionic Gaussian, Slater, centered bosonic Gaussian, and full bosonic Gaussian states, respectively.
Each quantity is the real dimension of the corresponding state manifold.
These scalings arise from ratios of formal degrees computed separately by representation theory, so it is apriori unclear why they should conspire to match the manifold dimension. To explain this, \cref{sec:dimension_law} gives a common explanation through a saddle-point calculation of the canonical fidelity $f_N=d_N/d_{N+1}$; we sketch the key idea here.
By \cref{eq:resolution_informal}, $d_N^{-1}=\int_{\cS}|\braket{\psi}{\psi_0}|^{2N}$ for a fixed reference state $\ket{\psi_0}\in\cS$, so this fidelity is an average over the family:
\begin{align}\label{eq:tilted_cartoon}
    f_N=\f{d_N}{d_{N+1}}=\Ex_{\psi\sim p_N}\!\lrb{\,\big|\braket{\psi}{\psi_0}\big|^2\,},
    \qquad p_N(\psi)\propto\big|\braket{\psi}{\psi_0}\big|^{2N}.
\end{align}
As $N$ grows, $p_N$ concentrates near the maximum of the overlap---the reference state itself---and this is where the dimension enters. Choose coordinates $\vec z\in\bR^m$ on the family around $\psi_0$, where $m=\dim\cS$, normalized so that
\begin{align}\label{eq:overlap_normal_coords}
    \big|\braket{\psi_{\vec z}}{\psi_0}\big|^{2}\;=\;1-|\vec z|^2\;+\;(\text{higher order}).
\end{align}
As an informal saddle-point approximation, let us now ignore the higher-order terms, the variation of the volume element, and everything far from $\psi_0$. Then $p_N(\vec z)\propto(1-|\vec z|^2)^N\simeq e^{-N|\vec z|^2}$ is a Gaussian of width $1/\sqrt N$ in each of the $m$ directions, and
\begin{align}\label{eq:equipartition_cartoon}
    1-f_N\;\simeq\;\Ex\big[|\vec z|^2\big]\;=\;\f{\dim\cS}{2N}\,:
\end{align}
each transverse direction of the family contributes $\tfrac1{2N}$ to the cloning infidelity, and nothing else about the family survives.
It follows that the canonical cloner reaches trace error $\epsilon$ with $N\simeq\dim\cS/(2\eps)$ copies, as desired.
This calculation is, of course, uncontrolled: it neglects (i) the anharmonicity of the overlap and the mass far from $\psi_0$, and (ii) the drift of the volume element in the coordinates $\vec z$, which carries the curvature of $\cS$ and is not universal.
In \cref{sec:dimension_law} we show that for a broad class of continuous families---including the Gaussian families $\Sff$, $\Sffq$ and $\Sbg$---all these deviations are rigorously controlled, and the saddle-point answer is exact to leading order: $1-f_N=\f{m}{2N}\big(1+O(m/N)\big)$, with explicit constants.

This dimension law determines the unrestricted complexity when the canonical cloner is optimal.
For centered bosonic Gaussian states, it instead describes a suboptimal construction: the manifold has dimension $n(n+1)$, while the optimal sample complexity is $\Theta_\epsilon(n)$ at fixed error.

For the compact families covered by \cref{cor:cloning-estimation}, the same ratio of formal degrees also gives the optimal fidelity of \emph{learning} a state from $N$ copies, so the same mechanism explains the dependence of their tomography complexity on $\dim\cS$.
Note that in the tomography reading one might hope to get this scaling from information counting alone---specifying a point of an $m$-dimensional manifold takes of order $m$ bits at fixed resolution---but by Holevo's theorem \cite{holevoBoundsQuantityInformation1971} each copy can supply up to $\log_2\dim\cH$ of these bits ($n$ bits for $n$ fermionic modes), so counting predicts $N\sim\dim\cS/\log\dim\cH$, for free fermions a factor $\sim n$ \emph{below} the true value.

\subsection{Other examples of resolving orbits}

The resolving-orbit construction also covers several simpler families, discussed in \cref{sec:other_resolving_orbits}.
For $D$ orthogonal states, the monomial group gives $d_N=D$ for every $N$, hence perfect cloning.
For products of $n$ unknown $D$-dimensional pure states, the optimal channel acts separately on the factors and the trace-distance sample complexity is $\Theta(n(D-1)/\epsilon)$.
For $n$-mode coherent states it is $\Theta(n/\epsilon)$, with the amplifier construction described in \cref{sec:active-bosonic-cloning}.
Finally, $q$-particle Bose condensates form the compact $U(n)$ orbit of $v^{\otimes q}$ and have degree $d_N=\binom{n+qN-1}{n-1}$.
For fixed $q\geq1$ and $n\geq2$, their sample complexity is $\Theta(n/\epsilon)$, compared with $\Theta(\binom{n+q-1}{q}/\epsilon)$ for arbitrary states in the same particle sector.

It is interesting to find further examples of resolving orbits.
What about other known families which have groups acting transitively on them?
$n$-qubit stabilizer states $\cS_{\mr{stab}}$ are an orbit of the Clifford group; is it resolving?
No, the $N$-copy span is not irreducible in general: this can be seen by computing the frame potential on $4$ copies \cite{grossSchurWeylDualityClifford2021,bittelCompleteTheoryClifford2026}.
Therefore, \cref{thm:resolving_cloning_informal} doesn't give an optimal cloner for this family;
the classical $\cS_{\mr{orth}}$ is the only discrete family known to us which is a resolving orbit.
Phase states provide a contrasting example: although they do not form a resolving orbit, their fidelity-optimal cloning channel can be determined exactly. We turn to this family next.

\subsection{Optimal cloner for phase states}
\label{sec:phase-results}

A basic family outside the resolving-orbit framework is given by the phase states
\begin{align}
    \ket{\psi_x}
    =
    \frac{1}{\sqrt d}
    \sum_{j=0}^{d-1}x_j\ket{j},
    \qquad
    x\in\C^d,
    \qquad
    |x_j|=1.
    \label{eq:phase-family-results}
\end{align}
These states form an orbit under diagonal phase rotations. However, because this group is abelian, its action on the tensor-power span decomposes into occupation-number sectors and is not irreducible. The phase states therefore do not form a resolving orbit, and the optimality theorem of \cref{thm:structured_optimal_cloning} does not apply. Nevertheless, we explicitly construct a cloner that maximizes the worst-case global fidelity with the ideal $N+K$-copy output among all quantum channels. Previous exact results cover arbitrary $N\to N+K$ qubit cloners and $1\to 1+K$ qutrit cloners~\cite{bruss2000phase, DArianoMacchiavello2003}. In arbitrary dimension, a known optimal construction applies when the number of added copies is $K=qd$ for an integer $q$~\cite{BuscemiDArianoMacchiavello2005}. This economical cloner acts particularly simply in the occupation basis: it adds $q$ to every occupation number. Our construction extends this idea to arbitrary $K$, distributing the added occupations so that their numbers differ by at most one.

For an occupation vector \(\mathbf b\in\N^d\) with \(|\mathbf b|=N\), let \(\ket{\Symb{N}{\mathbf b}}\) denote the normalized uniform superposition of all computational-basis strings with occupation \(\mathbf b\). Write the number of added copies as
\begin{align}
    K=qd+s,
    \qquad
    0\leq s<d,
\end{align}
and define the set of balanced occupations
\begin{align}
    \cC_{d,K}
    \coloneq
    \left\{
        \boldsymbol\Gamma\in\N^d:
        |\boldsymbol\Gamma|=K,\ 
        \max_j\Gamma_j-\min_j\Gamma_j\leq 1
    \right\}.
    \label{eq:balanced-occupations-results}
\end{align}
Thus every \(\boldsymbol\Gamma\in\cC_{d,K}\) has \(s\) entries equal to \(q+1\) and \(d-s\) entries equal to \(q\), and
\begin{align}
    |\cC_{d,K}|=\binom{d}{s}.
\end{align}

We construct the optimal cloner by combining these balanced shifts of the input occupation vector, as described in the following theorem.

\begin{theorem}[Optimal cloning of phase states, informal]
\label{thm:phase-results}
A fidelity-optimal \(N\to N+K\) cloner for the phase states is obtained by combining one Kraus operator \(A_{\boldsymbol\Gamma}\) for every balanced occupation \(\boldsymbol\Gamma\in\cC_{d,K}\). Its action on the symmetric occupation basis is
\begin{align}
    A_{\boldsymbol\Gamma}
    \ket{\Symb{N}{\mathbf b}}
    =
    \sqrt{
        \frac{
            \binom{N+K}{\mathbf b+\boldsymbol\Gamma}
        }{
            Z_{\mathbf b}
        }
    }
    \ket{\Symb{N+K}{\mathbf b+\boldsymbol\Gamma}},
    \qquad
    Z_{\mathbf b}
    =
    \sum_{\boldsymbol\Theta\in\cC_{d,K}}
    \binom{N+K}{\mathbf b+\boldsymbol\Theta}.
    \label{eq:phase-cloner-informal}
\end{align}
The channel
\begin{align}
    \Phi_{\mathrm{ph}}(X)
    =
    \sum_{\boldsymbol\Gamma\in\cC_{d,K}}
    A_{\boldsymbol\Gamma}X
    A_{\boldsymbol\Gamma}^{\dagger}
\end{align}
has the same fidelity with the ideal output $\ket{\psi_x}^{\otimes(N+K)}$ for every phase state and maximizes the worst-case fidelity among all quantum channels. Its fidelity is
\begin{align}
    F_{\mathrm{ph}}^{\star}(d,N,K)
    =
    \frac{|\cC_{d,K}|}{d^{2N+K}}
    \left(
        \sum_{|\mathbf b|=N}
        \frac{
            \sqrt{\binom{N}{\mathbf b}}
            \binom{N+K}{\mathbf b+\boldsymbol\Gamma_0}
        }{
            \sqrt{Z_{\mathbf b}}
        }
    \right)^2,
    \label{eq:phase-fidelity-informal}
\end{align}
where \(\boldsymbol\Gamma_0\) is any element of \(\cC_{d,K}\).
\end{theorem}

The form of the channel has a simple occupation-number interpretation. Each Kraus operator adds one of the most evenly distributed \(K\)-copy occupations to the input occupation. When \(K=qd\), there is a unique balanced occupation, \((q,\ldots,q)\), and the channel reduces to the economical cloner of~\cite{BuscemiDArianoMacchiavello2005}. When \(d\) does not divide \(K\), the optimal channel combines all coordinate permutations of the balanced occupation.

The proof uses phase covariance rather than irreducibility. Using phase symmetry, we can choose an optimal channel whose Kraus operators each change the occupation vector by a fixed amount. Each such operator sends an occupation state labeled by $\mathbf b$ to a multiple of the state labeled by $\mathbf b+\boldsymbol\Delta$, where $\boldsymbol\Delta$ is fixed for that operator.

To bound the fidelity of any channel, we compare different occupation shifts. If two entries of a shift differ by at least two, we move one unit from the larger entry to the smaller one. We show that verifying the bound for the resulting shift is sufficient to establish it for the original shift. Repeating this step reduces the proof to shifts whose entries differ by at most one, which are called \textit{balanced shifts}. Using phase covariance, we formulate the cloning problem as a semidefinite optimization problem (SDP). The channel in \cref{thm:phase-results} provides a primal feasible point supported on these balanced blocks, and we construct an explicit feasible dual point that has the same value, proving optimality.

The same channel also admits a moment representation in a form analogous to the Petz recovery map. In this form, the identity on the added systems that appears in the resolving-orbit cloner is replaced by the projector onto the balanced $K$-copy occupation subspace. We derive this representation, together with the formal statement and proof of \cref{thm:phase-results}, in \cref{sec:phase-main}.

The exact optimal-fidelity formula also allows us to determine how the number of input copies needed to produce one additional copy scales with the dimension and the target global infidelity.

\begin{corollary}[Fidelity sample complexity of phase states, informal]
\label{cor:phase-complexity-informal}
Let $N_{\min,\cS_{\mathrm{ph}}}^{F}(\delta)$ denote the minimum number of input copies required to produce one additional copy of an unknown phase state with global fidelity at least $1-\delta$ with the ideal $(N+1)$-copy output, for every input state. Uniformly for $d\geq 2$ and $0<\delta\leq 1/4$,
$$
N_{\min,\cS_{\mathrm{ph}}}^{F}(\delta)
=
\Theta\!\left(d\log\frac{1}{\delta}+\sqrt{\frac{d}{\delta}}\right).
$$
The optimal cloner described above achieves this bound for $K=1$.
\end{corollary}

The formal statement appears in \cref{cor:phase-one-copy-complexity}, with the proof in \cref{app:phase-one-copy-completion}. For trace-distance error $0<\epsilon\le1/4$ and dimension $d\ge\epsilon^{-2}$, this gives sample complexity $\Theta(d\log(1/\epsilon))$, improving on the $\Theta(d/\epsilon)$ complexity for unstructured states.

\subsection{Outlook}
\label{sec:outlook}

In this work, we considered a number of state families with useful algebraic properties that allowed us to establish bounds on cloning sample complexity, and in some cases find optimal cloning channels. However, a number of important structured classes remain uncharacterized.

Some state families form orbits of well-characterised groups, yet optimal cloning channels are not yet identified. One example is stabilizer states~\cite{bansalCloningHardLearning2026}, as discussed in \cref{sec:cloning_and_learning}. For phase states, the channels presented in this work have optimal fidelity and provide good bounds for complexity in terms of trace distance. However, this complexity in the high-precision regime is not exactly known, and cloners with optimal trace distance remain an open question. More generally, for orbit families which are not resolving, one may try a naive extension of our bounds from \cref{sec:orbits_cloners} (see \cref{app:bounds_nonresolving}) --- but we have not yet found a compelling application of this approach. And even among resolving orbits, the optimal cloners and cloning complexity of infinite-dimensional families, such as the bosonic Gaussian states, are only partially understood.

While product states are characterised by our resolving orbit framework (as explained in \cref{sec:other_resolving_orbits}), it is hard to say what happens under small deviations from a product form. One interesting question concerns the optimal cloning of states that can be generated by constant-depth circuits. This class of states does not admit a natural orbit structure, and the smallest set of states that forms a resolving orbit and contains all constant-depth circuit states may well be the set of all states. It is known that this class of states can be efficiently learned~\cite{landauLearningQuantumStates2025}; however, it is unclear how to exploit such a structure for cloning. Matrix product states are another fundamentally important target for an improved cloning protocol, which has also been studied extensively in the context of tomography \cite{cramer2010efficient}. 

Finally, natural classes of states for which optimal cloners could be studied are non-Gaussian bosonic and fermionic states. Unlike Gaussian states, non-Gaussian states do not form an orbit, which makes extending the techniques from this work difficult. The degree of non-Gaussianity is related to the hardness of learning different non-Gaussian states~\cite{mele2025efficient, mele2025learning, meleSymplecticRank2026}. Does this measure also affect the optimal cloning rate?

Going beyond cloning structured pure states, it is of interest to consider cloning of mixed states and unitaries (via higher-order operations~\cite{chiribella2008transforming}), similarly to \cite{chiribella2008optimal,sekatski2025cloning,d2005superbroadcasting,barnum1996noncommuting,jeon2026optimal,fanizza2026optimal}.

%% file: NotationsAndDefinitions.tex
\section{Basic notations and definitions}
\label{sec:basic_definitions}
We will consider pure states (complex unit vectors) in the separable Hilbert space $\cH$, possibly infinite-dimensional. Linear operators on $\cH$ will be denoted as $\BofH$, and in particular unitaries $\rU(\cH)$. The cone of density matrices on $\cH$ is denoted $\DofH$, and completely positive and trace preserving (CPTP) maps $\CPTP(\cH_1, \cH_2)$. 
For pure states $\ket{\psi}$, we will shorthand $\psi\defeq\ketbra{\psi}$. Families of states will always be sets of density matrices $\cS\subseteq \DofH$, even though most of the work concerns pure states -- this allows us to disregard the irrelevant global phases (which would add a confusion to our formalism otherwise).

Let us formally define the key notions of approximate cloning for structured quantum states. 
\begin{definition}[$N$-copy span]
\label{def:ncopy_span}
    For a family $\cS\subset \DofH$ of states on $\cH$, we define the \emph{$N$-copy span} $\cH^{(N)}_{\cS}=\Span(\{\ket{\psi}^{\otimes N}:\psi\in \cS\})\subset\cH^{\otimes N}$.
\end{definition}

Here and throughout the text, $\Span(A)$ denotes the \textit{closed} linear span of $A$, where $A$ is a set of subspaces of a Hilbert space; $\supp(\rho)$ means a subspace that supports $\rho$.

As the $N$ input copies are of a state which is promised to come from a given family $\cS$, for simplicity we consider channels which take input in the $N$-copy span $\cH^{(N)}_{\cS}$. (This can be extended to the full $\cH^{\otimes N}$ domain with `erasure' action on the complement of $\cH^{(N)}_{\cS}$.)

\begin{definition}[Approximate cloning]
\label{def:approximate_cloning}
    An $\epsilon$-approximate cloning channel (or cloner) from $N$ to $N+K$ copies on $\cS\subset \DofH$ is $T\in \CPTP\lrp{\cH^{(N)}_{\cS}, \cH^{\otimes N+K}}$, such that 
    \begin{align}
        \epsilon_{T,\cS}\defeq \sup_{\psi \in \cS} \f{1}{2} \|T(\psi^{\otimes N})-\psi^{\otimes N+K}\|_{1}\leq\epsilon,
    \end{align}
    where $\|\cdot\|_{1}$ is the trace norm.
    If a cloning channel $T_*$ is $\epsilon$-approximate for $\epsilon=\inf_{T} \epsilon_{T,\cS}$, we say $T_*$ is an \textit{optimal} cloner for $\cS$. Given $\cS$ and $\epsilon$, the \textit{cloning sample complexity} $N_{\mr{min},\cS}(\epsilon)$ is defined as the minimal $N$ for which there exists an $\epsilon$-approximate $N\rarr N+1$ cloner.
\end{definition}

Given a family $\cS$, we will aim to construct a good cloning channel $T$ (ideally, an optimal one). To be of interest, it needs to significantly outperform the optimal cloner for unstructured states (i.e., the set of all pure states $\cS_{\cH}$), as given by \cite{werner1998optimal}. Most crucially, we want to find families whose complexity $N_{\mr{min},\cS}(\epsilon)$ is dramatically improved compared to the unstructured case.

An alternative figure of merit for the output of a cloning channel is the fidelity rather than trace distance to $\psi^{\otimes N+K}$, and it implies the notions of optimality and complexity distinct from \cref{def:approximate_cloning}. Due to the operational meaning of trace distance, we will use it as a measure of cloning precision, unless said otherwise. One exception is phase states (\cref{sec:phase-main}), where we focus on constructing a cloning channel optimal w.r.t. fidelity.

%% file: CloningVsLearning.tex
\section{Cloning and learning for continuous families}
\label{sec:cloning_and_learning}

The relation between cloning and learning is subtle and requires careful examination. In particular, it differs significantly depending on whether the target state family $\cS$ is continuous or discrete. For continuous families, such as those considered in this work, even the optimal learning approaches cannot be used to construct $\epsilon$-approximate cloners.

Consider the extreme example of $\cS$ consisting of $d$ mutually orthogonal states $\ket{\psi_i}$, where perfect $1\rarr 2$ cloning proceeds by learning $i\in[d]$ through projective measurement of $\ketbra{\psi_i}{\psi_i}$ and then producing $\ket{\psi_i}^{\otimes2}$.
For this `maximally discrete' family of states, the optimal state learning also yields the optimal cloning strategy.
A powerful learn-clone strategy is also available for the less trivial family of stabilizer states on $n$ qubits. First, one uses the input copies to learn an exact classical stabilizer description and then prepares the requested $N+K$ copies. Learning can be done by Bell sampling, whose failure probability is at most $\epsilon$ for $N=O(n+\log(1/\epsilon))$ input copies \cite{montanaroLearningStabilizerStates2017}. Conditioning on success or failure shows that the resulting output has trace distance at most $\epsilon$ from the target, independently of $K$. A matching $\Omega(n)$ lower bound for sufficiently small constant error is now known \cite{bansalCloningHardLearning2026}, so the above procedure has asymptotically optimal constant-error sample complexity. However, the prefactors and subleading terms of this optimal complexity are not settled, and the optimal cloning channel for stabilizer states remains an open question.

The learn-clone strategy works for stabilizer states, because they form a discrete
family; this is what allows, up to a vanishing failure probability, to learn the state exactly. On a continuous family $\cS$ of pure states, the issue is that any state $\psi_1 \in \cS$ admits $\psi_2\in \cS$ that is close to $\psi_1$ and thus hard to distinguish from it. In particular, pick $\psi_2$ such that $\frac{1}{2}\|\psi_1-\psi_2\|_1=N^{-1/2}$. This yields a bound on trace distance $\frac{1}{2}\|\psi_1^{\otimes N}-\psi_2^{\otimes N}\|_1\leq \sqrt{3}/2<1$, hence discriminating $\psi_1^{\otimes N}$ and $\psi_2^{\otimes N}$ can only be done with a non-vanishing failure probability $\Omega(1)$. Because learning with accuracy $\eta<N^{-1/2}/2$ would allow discriminating between $\psi_1$ and $\psi_2$ (which are separated by $>2\eta$), it also needs to fail with a non-vanishing probability. Therefore, if one demands a learning procedure with an arbitrarily small chance of failure, one should expect accuracy no
better than $\eta=\Omega(N^{-1/2})$. Unlike for the discrete families, the exact learning ($\eta=0$) is impossible to achieve with finite resources.

This obstruction leads to a strikingly different performance of learn-clone
strategies, and in particular, prevents near-exact learn-cloning for any
continuous family $\cS$. To increase the number of copies even by one, the learn-clone strategy must
first re-prepare the copies consumed in the learning step, and the resulting product-state error accumulates irreperably. To see it concretely, suppose that after using $N$ input copies of $\psi\in \cS$, the
learning step successfully outputs $\td\psi$ such that
\begin{align}
    \f{1}{2}
    \left\|
        \td\psi
        -
        \psi
    \right\|_1
    =
    \eta .
\end{align}
If we then prepare $N+1$ copies of $\ket{\td\psi}$, the trace-distance error is
\begin{align}
    \f{1}{2}
    \left\|
        \td\psi^{\otimes (N+1)}
        -
        \psi^{\otimes (N+1)}
    \right\|_1
    =
    \sqrt{
        1-
        (1-\eta^2)^{N+1}
    } .
\end{align}
For this quantity to vanish as $N\rarr\infty$, one would need
$\eta=o(N^{-1/2})$. At the same time, as explained above, there is a lower bound $\eta=\Omega(N^{-1/2})$. We conclude that the learn-clone strategy for $\cS$
yields a worst-case cloning error bounded away from zero. To circumvent this, other strategies are required.

%% file: OrbitsAndCloners.tex
\section{Cloners for the resolving orbits}
\label{sec:orbits_cloners}

The formalism of this paper concerns state families $\cS$ which we call resolving orbits\footnote{This notion overlaps with Perelomov's generalized coherent states. Perelomov coherent states in the specific case of compact Lie groups and highest weight vectors as the base point of the orbit, such as those considered in \cite{chiribella2014optimal}, fall under our definition. Our setting is more general and allows for finite and non-reductive groups.}.
\begin{definition}
\label{def:resolving_orbit}
    For a separable Hilbert space $\cH$, we call a family of pure states $\cS$ in $\cH$ an \textit{$N$-resolving orbit} if there exists a group $G$ admitting a unimodular Haar measure and a strongly continuous unitary representation $\phi:G\rarr\rU(\cH)$ such that
    \begin{enumerate}
        \item $\cS$ is an orbit of $\phi$, i.e., $\cS = \{\phi(g)\, \psi_0 \,\phi^\dag(g) : g \in G\}$ for some pure $\psi_0 \in \DofH$, and
        \item For any integer $N'\geq N$, the representation $\phi^{\otimes N'}$ restricted to the $N'$-copy span $\cH^{(N')}_{\cS}$ is irreducible and square-integrable. \label{resolving_orbit_condition_2} 
    \end{enumerate}
    If $\cS$ is an $N$-resolving orbit already for $N=1$, we will say it is simply a \textit{resolving orbit}.
\end{definition}

\begin{remark}
\label{rmk:orbit_span}
The restriction of $\phi^{\otimes N}$ to the $N$-copy span $\cH^{(N)}_{\cS}$, referred to in condition~\ref{resolving_orbit_condition_2} of \cref{def:resolving_orbit}, exists because $\cH^{(N)}_{\cS}$ is $\phi^{\otimes N}$-invariant by construction. In particular, $\cH^{(N)}_{\cS}$ defines the subrepresentation of $\phi^{\otimes N}$ \emph{generated} by $\ket{\psi_0}^{\otimes N}$, i.e., the smallest closed invariant subspace containing $\ket{\psi_0}^{\otimes N}$.
\end{remark}

\begin{remark}
\label{rmk:resolving_orbit}
We use the term ``resolving'' to highlight the key operational consequence of \cref{def:resolving_orbit}. By Schur's Lemma, the irreducibility and square-integrability of $\phi^{\otimes N}$ over the group $G$ with unimodular Haar measure $d\mu$ guarantee a valid resolution of the identity on the subspace $\cH^{(N)}_{\cS}$: 
\begin{align}
\label{eq:orbit_resolution}
\int_G\psi^{\otimes N}_g \,d\mu(g)=\f{1}{d_{N}} \pi_{N},
\end{align} 
where pure $\ket{\psi_g}\defeq \phi (g)\ket{\psi_0}$, $\pi_{N}$ is the identity on (projector onto) $\cH^{(N)}_{\cS}$, and $d_{N}>0$ is a finite number known as the formal degree of the irreducible representation. Using \cref{eq:orbit_resolution}, one can compute it explicitly as
\begin{align}
\label{eq:finite_degree}
d_{N}=\lrp{\int_G \left|\langle \psi_g|\psi_0\rangle \right|^{2N} \,d\mu(g)}^{-1}.
\end{align} 
The convergence of the integral on the r.h.s. of \cref{eq:finite_degree} is in fact both necessary and sufficient for $\phi^{\otimes N}$ to be square-integrable on $\cH^{(N)}_{\cS}$ \cite{howe2012non}.
For finite-dimensional $\cH^{(N)}_{\cS}$ and a compact group $G$ with Haar measure normalized as $\int_Gd\mu(g)=1$, we have simply $d_N=\dim\,(\cH^{(N)}_{\cS})$.
\end{remark}

Resolving orbits form an extension of the formalism developed in \cite{werner1998optimal}. In particular, we have

\begin{theorem}
\label{thm:structured_optimal_cloning}
    Consider a separable Hilbert space $\cH$ and let the family of pure states $\cS\subset\DofH$ be an $N$-resolving orbit of a group $G$. Then there exists a cloning channel $T$ 
    from $N$ to $N+K$ 
    copies on the family $\cS$, given by
    \begin{align}
    \label{eq:optimal_cloner}
        T(\rho_{N})\coloneq\f{d_{N}}{d_{N+K}}\pi_{N+K} \lrp{\rho_{N}\otimes \bI_{K}} \pi_{N+K},
    \end{align}
    where $\rho_{N}\in\cD(\cH^{(N)}_{\cS})$, $\pi_{N+K}$ is a projector onto the $(N+K)$-copy span $\cH^{(N+K)}_{\cS}$, the identity operator $\bI_{K}$ acts on $\cH^{\otimes K}$, and $d_N$ is the formal degree as defined in \cref{rmk:resolving_orbit}.
    The channel $T$ is an $\epsilon$-approximate cloner for $\epsilon=1-f_{T,\cS}$, where $f_{T,\cS}$ is the guaranteed fidelity of $T(\psi^{\otimes N})$ to $\psi^{\otimes N+K}$, and equals
    \begin{align}
    f_{T,\cS}=\f{d_{N}}{d_{N+K}}.
    \end{align}
    If the group $G$ is amenable, in particular if $G$ is compact, the cloning channel $T$ is also \emph{optimal}. For any $G$, $T$ is optimal among $G$-covariant cloning channels.
\end{theorem}
Here a cloning channel $\Phi$ is $G$-covariant when
\begin{align}
\Phi\lrp{\phi(g)^{\otimes N}\rho (\phi(g)^\dag)^{\otimes N}}=\phi(g)^{\otimes(N+K)}\Phi(\rho)(\phi(g)^\dag)^{\otimes(N+K)}.
\end{align}
Here by $\phi(g)^{\otimes N}$ we mean its restriction to $\cH^{(N)}_{\cS}$, the domain of $\Phi$. Amenability of a group $G$ means that it admits an invariant mean, a linear functional $\theta$ on the bounded continuous functions on $G$ with Haar-measure-like properties. Namely, $\theta$ is positive ($f\geq0\Rightarrow\theta(f)\geq0$), unital ($\theta(1)=1$), and invariant under group shifts, $\theta\lrp{f(\cdot\,h)}=\theta\lrp{f(h\,\cdot)}=\theta(f)$ for all $h\in G$. A mean substitutes for the normalized Haar measure of a compact group, which does not exist when $G$ is non-compact; unlike a measure, $\theta$ is only finitely additive, but this turns out to be acceptable. Compact groups, such as those appearing in our finite-dimensional examples, are automatically amenable. Among the groups behind our infinite-dimensional examples, the Weyl--Heisenberg group (displacements) is amenable, while the metaplectic group (squeezing) is not.

The main inspiration of the proof of \cref{thm:structured_optimal_cloning} is to observe that the argument of \cite{werner1998optimal} mainly relies on $\rU(\cH)^{\otimes L}$ acting irreducibly on $\Sym^{L}(\cH)$ of $\cH^{\otimes L}$ for any $L$. This key idea generalizes to any resolving orbit. To account for the possible infinite-dimensionality of $\cH$ and non-compactness of $G$, we introduce some necessary modifications to the proof strategy of \cite{werner1998optimal}, leveraging the square-integrability and amenability assumptions. Note the subtle difference between optimality in trace distance $\epsilon$ and in fidelity $f_{T,\cS}$. We first show the latter and then demonstrate that the former is implied by the form of our cloning channel (which saturates the more general bound $\epsilon\geq 1-f_{T,\cS}$).

\begin{proof}

    Let $G$ and $\phi$ denote the group and its unitary representation, respectively, associated with the $N$-resolving orbit $\cS$ (\cref{def:resolving_orbit}).
 
 First, let us check that $T$ as defined in \cref{eq:optimal_cloner} is indeed a channel (CPTP map).
 For this, let us provide an explicit Kraus decomposition of $T$, $T(\rho)=\sum_j V_j \rho V_j^\dag$, where $j$ runs over the basis vectors $\ket{j}_{K}$ of $\cH^{\otimes K}$ and $V_j$ is a Kraus operator with output in $\cH^{\otimes N+K}$, acting on an arbitrary state $\ket{\psi}_N$ in $\cH^{(N)}_{\cS}$ as
 \begin{align}
 \label{eq:optimal_cloner_Kraus}
     V_j \,\ket{\psi}_N \defeq \sqrt{\f{d_{N}}{d_{N+K}}}\pi_{N+K}\ket{\psi}_N \otimes \ket{j}_{K}.
 \end{align}
 The form of $T(\rho)$ given in \cref{eq:optimal_cloner} follows explicitly from \cref{eq:optimal_cloner_Kraus}, implying complete positivity of $T$. For trace preservation, we need to show that $\sum_j V^\dag_j V_j=\pi_N$ (projector onto $\cH^{(N)}_{\cS}$, the domain of $T$). Using the resolution of the projector $\pi_{N+K}$ in terms of $(N+K)$ copies of a state $\psi_g\defeq \phi(g)\,\psi_0\,\phi(g)^\dag$ from the orbit $\cS$ (see \cref{rmk:resolving_orbit}), we can rewrite the expectation value of $\sum_j V^\dag_j V_j$ on a general state $\ket{u}\in\cH^{(N)}_{\cS}$ as
 \begin{align}
     \sum_j &\bra{u}V^\dag_j V_j\ket{u}=\f{d_N}{d_{N+K}}\sum_j \bra{u}\otimes \bra{j} \pi_{N+K}\ket{u}\otimes\ket{j}\\
     &=d_{N}\sum_j\int_G |\braket{u}{\psi_g^{\otimes N}}|^2 \cdot|\braket{j}{\psi_g^{\otimes K}}|^2 d\mu(g)\\
     &=d_{N}\int_G |\braket{u}{\psi_g^{\otimes N}}|^2  d\mu(g)=\bra{u}\pi_N\ket{u}.
 \end{align}
 Here we used that $\{\ket{j}\}$ is a complete basis. As $\ket{u}$ was arbitrary, it follows that $\sum_j V^\dag_j V_j=\pi_N$ by polarization identity.
In infinite dimension, this operator identity is understood in the strong topology.
For every positive trace-class $\rho$, the finite partial sums $\sum_{j\in F}V_j\rho V_j^\dag$ increase, their traces converge to $\Tr\rho$, and hence they converge in trace norm.
This defines a normal CPTP map, and the expression $\pi_{N+K}(\rho\otimes\bI_K)\pi_{N+K}$ in \cref{eq:optimal_cloner} is interpreted through that trace-norm-convergent Kraus series.

 Next, we find that the cloning fidelity of $T$ is $f_{T,\cS}=\f{d_{N}}{d_{N+K}}$ for any input $\kpsi\in\cS$.
 Observe that $\psi^{\otimes N+K} \pi_{N+K} \lrp{\psi^{\otimes N}\otimes \bI_{K}} \pi_{N+K}=\psi^{\otimes N+K}$ for any $\kpsi\in \cS$, because $\psi^{\otimes N+K}$ sits in the support of both projectors $\pi_{N+K}$ and $\psi^{\otimes N}\otimes \bI_{K}$.
 Thus $\Tr\lrp{\psi^{\otimes N+K}T(\psi^{\otimes N})}=\f{d_{N}}{d_{N+K}}$, implying the claimed fidelity.

In what follows, we will prove a fidelity bound $\f{d_N}{d_{N+K}}$ for covariant maps resembling adjoints (Heisenberg picture) of cloning channels, but not assumed to be normal. The bound directly implies that $T$ carries optimal fidelity among covariant cloning channels. To prove optimality among all channels for amenable $G$, we then apply an invariant mean over $G$ to the adjoint of a general channel, producing a map of exactly this type --- covariant but in general not normal. For simplicity, as in the earlier definition of covariance, by $\phi(g)^{\otimes N}$ we mean its restriction to $\cH^{(N)}_{\cS}$ where necessary. Finally, we will show that optimal fidelity for this channel also implies optimal trace distance, as desired.

Let $\Lambda:\cB(\cH^{\otimes N+K})\to\cB(\cH^{(N)}_\cS)$ be a unital, completely positive map that is covariant,
\begin{align}
\Lambda\lrp{\phi(g)^{\otimes N+K}\,X\,(\phi(g)^\dag)^{\otimes N+K}}=\phi(g)^{\otimes N}\,\Lambda(X)\,(\phi(g)^\dag)^{\otimes N},
\end{align}
but not necessarily normal. For $E\subset G$ of finite measure, consider $\int_E\psi_g^{\otimes N+K}\,d\mu(g)$: it is a norm-convergent Bochner integral, bounded by $\f{\pi_{N+K}}{d_{N+K}}$ because of \cref{eq:orbit_resolution} and the positivity of the integrand. Together with the covariance of $\Lambda$, this gives
\begin{align}
\int_E \phi(g)^{\otimes N}\,\Lambda\lrp{\psi_0^{\otimes N+K}}\,(\phi(g)^\dag)^{\otimes N}\,d\mu(g)=\Lambda\!\lrp{\int_E \psi_g^{\otimes N+K}\,d\mu(g)}\leq\frac{\Lambda(\pi_{N+K})}{d_{N+K}}\leq\frac{\pi_N}{d_{N+K}}.
\label{eq:finite-region-effect-bound}
\end{align}
On the other hand, for any finite-rank operator $B$ with $0\leq B\leq \Lambda\lrp{\psi_0^{\otimes N+K}}$, square-integrability gives $\int_G\phi(g)^{\otimes N}B\,(\phi(g)^\dag)^{\otimes N}\, d\mu(g)=\f{\Tr B}{d_N}\,\pi_N$. The integrand is weakly continuous, so this integral is the supremum of the integrals over compact $E\subset G$; comparing with \cref{eq:finite-region-effect-bound} yields $\Tr B\leq \f{d_N}{d_{N+K}}$. Taking $B=\Lambda\lrp{\psi_0^{\otimes N+K}}^{1/2}Q\,\Lambda\lrp{\psi_0^{\otimes N+K}}^{1/2}$ with finite-rank projections $Q\uparrow\pi_N$, we conclude
\begin{align}
\Tr\lrp{\psi_0^{\otimes N}\,\Lambda(\psi_0^{\otimes N+K})}\leq\Tr \Lambda\lrp{\psi_0^{\otimes N+K}}\leq \f{d_N}{d_{N+K}}.
\label{eq:covariant-effect-trace-bound}
\end{align}

Consider now a covariant cloning channel $\Phi$. Its adjoint $\Lambda=\Phi^*$ is unital, completely positive and covariant, and the left-hand side of \cref{eq:covariant-effect-trace-bound} equals the fidelity $\Tr\lrp{\psi_0^{\otimes N+K}\,\Phi(\psi_0^{\otimes N})}$ at the base point. As covariance makes the fidelity the same at every orbit point, $f_{\Phi,\cS}\leq\f{d_N}{d_{N+K}}$, proving the optimality of $T$ among covariant channels.

Let us now extend this bound to all cloning channels, assuming $G$ is amenable. Let $\Phi$ be any cloning channel with adjoint $\Lambda=\Phi^*$, and let $\theta$ be an invariant mean on the bounded continuous functions on $G$, which exists by amenability. For any fixed $X$ and $\ket u,\ket v\in\cH^{(N)}_\cS$, the function
\begin{align}
h\;\longmapsto\;\bra u\, (\phi(h)^\dag)^{\otimes N}\,\Lambda\lrp{\phi(h)^{\otimes N+K}\,X\,(\phi(h)^\dag)^{\otimes N+K}}\,\phi(h)^{\otimes N}\,\ket v
\end{align}
is bounded and continuous, because $\Lambda$ is normal (being the adjoint of a channel) and $\phi$ is strongly continuous. Applying $\theta$ to these functions defines a unital completely positive map $\overline\Lambda$, which is covariant by the invariance of $\theta$, but not necessarily normal. Its base-point value is the mean of the fidelities of $\Phi$ along the orbit,
\begin{align}
\Tr\lrp{\psi_0^{\otimes N}\,\overline\Lambda(\psi_0^{\otimes N+K})}=\theta_h\lrp{\Tr\lrp{\psi_h^{\otimes N+K}\,\Phi(\psi_h^{\otimes N})}}\geq f_{\Phi,\cS},
\end{align}
where the inequality holds because a mean, being positive and unital, dominates the infimum of its argument. Since \cref{eq:covariant-effect-trace-bound} applies to $\overline\Lambda$, we conclude $f_{\Phi,\cS}\leq\f{d_N}{d_{N+K}}$ for every cloning channel $\Phi$, so $T$ is optimal.

Finally, let us show that the optimality in fidelity for this setting also implies optimality in trace distance. The identity used in the fidelity computation above shows that $\psi^{\otimes N+K}$ is an eigenvector of $T(\psi^{\otimes N})$ with eigenvalue $\f{d_N}{d_{N+K}}$, so $T(\psi^{\otimes N})=\f{d_N}{d_{N+K}}\,\psi^{\otimes N+K}+R$ with $R\geq0$ supported orthogonally to $\psi^{\otimes N+K}$ and $\Tr R=1-\f{d_N}{d_{N+K}}$. The difference $T(\psi^{\otimes N})-\psi^{\otimes N+K}$ thus has positive part $R$ and negative part $\lrp{1-\f{d_N}{d_{N+K}}}\psi^{\otimes N+K}$, whence
\begin{align}
\frac12\left\|T(\psi^{\otimes N})-\psi^{\otimes N+K}\right\|_1=1-\frac{d_N}{d_{N+K}}.
\label{eq:canonical-cloner-exact-trace-error}
\end{align}
Conversely, for every state $\rho$ and pure target $P$, the two-outcome measurement $\{P,\bI-P\}$ gives
\begin{align}
\frac12\|\rho-P\|_1\geq1-\Tr(P\rho),
\label{eq:trace-distance-fidelity-lower-bound}
\end{align}
so the fidelity bound $\f{d_N}{d_{N+K}}$ forces a trace error of at least $1-\f{d_N}{d_{N+K}}$ on every cloning channel, which $T$ attains.

\end{proof}

%% file: FreeFermions.tex
\section{Optimal cloning of Fermionic Gaussian states}
\label{sec:optimal_fermionic_gaussian}

Fermionic particles on $n$ sites can be described using states on a $2^n$ dimensional Hilbert space $\Hf$ and the so-called Majorana fermion operators $c_\alpha\in\cB(\Hf)$, which satisfy the algebra $c_\alpha c_\beta+c_\beta c_\alpha=2\delta_{\alpha \beta}$. The standard basis in $\Hf$ is formed by the shared eigenstates of $n$ mutually commuting \textit{number operators} $n_{j}\defeq\f{1}{2}(1-ic_{2j-1} c_{2j} )$. Their eigenvalues $s_j\in\{0,1\}$ are interpreted as the presence or absence of a fermion particle at a site $j$. A basis state $\ket{0}$ without particles, i.e., the one for the trivial bitstring $(s_1,..s_n)=s=0$, is referred to as the vacuum.

The set of free-fermion states $\Sff$, also known as pure fermionic Gaussian states, is the set of ground states of the Hamiltonians of the following form:
\begin{equation}
\label{eq:H_free_ferm}
    H=\sum^{2n}_{\alpha ,\beta=1} ih_{\alpha \beta} c_\alpha c_\beta,
\end{equation}
where $h$ is a real antisymmetric matrix (and for this definition, we will require that the ground state of $H$ be non-degenerate). Such Hamiltonians are referred to as free-fermion due to the absence of genuine interactions between the particles (the latter would correspond to terms of weight higher than $2$).

As $H$ in \cref{eq:H_free_ferm} commutes with a parity operator $\Par=i^{n}c_1 c_2..c_{2n}$, free fermion states split into those with even and odd parity --- a notion that plays a key role when proving that $\Sff$ is a resolving orbit.
\begin{align}
\label{eq:parity_split}
    \Sff=\Sffp\cup\Sffm,~\Sffpm\defeq\left\{\ket{\psi}\in\Sff\,|\,\Par\ket{\psi}=\pm\kpsi\right\}
\end{align}

The proofs below rest on the representation theory of the groups $\Spin(2n)$ and $\Pin(2n)$, realized on the Fock space by the fermionic Gaussian unitaries; the necessary background is collected in \cref{app:ff_background}. We first treat the fixed-particle-number states as a warm-up, then the full family.

An important subfamily of free-fermion states are those with fixed particle number $q$, also known as Slater states,
\begin{align}
\label{eq:slater_def}
    \Sffq\defeq\left\{\kpsi\in\Sff \,:\, \sum^n_{j=1} n_{j}\kpsi=q\kpsi\right\}.
\end{align}
Our formalism applies to $\Sffq$ as follows.

\begin{proposition}
\label{prop:slater_orbit}
    The state family $\Sffq$ is a resolving orbit of the group $\rU(n)$, with $N$-copy span $\cH^{(N)}_\Sffq\subset\Hf^{\otimes N}$ of dimension
    \begin{align}
    \label{eq:slater_dimension}
        d_{\Sffq,N}=\prod_{1\leq i\leq q<j\leq n}\f{N+j-i}{j-i}.
    \end{align}
    Therefore, it admits an optimal cloning channel from $N$ to $N+K$ copies, as described in \cref{thm:structured_optimal_cloning}; it is $\epsilon$-approximate for $\epsilon=1-\ffqf$, where $\ffqf$ is
    \begin{align}
    \label{eq:slater_fidelity}
        \ffqf=\f{d_{\Sffq,N}}{d_{\Sffq,N+K}}=\prod_{1\leq i\leq q<j\leq n}\f{N+j-i}{N+K+j-i}.
    \end{align}
    Considering $K=1$, this certifies that the cloning complexity $N_{\mr{min}, \Sffq}(\epsilon)=\f{q(n-q)}{\epsilon}\lrp{1+o(1)}$, in particular,
    \begin{align}
    \label{eq:slater_N_bounds}
        \f{q(n-q)(1-\epsilon)}{\epsilon}-n \leq N_{\min, \Sffq} \leq \f{q(n-q)}{\epsilon}.
    \end{align}
\end{proposition}

\begin{proof}
    We use the background of \cref{app:slater}: under the identification $\cH_q\cong\Lambda^q\bC^n$, the number-conserving Gaussian unitaries realize the representation $u\mapsto u^{\wedge q}$ of $\rU(n)$; the family $\Sffq$ is a single orbit of this action with base point $\ket{\psi_0}=\ket{1\cdots10\cdots0}$, and it spans the irreducible representation $\Lambda^q\bC^n=\Span(\Sffq)$ (\cref{eq:slater_orbit,eq:slater_span}).

    The base point is a highest-weight vector of this representation, of weight
    \begin{align}
    \label{eq:slater_weight}
        \omega_q=(\underbrace{1,\dots,1}_{q},0,\dots,0),
    \end{align}
    the $q$-th fundamental weight of $\mathfrak{gl}_n$ \cite{FultonHarris}. To make this tangible: as Cartan subalgebra of $\mathfrak u(n)$ one may take the span of the number operators $n_j$, so that each Fock basis state $\ket s$ is a weight vector, of weight $s=(s_1,\dots,s_n)$; the raising operators are represented by the hopping terms $\{a^\dag_ia^{\phantom\dag}_j\}_{i<j}$ (moving a particle to a lower-indexed mode), and all of them annihilate $\ket{\psi_0}$, whose particles fill the $q$ lowest-indexed modes.

    Now fix $N\geq1$. By \cref{rmk:orbit_span}, the $N$-copy span $\cH^{(N)}_{\Sffq}$ is the subrepresentation of $u\mapsto(u^{\wedge q})^{\otimes N}$ generated by $\ket{\psi_0}^{\otimes N}$; and by the Cartan product theorem (\cref{thm:cartan_product} in \cref{app:cartan}, iterated), $\ket{\psi_0}^{\otimes N}$ is a highest-weight vector, of weight $N\omega_q$, that generates an \emph{irreducible} subrepresentation. (The theorem concerns the semisimple part $\mathfrak{su}(n)$ of $\mathfrak u(n)=\mathfrak{su}(n)\oplus\mathfrak u(1)$; since the central $\rU(1)$ acts on $(\Lambda^q\bC^n)^{\otimes N}$ by scalars, $\mathrm{SU}(n)$- and $\rU(n)$-invariant subspaces coincide.) Thus $\rU(n)$ acts irreducibly on $\cH^{(N)}_{\Sffq}$ for every $N\geq1$; since $\rU(n)$ is compact and $\cH^{(N)}_{\Sffq}$ is finite-dimensional, the square-integrability requirement of \cref{def:resolving_orbit} holds trivially, with formal degree $d_{\Sffq,N}=\dim\cH^{(N)}_{\Sffq}$ (\cref{rmk:resolving_orbit}), and $\Sffq$ is a resolving orbit.

    For the dimension, the Weyl dimension formula for $\mathfrak{sl}_n$ \cite{FultonHarris} gives for the irreducible representation of highest weight $\lambda=(\lambda_1\geq\dots\geq\lambda_n)$
    \begin{align}
    \label{eq:gl_weyl_dim}
        \dim V_\lambda=\prod_{1\leq i<j\leq n}\f{\lambda_i-\lambda_j+j-i}{j-i},
    \end{align}
    and substituting $\lambda=N\omega_q$ (so that $\lambda_i-\lambda_j=N$ for $i\leq q<j$ and $=0$ otherwise) yields \cref{eq:slater_dimension}. The claimed cloning channel and its fidelity, \cref{eq:slater_fidelity}, follow from \cref{thm:structured_optimal_cloning}.

    It remains to prove \cref{eq:slater_N_bounds}. Since the channel of \cref{thm:structured_optimal_cloning} is optimal, $N_{\min,\Sffq}$ is the smallest $N$ with $\epsilon_N\leq\epsilon$, where $\epsilon_N\defeq1-d_{\Sffq,N}/d_{\Sffq,N+1}$. Fixing $i$ in \cref{eq:slater_dimension} and taking the product over $j$, one finds
    \begin{align}
    \label{eq:slater_ratio}
        1-\epsilon_N=\prod_{i=1}^{q}\f{N+q-i+1}{N+n-i+1}
        =\prod_{t=1}^{q}\lrp{1-x_t},\qquad x_t\defeq\f{n-q}{N+n-q+t}\in[0,1].
    \end{align}
    For the upper bound, the inequality $\prod_t(1-x_t)\geq1-\sum_t x_t$, together with $x_t\leq x_1$, gives $\epsilon_N\leq q(n-q)/(N+1)$; hence $\epsilon_N\leq\epsilon$ already for $N=\lceil q(n-q)/\epsilon\rceil-1$, and so $N_{\min,\Sffq}\leq q(n-q)/\epsilon$. For the lower bound, let $N=N_{\min,\Sffq}$, so that $\epsilon_N\leq\epsilon$. The inequality $1-x\leq e^{-x}$ applied to each factor of \cref{eq:slater_ratio}, together with $x_t\geq x_q$, gives
    \begin{align}
        \ln\f{1}{1-\epsilon_N}\geq\sum_{t=1}^{q}x_t\geq\f{q(n-q)}{N+n}.
    \end{align}
    Since $\ln\f{1}{1-x}\leq\f{x}{1-x}$ for $x\in[0,1)$ and $\f{x}{1-x}$ is increasing, this implies $\f{\epsilon}{1-\epsilon}\geq\f{\epsilon_N}{1-\epsilon_N}\geq\f{q(n-q)}{N+n}$, which rearranges to the lower bound in \cref{eq:slater_N_bounds}. Finally, \cref{eq:slater_N_bounds} gives $N_{\min,\Sffq}=\f{q(n-q)}{\epsilon}\lrp{1+o(1)}$ as $\epsilon\to0$, uniformly in $n$ and $q$, since $n\leq2q(n-q)$.
\end{proof}

We now turn to the full family $\Sff$, which requires the larger group $\Pin(2n)$.

\begin{proposition}
\label{prop:free_fermion_+_orbit}
    The state family $\Sff$ is a resolving orbit of the group $\Pin(2n)$, with $N$-copy span $\cH^{(N)}_\Sff\subset\Hf^{\otimes N}$ of dimension
    \begin{align}
    \label{eq:ff_dimension}
        d_{\Sff,N}=2\prod_{1\leq i<j\leq n}\f{N+2n-i-j}{2n-i-j}.
    \end{align}
    Therefore, it admits an optimal cloning channel from $N$ to $N+K$ copies, as described in \cref{thm:structured_optimal_cloning}; it is $\epsilon$-approximate for $\epsilon=1-\fff$, where $\fff$ is
    \begin{align}
    \label{eq:ff_fidelity}
        \fff=\f{d_{\Sff,N}}{d_{\Sff,N+K}}=\prod_{1\leq i<j\leq n}\f{N+2n-i-j}{N+K+2n-i-j}.
    \end{align}
    Considering $K=1$, this certifies that the cloning complexity $N_{\min,\Sff}=\f{n(n-1)}{2\epsilon}\lrp{1+o(1)}$, in particular,
    \begin{align}
    \label{eq:ff_N_bounds}
        \f{n(n-1)(1-\epsilon)}{2\epsilon}-2(n-1) \leq N_{\min,\Sff} \leq \f{n(n-1)}{2\epsilon}.
    \end{align}
\end{proposition}

\begin{proof}
    We freely use the background of \cref{app:ff_background}: $\phi$ denotes the spinor representation on $\Hf$ of the group $\Pin(2n)=\Spin(2n)\cup\mu\,\Spin(2n)$ generated by the Gaussian unitaries $\Spin(2n)$ and the parity-flipping unitary $\mu=c_{2n-1}$, and $\Sff$ is a single orbit of $\Pin(2n)$.

    The even and odd subspaces are spanned by the Gaussian states of the respective parity, $\Span(\Sffpm)=\Hfpm$, and each fixed-parity family is a single orbit of $\Spin(2n)$ with base points $\chi_+=\ketbra{0}$ and $\chi_-=\ketbra{e}$, $\ket{e}=a_n^\dag\vac=\ket{0\cdots01}$ (Eqs.~\ref{eq:sector_orbits}--\ref{eq:sector_spans}). Accordingly, since $\Sff=\Sffp\cup\Sffm$, the $N$-copy span splits into the $N$-copy spans of the fixed-parity subfamilies,
    \begin{align}
    \label{eq:ncopy_split}
        \cH^{(N)}_{\Sff}=\cV^{+}_N\oplus\cV^{-}_N,\qquad \cV^{\pm}_N\defeq\cH^{(N)}_{\Sffpm},
    \end{align}
    an orthogonal sum since $\cV^{\pm}_N\subset(\Hfpm)^{\otimes N}$.

    As representations of $\Spin(2n)$, the two parity sectors $\Hfpm$ are the half-spinor representations \cite{FultonHarris}: irreducible representations with highest weights
    \begin{align}
    \label{eq:ff_highest_weights}
        \alpha=\lrp{\tfrac12,\dots,\tfrac12}\ \text{ for }\Hfp,\qquad \beta=\lrp{\tfrac12,\dots,\tfrac12,-\tfrac12}\ \text{ for }\Hfm,
    \end{align}
    whose highest-weight vectors are precisely the base states $\vac$ and $\ket e$ (see also \cite{helsen2022matchgate}).
    To make this tangible: as Cartan subalgebra of $\mathfrak{so}(2n)$ one may take the span of the operators $\f12-n_j$, so that each Fock basis state $\ket s$ is a weight vector, of weight $(\f12-s_1,\dots,\f12-s_n)$; with this choice the vacuum has weight $\alpha$ and $\ket e$ has weight $\beta$. The raising operators (root vectors of the standard positive roots $\{L_i-L_j,\,L_i+L_j\}_{i<j}$ of $\mathfrak{so}(2n)$) are represented by the hopping terms $\{a^\dag_i a^{\phantom\dag}_j\}_{i>j}$ (moving a particle to a higher-indexed mode) together with the pair annihilators $\{a_i a_j\}_{i<j}$, and one checks directly that all of them annihilate $\vac$ as well as $\ket{e}$, whose single particle occupies the highest-indexed mode. Note that the sign convention for the Cartan subalgebra is opposite to the one used for $\rU(n)$ in the proof of \cref{prop:slater_orbit}; this is what makes the vacuum a highest- rather than lowest-weight vector.

    Now fix $N\geq1$. Since $\Sffpm$ are orbits of $\Spin(2n)$ with the base points above, \cref{rmk:orbit_span} identifies $\cV^{\pm}_N$ with the subrepresentations of $\phi^{\otimes N}$ (viewed as a representation of $\Spin(2n)$) generated by $\vac^{\otimes N}$ and $\ket{e}^{\otimes N}$, respectively; and by the Cartan product theorem (\cref{thm:cartan_product}, iterated), these vectors are highest-weight vectors, of weights $N\alpha$ and $N\beta$, each generating an \emph{irreducible} subrepresentation. Since $N\alpha\neq N\beta$ for every $N\geq1$, the two representations are moreover inequivalent \cite{FultonHarris}. Their common dimension follows from the Weyl dimension formula for $\mathfrak{so}(2n)$ \cite{FultonHarris}: an irreducible representation of highest weight $\lambda=(\lambda_1\geq\dots\geq\lambda_{n-1}\geq|\lambda_n|)$ has dimension
    \begin{align}
    \label{eq:weyl_dim}
        \dim V_\lambda=\prod_{1\leq i<j\leq n}\f{\ell_i^2-\ell_j^2}{m_i^2-m_j^2},\qquad \ell_i=\lambda_i+n-i,\quad m_i=n-i.
    \end{align}
    Since $\ell_n=\lambda_n$ (the shift $n-i$ vanishes at $i=n$), \cref{eq:weyl_dim} depends on $\lambda_n$ only through $\lambda_n^2$, so $N\alpha$ and $N\beta$ give the same value \cite{FultonHarris}. Substituting $\lambda=N\alpha$, we have $\ell_i=\f N2+n-i$, so $\ell_i^2-\ell_j^2=(j-i)(N+2n-i-j)$ while $m_i^2-m_j^2=(j-i)(2n-i-j)$, and the factors $(j-i)$ cancel:
    \begin{align}
    \label{eq:sector_dim}
        \dim \cV^{+}_N=\dim \cV^{-}_N=\prod_{1\leq i<j\leq n}\f{N+2n-i-j}{2n-i-j}.
    \end{align}
    This proves \cref{eq:ff_dimension}.

    Let us now show that the full group $\Pin(2n)$ acts irreducibly on $\cH^{(N)}_{\Sff}$. Any $\Pin(2n)$-invariant subspace is in particular $\Spin(2n)$-invariant; and since $\cH^{(N)}_{\Sff}$ is an orthogonal sum, \cref{eq:ncopy_split}, of two \emph{inequivalent} $\Spin(2n)$-irreducibles, its only $\Spin(2n)$-invariant subspaces are $\{0\}$, $\cV^{+}_N$, $\cV^{-}_N$, and $\cH^{(N)}_{\Sff}$ itself (by Schur's lemma). But the parity-flipping unitary acts on $N$ copies as $\mu^{\otimes N}$, mapping $\psi^{\otimes N}\mapsto(\mu\psi)^{\otimes N}$ and hence exchanging $\cV^{+}_N\leftrightarrow\cV^{-}_N$ (as $\mu$ maps $\Sffp$ onto $\Sffm$, \cref{app:spin_pin}); so neither sector alone is $\Pin(2n)$-invariant, and $\phi^{\otimes N}$ restricted to $\cH^{(N)}_{\Sff}$ is irreducible. We stress that the inequivalence of the two sectors \emph{at every tensor power} is what makes this argument work: had they carried equivalent representations, $\cH^{(N)}_{\Sff}$ would contain additional `diagonal' invariant subspaces that would need to be ruled out.
    Finally, $\Pin(2n)$ is compact and $\cH^{(N)}_{\Sff}$ is finite-dimensional, so the square-integrability requirement of \cref{def:resolving_orbit} holds trivially, with formal degree $d_{\Sff,N}=\dim\cH^{(N)}_{\Sff}$ (\cref{rmk:resolving_orbit}). This establishes that $\Sff$ is a resolving orbit; the claimed cloning channel and its fidelity, \cref{eq:ff_fidelity}, follow from \cref{thm:structured_optimal_cloning}, the overall factors of $2$ canceling in the ratio $d_{\Sff,N}/d_{\Sff,N+K}$.

        It remains to prove \cref{eq:ff_N_bounds}. Since the channel of \cref{thm:structured_optimal_cloning} is optimal, $N_{\min,\Sff}$ is the smallest $N$ with $\epsilon_N\leq\epsilon$, where $\epsilon_N\defeq1-d_{\Sff,N}/d_{\Sff,N+1}$. Fixing $i$ in \cref{eq:ff_dimension} and taking the product over $j$, one finds
    \begin{align}
    \label{eq:ff_ratio}
        1-\epsilon_N=\prod_{i=1}^{n-1}\f{N+n-i}{N+2(n-i)}
        =\prod_{a=1}^{n-1}\lrp{1-x_a},\qquad x_a\defeq\f{a}{N+2a}\in[0,1].
    \end{align}
    For the upper bound, the inequality $\prod_a(1-x_a)\geq1-\sum_a x_a$, together with $x_a\leq\f{a}{N+2}$, gives $\epsilon_N\leq\f{n(n-1)}{2(N+2)}$; hence $\epsilon_N\leq\epsilon$ already for $N=\lceil n(n-1)/(2\epsilon)\rceil-2$, and so $N_{\min,\Sff}\leq n(n-1)/(2\epsilon)$. For the lower bound, let $N=N_{\min,\Sff}$, so that $\epsilon_N\leq\epsilon$. The inequality $1-x\leq e^{-x}$ applied to each factor of \cref{eq:ff_ratio}, together with $x_a\geq\f{a}{N+2n-2}$, gives
    \begin{align}
        \ln\f{1}{1-\epsilon_N}\geq\sum_{a=1}^{n-1}x_a\geq\f{n(n-1)}{2(N+2n-2)}.
    \end{align}
    Since $\ln\f{1}{1-x}\leq\f{x}{1-x}$ for $x\in[0,1)$ and $\f{x}{1-x}$ is increasing, this implies $\f{\epsilon}{1-\epsilon}\geq\f{\epsilon_N}{1-\epsilon_N}\geq\f{n(n-1)}{2(N+2n-2)}$, which rearranges to the lower bound in \cref{eq:ff_N_bounds}. Finally, \cref{eq:ff_N_bounds} gives $N_{\min,\Sff}=\f{n(n-1)}{2\epsilon}\lrp{1+o(1)}$ as $\epsilon\to0$, uniformly in $n$, since $2(n-1)\leq n(n-1)$.
\end{proof}

%% file: ActiveBosonicCloning.tex
\section{Cloning of Bosonic Gaussian states}
\label{sec:active-bosonic-cloning}

We study cloning pure Gaussian states of $n$ bosonic modes without an energy bound.
The cloning cost depends on which parameters are unknown: displacement, squeezing, or both.
On the Fock space $\cF_n$, these choices give the coherent family $\cS_{\mathrm{coh}}=\{\ketbra\alpha\}$, the centered family $\cS_{\mathrm{sq}}=\{\ketbra{\Omega_W}\}$, and the full pure Gaussian family $\Sbg=\{D(\alpha)\ketbra{\Omega_W}D(\alpha)^\dagger\}$.
Here $D(\alpha)$ displaces the state by $\alpha\in\bC^n$, while $W=W^{\mathsf T}\in\bC^{n\times n}$ with $\|W\|<1$ specifies the squeezing and mixing of the modes.
The respective real dimensions are $2n$, $n(n+1)$, and $n(n+3)$.
The bosonic notation and all proofs are collected in \cref{app:bosonic-cloning-proofs}.

We clone $N$ input copies into $M=N+K>N$ output copies and optimize over all channels, including those that couple the modes or use non-Gaussian operations.
Let $\cH_r^{\boldsymbol{\cdot}}$ be the $r$-copy span of the corresponding family and $\pi_r^{\boldsymbol{\cdot}}$ its projector.
The label $\boldsymbol{\cdot}$ is $\mathrm{coh}$ for coherent states, $\mathrm{sq}$ for squeezed vacua (centered pure Gaussian states), or $\mathrm{bg}$ for the full family of pure bosonic Gaussian states.
The guaranteed global fidelity and its unrestricted optimum are
\begin{equation}
f_{\cS}(\Phi)=\inf_{\psi\in\cS}\Tr\!\left[\psi^{\otimes M}\Phi(\psi^{\otimes N})\right],
\qquad f_{\cS}^\star(N,M)=\sup_\Phi f_{\cS}(\Phi).
\label{eq:ab-objectives}
\end{equation}
The optimal trace error $\epsilon_{\cS}^\star(N,M)$ and the sample complexity $N_{\mr{min},\cS}(\epsilon)$ follow \cref{def:approximate_cloning}; the latter always refers to adding one copy.
The distinction between the two error criteria matters here: for a pure target $Q$ and an output state $\sigma$,
\begin{equation}
1-\Tr(Q\sigma)\leq\frac12\|\sigma-Q\|_1\leq\sqrt{1-\Tr(Q\sigma)}.
\label{eq:ab-metric-comparison}
\end{equation}
Canonical and isometric cloners attain the left and right bounds, respectively.
We therefore distinguish fidelity optimality from trace-distance optimality throughout.

\Cref{tab:ab-orbits} summarizes how the three promises affect cloning.
For coherent states, the canonical construction is unrestrictedly optimal for both criteria.
For centered states, a smaller symmetry group gives an affine cloner that can improve on the canonical one, solves the one-mode fidelity problem for $1\leq K\leq4$, and achieves linear sample complexity in $n$ at fixed trace error.
Allowing both unknown squeezing and unknown displacement leaves a gap between linear and quadratic sample bounds.

\begin{table}[ht]
\centering\small
\long\def\abcell#1#2{\parbox[t]{#1\linewidth}{\raggedright#2\strut}}%
\setlength{\tabcolsep}{4pt}
\renewcommand{\arraystretch}{1.15}
\begin{tabular}{@{}lll@{}}
\hline
\abcell{0.24}{State family}&\abcell{0.37}{Optimal global fidelity $f^\star$}&\abcell{0.33}{Trace error; sample complexity}\\\noalign{\vskip0.6ex}
\hline
\abcell{0.24}{Coherent states}
&\abcell{0.37}{$(N/M)^n$ (\cref{prop:ab-coherent}).}
&\abcell{0.33}{$\epsilon^\star=1-f^\star$;\par $N_{\min}=\frac n\epsilon[1+O(\epsilon)]$.}\\\noalign{\vskip1ex}
\abcell{0.24}{One-mode squeezed vacua}
&\abcell{0.37}{Exact formula in \cref{eq:ab-one-mode-optimum}.}
&\abcell{0.33}{$N_{\min}=\Theta(1/\epsilon)$ (\cref{thm:ab-one-mode}).}\\\noalign{\vskip1ex}
\abcell{0.24}{Centered Gaussian states, $n\geq2$}
&\abcell{0.37}{$1-f^\star=\frac{nK^2}{8N}[1+O(n/N)]$,\par $1\leq K\leq4$, $N\geq2n$ (\cref{prop:ab-centered-uniform-optimum}).}
&\abcell{0.33}{$\Omega(n/\epsilon)\leq N_{\min}\leq O(\min\{n^2/\epsilon,n/\epsilon^2\})$.}\\\noalign{\vskip1ex}
\abcell{0.24}{All pure Gaussian states}
&\abcell{0.37}{$f^\star\leq(N/M)^n$,\par all $M>N\geq1$ (\cref{cor:ab-gaussian-sample-complexity}).}
&\abcell{0.33}{$\Omega(n/\epsilon)\leq N_{\min}\leq O(n^2/\epsilon)$.}\\\noalign{\vskip0.6ex}
\hline
\end{tabular}
\caption{Cloning $N\to M=N+K$ copies over all channels.
Sample complexity refers to $K=1$; the order bounds have absolute constants for $0<\epsilon\leq1/2$, and the coherent asymptotic is uniform as $\epsilon\to0$.
Centered states have $N_{\min}=\Theta_\epsilon(n)$ at every fixed $0<\epsilon<1$.
Exact trace optima remain open except for coherent states.}
\label{tab:ab-orbits}
\end{table}

We begin by recovering the known coherent-state cloner within the resolving-orbit framework.
The displacement group is amenable, so averaging preserves the fidelity guarantee and the optimality bound applies to arbitrary channels.

\begin{theorem}[Optimal cloning of coherent states]
\label{prop:ab-coherent}
Let $n\geq1$ and $M>N\geq1$ be integers.
The unrestricted optima for coherent states are
\begin{equation}
f_{\mathrm{coh}}^\star(N,M)=\left(\frac NM\right)^n,
\qquad \epsilon_{\mathrm{coh}}^\star(N,M)=1-\left(\frac NM\right)^n.
\label{eq:ab-coherent-optimum}
\end{equation}
Both values are attained by the canonical channel
\begin{equation}
\Phi_{\mathrm{coh}}(\rho)=\left(\frac NM\right)^n\pi_M^{\mathrm{coh}}\bigl(\rho\otimes\bI^{\otimes K}\bigr)\pi_M^{\mathrm{coh}}.
\label{eq:ab-coherent-channel}
\end{equation}
It concentrates the inputs into one collective copy, amplifies each mode with a quantum-limited amplifier of intensity gain $M/N$, and distributes the result over $M$ copies using a passive network and vacuum ancillas.
Its exact sample complexity is, for $0<\epsilon<1$,
\begin{equation}
N_{\mr{min},\cS_{\mathrm{coh}}}(\epsilon)
=\left\lceil\frac1{(1-\epsilon)^{-1/n}-1}\right\rceil
=\frac n\epsilon[1+O(\epsilon)]
\label{eq:ab-coherent-sample-complexity}
\end{equation}
as $\epsilon\to0$, uniformly in $n\geq1$.
In particular, $N_{\mr{min},\cS_{\mathrm{coh}}}(\epsilon)=\Theta(n/\epsilon)$ uniformly for $n\geq1$ and $0<\epsilon\leq1/2$, with absolute implicit constants.
The same guarantees hold for the displacement orbit of any fixed, known pure Gaussian state.
\end{theorem}

The amplifier realization and global-fidelity optimum are known~\cite{braunstein2001coherentCloning,cerf2005nonGaussianCloning,zhaoChiribella2017}.
The one-mode trace optimum follows from the zero-temperature limit of Ref.~\cite{gutaMatsumoto2006}, and the $1\to2$ projector expression appears in Ref.~\cite{dArianoMacchiavelloSacchi2001}.
Non-Gaussian operations can improve single-clone fidelity~\cite{cerf2005nonGaussianCloning}.

We next consider centered pure Gaussian states, obtained from the vacuum by Gaussian unitaries without displacement.
Up to an overall phase, these unitaries form the metaplectic group $\operatorname{Mp}(2n,\bR)$.
\Cref{thm:structured_optimal_cloning} guarantees optimality of the canonical cloner among channels covariant under this group, but nonamenability prevents the averaging argument used for coherent states from establishing unrestricted optimality.

To obtain bounds over non-covariant channels, we use the subgroup $S_n$ whose unitaries act on position wavefunctions as
\begin{equation}
(U_{(A,B)}\psi)(q)=(\det A)^{1/2}e^{iq^{\mathsf T}Bq/2}\psi(A^{\mathsf T}q),
\label{eq:ab-affine-schrodinger}
\end{equation}
where $q\in\bR^n$, $A$ is a real upper-triangular $n\times n$ matrix with positive diagonal, and $B$ is a real symmetric $n\times n$ matrix.
These operations combine linear changes of position coordinates with quadratic phases, and their action on the vacuum reaches every centered pure Gaussian state.
The group $S_n$ is solvable and hence amenable, allowing us to bound arbitrary cloners through $S_n$-covariant dilations (\cref{lem:ab-covariant-reduction}).

For $M>N\geq n$, the actions of $S_n$ on the input and output copy spans are unitarily equivalent (\cref{lem:ab-common-copy-model}).
Thus there is an isometry $V_{\mathrm{aff}}:\cH_N^{\mathrm{sq}}\to\cH_M^{\mathrm{sq}}$ satisfying $V_{\mathrm{aff}}U_g^{\otimes N}=U_g^{\otimes M}V_{\mathrm{aff}}$ for every $g\in S_n$.
We call the channel $\rho\mapsto V_{\mathrm{aff}}\rho V_{\mathrm{aff}}^\dagger$ the \emph{affine cloner}.
Its fidelity is constant on the family and equals $|\langle0|^{\otimes M}V_{\mathrm{aff}}|0\rangle^{\otimes N}|^2$.
For one mode, this construction is available for every $N\geq1$.
Allowing an $S_1$-covariant isometry into the outputs and an auxiliary environment, which is then traced out, gives the additional channels needed to attain the optimal fidelity for every $K$.

\begin{theorem}[Optimal global fidelity for one-mode squeezed vacua]
\label{thm:ab-one-mode}
For integers $M>N\geq1$ and $K=M-N$,
\begin{equation}
f_{\mathrm{sq},1}^\star(N,M)=
\begin{cases}
\displaystyle\frac{\Gamma((N+M)/4)^2}{\Gamma(N/2)\Gamma(M/2)},&1\leq K\leq4,\\[1ex]
\displaystyle\frac{N}{M-2},&K>4.
\end{cases}
\label{eq:ab-one-mode-optimum}
\end{equation}
The affine cloner attains the first branch, and the channel in \cref{eq:ab-negative-channel} attains the second.
The optimal trace error obeys
\begin{equation}
1-f_{\mathrm{sq},1}^\star
\leq\epsilon_{\mathrm{sq},1}^\star
\leq\min\left\{\sqrt{1-f_{\mathrm{sq},1}^\star},\frac{K}{M-2}\right\},
\label{eq:ab-one-mode-trace-bracket}
\end{equation}
where the second upper bound is available for $N\geq3$.
For $K=1$ and $N\geq3$, this gives $1/(8N+9)\leq\epsilon_{\mathrm{sq},1}^\star(N,N+1)\leq1/(N-1)$ and hence $N_{\mr{min},\cS_{\mathrm{sq},1}}(\epsilon)=\Theta(1/\epsilon)$ as $\epsilon\to0$.
\end{theorem}

At $K=4$, the two expressions in \cref{eq:ab-one-mode-optimum} agree and give $N/(N+2)$.
For one mode and $K>4$, the channel in \cref{eq:ab-negative-channel}, which uses an auxiliary environment that is traced out, achieves strictly higher global fidelity than the affine cloner.
The exact minimum trace-distance error $\epsilon_{\mathrm{sq},1}^\star(N,M)$ remains open; \cref{eq:ab-one-mode-trace-bracket} gives lower and upper bounds.

For several modes, we obtain an upper bound on the optimal global fidelity by restricting the input family to tensor products of independently chosen one-mode squeezed vacua.
The resulting bound, \cref{eq:ab-centered-upper}, applies to arbitrary channels, including those acting jointly on all modes.
For $1\leq K\leq4$, the affine cloner matches the corresponding lower bound on infidelity to leading order, $nK^2/(8N)$, as $n/N\to0$ (\cref{eq:ab-centered-optimal-asymptotic}).

\begin{theorem}[Cloning centered pure Gaussian states]
\label{prop:ab-affine}
\label{prop:ab-centered-fidelity-upper}
\label{prop:ab-centered-uniform-optimum}
\label{cor:ab-centered-sample-complexity}
Let $n\geq1$ and $M>N\geq1$ be integers.
When $N\geq n$, the affine cloner defined above has constant global fidelity and trace error
\begin{equation}
f_{\mathrm{aff},n}(N,M)=\frac{\Gamma_n((N+M)/4)^2}{\Gamma_n(N/2)\Gamma_n(M/2)},
\qquad\epsilon_{\mathrm{aff},n}=\sqrt{1-f_{\mathrm{aff},n}},
\label{eq:ab-affine-fidelity}
\end{equation}
where $\Gamma_n(s)=\pi^{n(n-1)/4}\prod_{j=1}^n\Gamma(s-(j-1)/2)$.
For every $M>N\geq1$, the unrestricted optimum satisfies
\begin{equation}
f_{\mathrm{sq},n}^\star(N,M)\leq[f_{\mathrm{sq},1}^\star(N,M)]^n.
\label{eq:ab-centered-upper}
\end{equation}
When $1\leq K\leq4$ and $N\geq2n$,
\begin{equation}
1-f_{\mathrm{sq},n}^\star(N,N+K)=\frac{nK^2}{8N}+R_\star,
\qquad |R_\star|\leq\frac{n^2K^2}{N^2}.
\label{eq:ab-centered-optimal-asymptotic}
\end{equation}
For every $0<\epsilon<1$, the sample complexity of adding one copy satisfies
\begin{equation}
\begin{aligned}
N_{\mr{min},\cS_{\mathrm{sq}}}(\epsilon)
&\geq\frac{n}{8[-\log(1-\epsilon)]}-1,\\
N_{\mr{min},\cS_{\mathrm{sq}}}(\epsilon)
&\leq\min\left\{\left\lceil2n-1+\frac{n(n+1)}{2\epsilon}\right\rceil,
\left\lceil\max\left\{2n,\frac{3n}{4\epsilon^2}\right\}\right\rceil\right\}.
\end{aligned}
\label{eq:ab-centered-sample-bounds}
\end{equation}
Thus $N_{\mr{min},\cS_{\mathrm{sq}}}(\epsilon)=\Theta_\epsilon(n)$ at each fixed $0<\epsilon<1$; uniformly for $0<\epsilon\leq1/2$, it is at least $n/(32\epsilon)$.
\end{theorem}

The two constructions favor different regimes.
At fixed trace error, the affine cloner needs only linearly many input copies in the number of modes, matching the lower bound; the canonical construction needs quadratically many.
The canonical guarantee has a better dependence on accuracy: halving a small allowed trace error roughly doubles its input requirement, whereas the affine bound grows by a factor of four.

The canonical cloner is optimal when we require covariance under every squeezing and mode-mixing transformation.
The affine cloner only needs to respect the smaller subgroup defined above.
This leaves more freedom to approximate the target, while transitivity still ensures the same fidelity for every centered state.
The resulting improvement shows that imposing the full symmetry can increase the cloning cost.
For squeezed vacua, the dimension law in \cref{sec:dimension_law} therefore describes the canonical construction, whose performance can be improved by relaxing this covariance requirement.
The quantitative comparisons are given in \cref{app:bosonic-cloning-proofs}.

Finally, suppose that both squeezing and displacement are unknown.
Cloning this full family is at least as hard as cloning coherent states, which gives a lower bound on the required number of copies.
For an upper bound, we use the canonical cloner for the Jacobi group, which combines displacements with Gaussian unitaries without displacement.
These bounds leave a gap between linear and quadratic growth in the number of modes at fixed trace error.

\begin{theorem}[Cloning all pure Gaussian states]
\label{cor:ab-gaussian-sample-complexity}
For integers $n\geq1$ and $M>N\geq1$, $f_{\mathrm{bg}}^\star(N,M)\leq(N/M)^n$.
When $N\geq2n+2$, the canonical channel \cref{eq:optimal_cloner} is optimal among Jacobi-covariant channels, with
\begin{align}
f_{\mathrm{can}}^{\mathrm{bg}}(N,M)
&=\left(\frac NM\right)^nf_{\mathrm{can}}^{\mathrm{sq}}(N-1,M-1)
\label{eq:ab-full-factorization}\\
&=\left(\frac NM\right)^n\prod_{1\leq i\leq j\leq n}\frac{N-1-i-j}{M-1-i-j},
\qquad\epsilon_{\mathrm{can}}^{\mathrm{bg}}=1-f_{\mathrm{can}}^{\mathrm{bg}}.
\label{eq:ab-full-canonical}
\end{align}
Here $f_{\mathrm{can}}^{\mathrm{sq}}$ is the centered canonical fidelity given in \cref{eq:ab-centered-canonical}.
For every $0<\epsilon<1$,
\begin{equation}
\frac{n}{-\log(1-\epsilon)}-\frac12
\leq N_{\mr{min},\Sbg}(\epsilon)
\leq\left\lceil\frac{n(n+3)}{2\epsilon}+2n\right\rceil.
\label{eq:ab-gaussian-sample-complexity}
\end{equation}
As $\epsilon\to0$, the canonical channel's sample complexity is $N_{\mathrm{can},\mathrm{bg}}(\epsilon)=n(n+3)[1+O(\epsilon)]/(2\epsilon)$, uniformly in $n$.
\end{theorem}

%% file: DimensionSaddlePoint.tex
\section{Dimension scaling of the cloning fidelity from saddle-point approximation}
\label{sec:dimension_law}
 
In \cref{sec:dimension_law_overview} we observed that for the fermionic and bosonic families the $N\rarr N+1$ cloning fidelity $f_N=d_N/d_{N+1}$ follows the pattern $1-f_N\simeq\dim\cS/(2N)$, and we sketched a saddle-point explanation. In this section, we make the explanation rigorous. We formulate geometric conditions on the family $\cS$, under which the pattern holds with explicit error bounds (\cref{thm:uniform_dimension_law_fN}), and we verify these conditions for the Gaussian families (\cref{sec:verifications}). The argument does not use the representation theoretic meaning of $d_N$ as the formal degree; instead, we only view $1/d_N$ as the $N$-th moment of the overlap with a reference state, \cref{eq:finite_degree}.

The logic of this section flows as follows. Given a family $\cS$ with a measure $\nu$ and a reference state $\psi_0\in\cS$, one chooses coordinates around $\psi_0$, computes the overlap with $\psi_0$ and the density of $\nu$ in these coordinates, and verifies that the assumptions \ref{assum:A0}--\ref{assum:A4} of \cref{sec:dimension_law_assumptions} hold with some constant $\kappa$. \Cref{thm:uniform_dimension_law_fN} then gives $1-f_N=\dim\cS/(2N)$ with an explicit error bound, and \cref{cor:uniform_dimension_law_complexity} translates this into the cloning sample complexity. For the families of this work, in \cref{sec:verifications} we establish the existence of coordinates in which the assumptions hold with $\kappa=8$.

\subsection{Setup and the Gaussian approximation}
\label{sec:dimension_law_setup}
 
Throughout this section, $\cS$ is a family of pure states which is a smooth manifold of real dimension $m$, equipped with a measure $\nu$. We fix a reference state $\psi_0\in\cS$ and denote
\begin{align}
    F(\psi)\defeq\Tr(\psi\,\psi_0)=|\braket{\psi}{\psi_0}|^2\in[0,1],
    \qquad h(\psi)\defeq-\log F(\psi)\in[0,\infty].
\end{align}
We consider the moments $I_t\defeq\int_\cS F^t\,d\nu(\psi)$, assumed to be finite for all $t\geq L$, and study the ratio of consecutive moments $f_N\defeq I_{N+1}/I_N$ for integer $N\geq L$. This is the object of interest for cloning: if $\cS$ is an $L$-resolving orbit of a group $G$ and $\nu$ is the invariant measure on $\cS$ induced by the Haar measure of $G$ (\cref{rmk:resolving_orbit}), then $I_N=1/d_N$ by \cref{eq:finite_degree}, so that $f_N=d_N/d_{N+1}$ is the fidelity of the $N\rarr N+1$ cloner of \cref{thm:structured_optimal_cloning}. However, the group structure does not play a role in the analysis below, which concerns the large-$N$ behaviour of $f_N$ for any $\cS$ and $\nu$ that meet the assumptions spelled out in \cref{sec:dimension_law_assumptions}.

Since $I_N-I_{N+1}=\int_\cS(1-F)F^N\,d\nu(\psi)$, the infidelity is an average over the family,
\begin{align}
\label{eq:tilted_expectation}
    1-f_N\;=\;\Ex_{p_N}\big[1-F\big]\;=\;\Ex_{p_N}\big[1-e^{-h}\big],
    \qquad
    p_N(\dd\psi)\defeq\f{F(\psi)^N\,d\nu(\psi)}{I_N}.
\end{align}
Note that \cref{eq:tilted_expectation} is a ratio: any constant factor in $\nu$ cancels.

\emph{The Gaussian approximation.}
Suppose for a moment that there are coordinates $\vec z\in\bR^m$ on $\cS$ around $\psi_0$ in which $h=|\vec z|^2$ exactly and $\nu=\dd^m\vec z$, and that the region far from $\psi_0$ can be ignored altogether --- this is the saddle-point picture of \cref{sec:dimension_law_overview}. Then $p_N$ is $\propto e^{-N|\vec z|^2}$, and a Gaussian integral gives
\begin{align}
\label{eq:gaussian_moment}
    \Ex_{p_N}[h]=\Ex_{p_N}\big[|\vec z|^2\big]=\f{m}{2N}\,.
\end{align}
Combined with $1-e^{-h}=h-O(h^2)$ and $\Ex_{p_N}[h^2]=O(m^2/N^2)$, \cref{eq:tilted_expectation} gives $1-f_N=\f{m}{2N}\big(1+O(m/N)\big)$.
 
The claim of this section is that for large $N$ the actual $p_N$ is effectively Gaussian, in the sense that the computation in \cref{eq:gaussian_moment} goes through with controlled errors. The proof of \cref{thm:uniform_dimension_law_fN} is the computation above, carried out for the actual $h$ and $\nu$. Three contributions were neglected: the anharmonicity of $h$, the variation of the volume density, and the region far from $\psi_0$. In order to control these, we use the assumptions below.

\subsection{Assumptions and the main statement}
\label{sec:dimension_law_assumptions}
 
We use a single regularity constant $\kappa\geq1$ in all assumptions (for the families of this work, $\kappa=8$ will do) and define the radius $R$ of the chart and the scale $u_0$ of the overlap gap as
\begin{align}
\label{eq:scales}
    R^2\defeq\f{1}{8\kappa},\qquad u_0\defeq\f{R^2}{4}=\f{1}{32\kappa}.
\end{align}
\begin{enumerate}[label=\textup{(A\arabic*)},start=0,itemsep=2pt,leftmargin=3em]
\item\label{assum:A0} \emph{(chart)} There is a neighborhood $U\subseteq\cS$ of $\psi_0$ and a chart $\vec z$ mapping $U$ bijectively onto the ball $B_R\defeq\{\vec z\in\bR^m:|\vec z|<R\}$, with $\vec z(\psi_0)=\vec 0$, in which $d\nu(\psi_{\vec z})=w(\vec z)\,\dd^m\vec z$ for a smooth positive density $w$, and $h$ is smooth. We write $r\defeq|\vec z|$.
\item\label{assum:A1} \emph{(anharmonicity)} On $B_R$,
\begin{align}\label{eq:A1}
    \big|h-r^2\big|\;\leq\;\kappa\,r^4
    \qquad\text{and}\qquad
    \big|\vec z\cdot\nabla h-2r^2\big|\;\leq\;\kappa\,r^4 .
\end{align}
\item\label{assum:A2} \emph{(density drift)} On $B_R$, $\big|\vec z\cdot\nabla\log w\big|\leq\kappa\,m\,r^2$.
\item\label{assum:A3} \emph{(separation)} $F\leq e^{-u_0}$ on $\cS\setminus U$.
\item\label{assum:A4} \emph{(no excess mass)} There is a real $L'\geq L$ with
\begin{align}\label{eq:A4}
    u_0L'\;\leq\;\kappa\,m
    \qquad\text{and}\qquad
    I_{L'}\;\leq\;w(\vec 0)\,\pi^{m/2}\lrp{\f{2u_0}{m}}^{m/2}e^{\,\kappa m}.
\end{align}
\end{enumerate}

Let us comment on the meaning of these conditions. \ref{assum:A1} says that in the chart the overlap is the Gaussian $e^{-r^2}$ at the leading order --- the coordinates are normalized as in \cref{sec:dimension_law_overview} --- and that the anharmonic corrections, both to $h$ and to its radial derivative $\vec z\cdot\nabla h=\f{\dd}{\dd t}h(t\vec z)\big|_{t=1}$, are of relative size at most $\kappa r^2$. \ref{assum:A2} controls the volume element, through the first radial derivative of $\log w$: integrating along the ray from $\vec 0$ to $\vec z$, it says that the density changes by a factor at most $e^{\pm\kappa mr^2/2}$ within the chart. Note that the rate is allowed to grow linearly with $m$; this allowance turns out to be harmless and is needed by our families (in \cref{sec:verifications}, $\log w$ is a multiple of $h$ with a coefficient $O(m)$).
The last two assumptions concern the region away from $\psi_0$. On the chart, \ref{assum:A1} and \cref{eq:scales} give $h\geq r^2(1-\kappa r^2)\geq\f78r^2$, so the overlap on the outer part $r^2\geq R^2/2$ of the chart is at most $e^{-7u_0/4}$; \ref{assum:A3} demands that the rest of the family is separated from $\psi_0$ at least as well, by an overlap gap $e^{-u_0}$. Thanks to this gap, $F^N\le e^{-u_0N}$ on the far region, so its contribution to $I_N$ decays exponentially in $N$ — provided its total weight is under control, which is the meaning of \ref{assum:A4}. The right-hand side of the second inequality in \cref{eq:A4} is, up to the factor $e^{\kappa m}$, the Gaussian integral $\int_{\bR^m}w(\vec 0)e^{-t|\vec z|^2}\dd^m\vec z$ at the exponent $t=m/(2u_0)$. This is the Gaussian with $\Ex[|\vec z|^2]=u_0$, i.e., the one concentrated inside the chart on the scale where the Gaussian approximation is used. \ref{assum:A4} thus asks that the whole family, weighted by $F^{L'}$, carries at most $e^{\kappa m}$ times the weight of this Gaussian. The reserve exponent $L'$ is needed for non-compact families, where $\nu(\cS)=\infty$: the tail estimate in the proof spends $N-L'$ powers of $F$ on the exponential suppression and keeps $L'$ of them to make the remaining integral converge. The first inequality in \cref{eq:A4} ensures that the retained powers cost at most a factor $e^{u_0L'}\leq e^{\kappa m}$.

\begin{theorem}[Dimension law]
\label{thm:uniform_dimension_law_fN}
Assume \ref{assum:A0}--\ref{assum:A4} with constant $\kappa$. Then, for every integer $N\geq N_0\defeq 4096\,\kappa^2 m$,
\begin{align}\label{eq:main_sandwich}
    \abs{\;\f{2N}{m}\lrp{1-f_N}\;-\;1\;}\;\leq\;\f{8\,\kappa m}{N}\,.
\end{align}
\end{theorem}
 
The numerical constants are loose and perhaps can be optimized. The order $m/N$ of the correction, however, cannot be improved, and its coefficient is not universal --- this is where the curvature of $\cS$ enters. For instance, for unstructured states in $\cH=\bC^d$ ($m=2(d-1)$) the exact fidelity $f_N=\f{N+1}{N+d}$ of Werner's cloner \cite{werner1998optimal} gives $\f{2N}{m}(1-f_N)-1=-\f{d}{N+d}\simeq-\f{m}{2N}$.

For cloning, the theorem translates into the following.
 
\begin{corollary}[Dimension law for the cloning sample complexity]
\label{cor:uniform_dimension_law_complexity}
Let $\cS$ be an $L$-resolving orbit whose invariant measure satisfies \ref{assum:A0}--\ref{assum:A4} with constant $\kappa$, and let $T_N$ be the $N\rarr N+1$ cloner of \cref{thm:structured_optimal_cloning}. There is $\epsilon_0=\Omega(\kappa^{-2})$ such that for all $\epsilon\leq\epsilon_0$, $T_N$ is an $\epsilon$-approximate cloner (\cref{def:approximate_cloning}) if and only if $N\geq N_\epsilon$, where
\begin{align}\label{eq:Nmin}
    N_\epsilon\;=\;\f{m}{2\epsilon}\lrp{1+O(\kappa\epsilon)}.
\end{align}
Consequently $N_{\mr{min},\cS}(\epsilon)\leq N_\epsilon$, with equality whenever $T_N$ is an optimal cloner --- in particular, whenever $\cH$ is finite-dimensional.
\end{corollary}

The proofs of \cref{thm:uniform_dimension_law_fN,cor:uniform_dimension_law_complexity} are collected in \cref{app:dimension-law-proofs}.
The verification of assumptions \ref{assum:A0}--\ref{assum:A4} for the state families is given in \cref{sec:verifications}.
There we construct determinantal charts for unstructured pure states, Slater states, fermionic Gaussian states, and all pure bosonic Gaussian states.
These checks allow the common choice $\kappa=8$, uniformly in the dimension, the number of modes, and the particle number.

%% file: ComplementaryChannelsAndCircuits.tex
\section{Complementary channels and circuit implementations}
\label{sec:complementary-channels-circuits}

Every channel is a marginal of an isometry, and the other marginal, the complementary channel, describes what the environment learns.
For the canonical cloner of a compact resolving orbit, both marginals are natural objects.
The first is the optimal cloner.
The second is a measure-and-prepare channel that outputs $K$ noisy copies of the complex-conjugate state, and it is the optimal channel for this conjugation task (\cref{thm:structured-cloning-complementarity}).
For universal pure states, this complementarity between cloning and transposition was recently established by Brzi\'c et al.~\cite{Brzic2026}; we extend it to every compact orbit with irreducible copy spans.

The complementary channel has two consequences.
First, since the environment only performs a measurement, the optimal cloning fidelity equals the optimal fidelity of learning the state well enough to prepare $K$ fresh copies.
The same identity bounds every learning-based cloner and makes the obstruction of \cref{sec:cloning_and_learning} quantitative (\cref{cor:cloning-estimation}).
Second, for fermions the conjugate state becomes the physical particle--hole conjugate (\cref{subsec:fermionic-environment}).
Finally, Howe duality turns the Stinespring isometry into a circuit made of Howe transforms and a single Clebsch--Gordan coupling, as in the Schur-transform circuit of Ref.~\cite{Brzic2026} (\cref{subsec:howe-cg-circuit}).
This is an exact representation-theoretic description, not an efficient gate decomposition.

\subsection{One isometry, two channels}
\label{subsec:structured-complement}

Let $G$ be a compact group with normalized Haar measure $\mu$, let $\phi:G\rarr\rU(\cH)$ be a finite-dimensional unitary representation, and let $\cS=\{\psi_g\}$ be the orbit of a unit vector $\ket{\psi_0}$, with $\ket{\psi_g}\defeq\phi(g)\ket{\psi_0}$.
Fix $N,K\geq1$ and $M=N+K$, and write $\cH_r\defeq\cH^{(r)}_{\cS}$, $\pi_r$ for its projector, and $d_r=\dim\cH_r$.
We assume that $\phi^{\otimes r}$ is irreducible on $\cH_r$ for $r\in\{N,M\}$, as it is for every $N$-resolving orbit; irreducibility of $\cH_K$ is not needed.
By \cref{eq:orbit_resolution}, with $d_r$ the dimension,
\begin{align}
\int_G\psi_g^{\otimes r}\,d\mu(g)=\frac{\pi_r}{d_r},
\qquad r\in\{N,M\}.
\label{eq:compact-orbit-tight-frame}
\end{align}

Regrouping the first $N$ and the last $K$ tensor factors defines an equivariant isometry
\begin{align}
J_{N,K}:\cH_M\hookrightarrow\cH_N\otimes\cH_K,
\qquad
J_{N,K}\ket{\psi_g}^{\otimes M}=\ket{\psi_g}^{\otimes N}\otimes\ket{\psi_g}^{\otimes K}.
\label{eq:cartan-regrouping-inclusion}
\end{align}
It is the restriction of the identification $\cH^{\otimes M}\cong\cH^{\otimes N}\otimes\cH^{\otimes K}$, and it lands in $\cH_N\otimes\cH_K$ because every vector $\ket{\psi_g}^{\otimes M}$ does.
Its range projector $P_{N,K}\defeq J_{N,K}J_{N,K}^\dagger$ is $\pi_M$ written in the regrouped tensor product.
The resolution identity at $r=M$, followed by a partial trace and the identity at $r=N$, gives
\begin{align}
P_{N,K}=d_M\int_G\psi_g^{\otimes N}\otimes\psi_g^{\otimes K}\,d\mu(g),
\qquad
\Tr_{\cH_K}P_{N,K}=\frac{d_M}{d_N}\,\bI_{\cH_N}.
\label{eq:cartan-projector-partial-trace}
\end{align}
Fix a product basis of $\cH^{\otimes K}$ and let $\mathcal K$ be complex conjugation in it; for $\ket v\in\cH^{\otimes K}$ write $\ket{\overline v}\defeq\mathcal K\ket v$, and for an operator $\overline X\defeq\mathcal KX\mathcal K$.
Choose an orthonormal basis $(e_\alpha)_{\alpha=1}^{d_K}$ of $\cH_K$.
Define
\begin{align}
A_\alpha&\defeq\sqrt{\frac{d_N}{d_M}}\,J_{N,K}^\dagger\lrp{\bI_{\cH_N}\otimes\ket{e_\alpha}},
\label{eq:structured-cloner-minimal-kraus}\\
\mathsf{St}_{N,K}\ket x&\defeq\sum_{\alpha=1}^{d_K}A_\alpha\ket x\otimes\ket{\overline{e_\alpha}}
\;\in\;\cH_M\otimes\overline{\cH_K}.
\label{eq:structured-cloner-stinespring}
\end{align}
The environment $\overline{\cH_K}=\mathcal K\cH_K$ is a subspace of the physical space $\cH^{\otimes K}$, and it carries the complex-conjugate representation $\overline{\phi^{\otimes K}}$.

\begin{theorem}[Cloning--conjugation complementarity]
\label{thm:structured-cloning-complementarity}
Under the assumptions above, $\mathsf{St}_{N,K}:\cH_N\to\cH_M\otimes\overline{\cH_K}$ is an equivariant isometry, $\mathsf{St}_{N,K}U^{(N)}_g=(U^{(M)}_g\otimes\overline{U^{(K)}_g})\mathsf{St}_{N,K}$ with $U^{(r)}_g\defeq\phi(g)^{\otimes r}|_{\cH_r}$, and it is a minimal Stinespring dilation of the canonical cloner.
Its two marginals are
\begin{align}
T(\rho)&=\Tr_{\overline{\cH_K}}\lrp{\mathsf{St}_{N,K}\rho\,\mathsf{St}_{N,K}^\dagger}
=\frac{d_N}{d_M}\,J_{N,K}^\dagger\lrp{\rho\otimes\bI_{\cH_K}}J_{N,K},
\label{eq:structured-stinespring-clone-branch}\\
\widetilde T_{N\rarr K}(\rho)&=\Tr_{\cH_M}\lrp{\mathsf{St}_{N,K}\rho\,\mathsf{St}_{N,K}^\dagger}
=d_N\int_G\Tr\lrp{\rho\,\psi_g^{\otimes N}}\,\overline{\psi_g}^{\otimes K}\,d\mu(g).
\label{eq:structured-stinespring-complement}
\end{align}
The first is the canonical cloner of \cref{eq:optimal_cloner}.
The second is entanglement breaking, and it is optimal among all channels from $\cB(\cH_N)$ to $\cB(\overline{\cH}^{\otimes K})$ for the conjugation task $\psi_g^{\otimes N}\mapsto\overline{\psi_g}^{\otimes K}$: every such channel has Haar-average and worst-case fidelity at most
\begin{align}
f^{\mathrm{conj}}_{\cS}(N,K)=\frac{d_N}{d_{N+K}},
\label{eq:structured-conjugation-fidelity}
\end{align}
and $\widetilde T_{N\rarr K}$ attains this value at every orbit point.
The minimal environment has dimension $d_K$.
\end{theorem}

The optimal conjugation fidelity equals the cloning fidelity $f_{T,\cS}$ of \cref{thm:structured_optimal_cloning}: cloning $K$ extra copies and producing $K$ conjugate copies are equally hard, and one isometry does both optimally.

\begin{proof}
\emph{Isometry and equivariance.}
By \cref{eq:cartan-projector-partial-trace}, $\sum_\alpha A_\alpha^\dagger A_\alpha=\frac{d_N}{d_M}\Tr_{\cH_K}P_{N,K}=\bI_{\cH_N}$.
Writing $\mathsf{St}_{N,K}\ket x=\sqrt{d_N/d_M}\,(J_{N,K}^\dagger\otimes\bI)\lrp{\ket x\otimes\sum_\alpha\ket{e_\alpha}\otimes\ket{\overline{e_\alpha}}}$, equivariance follows from that of $J_{N,K}$ and the invariance of $\sum_\alpha\ket{e_\alpha}\otimes\ket{\overline{e_\alpha}}$ under $U^{(K)}_g\otimes\overline{U^{(K)}_g}$.

\emph{Clone marginal.}
Tracing out the environment gives $\sum_\alpha A_\alpha\rho A_\alpha^\dagger$, which is \cref{eq:structured-stinespring-clone-branch}.
In the original tensor product this is $\frac{d_N}{d_M}\pi_M(\rho\otimes\bI_K)\pi_M$, because $\pi_M(\bI\otimes(\bI_K-\pi_K))=0$.

\emph{Environment marginal.}
Tracing out $\cH_M$ gives $\sum_{\alpha\beta}\Tr(A_\beta^\dagger A_\alpha\rho)\ketbra{\overline{e_\alpha}}{\overline{e_\beta}}$.
By \cref{eq:cartan-projector-partial-trace},
\begin{align}
A_\beta^\dagger A_\alpha
=\frac{d_N}{d_M}\lrp{\bI\otimes\bra{e_\beta}}P_{N,K}\lrp{\bI\otimes\ket{e_\alpha}}
=d_N\int_G\psi_g^{\otimes N}\,\bra{e_\beta}\psi_g^{\otimes K}\ket{e_\alpha}\,d\mu(g).
\end{align}
For Hermitian $X$ supported on $\cH_K$ one has $\sum_{\alpha\beta}\bra{e_\beta}X\ket{e_\alpha}\ketbra{\overline{e_\alpha}}{\overline{e_\beta}}=\mathcal KX\mathcal K=\overline X$, which gives \cref{eq:structured-stinespring-complement}.
The operators $d_N\psi_g^{\otimes N}d\mu(g)$ form a POVM by \cref{eq:compact-orbit-tight-frame} at $r=N$, so $\widetilde T_{N\rarr K}$ measures this POVM and prepares $\overline{\psi_g}^{\otimes K}$; such a channel is entanglement breaking~\cite{HorodeckiShorRuskai2003}.

\emph{Optimality.}
Let $\cR$ be any channel from $\cB(\cH_N)$ to $\cB(\overline{\cH}^{\otimes K})$, with Choi operator $C_\cR=\sum_{ij}\ketbra ij\otimes\cR(\ketbra ij)$ in a basis of $\cH_N$.
Then $\Tr(Y\cR(X))=\Tr(C_\cR(X^{\mathsf T}\otimes Y))$ and $\Tr C_\cR=d_N$.
Since $(\psi_g^{\otimes N})^{\mathsf T}=\overline{\psi_g}^{\otimes N}$, \cref{eq:cartan-projector-partial-trace} gives for the Haar-average fidelity
\begin{align}
\int_G\Tr\lrp{\overline{\psi_g}^{\otimes K}\cR(\psi_g^{\otimes N})}\,d\mu(g)
=\frac{1}{d_M}\Tr\lrp{C_\cR\,\overline{P_{N,K}}}
\leq\frac{\Tr C_\cR}{d_M}
=\frac{d_N}{d_M},
\label{eq:conjugation-figures-of-merit}
\end{align}
because $\overline{P_{N,K}}$ is a projection.
The worst-case fidelity is at most the average.
Conversely, at every orbit point, \cref{eq:finite_degree} gives
\begin{align}
\Tr\lrp{\overline{\psi_h}^{\otimes K}\,\widetilde T_{N\rarr K}(\psi_h^{\otimes N})}
=d_N\int_G|\langle\psi_g|\psi_h\rangle|^{2M}\,d\mu(g)
=\frac{d_N}{d_M}.
\end{align}
Since the average bound is attained pointwise, the average and worst-case optima coincide, and no twirling argument is needed.

\emph{Minimality.}
If $\sum_\alpha c_\alpha A_\alpha=0$, put $\ket v=\sum_\alpha c_\alpha\ket{e_\alpha}$, so that $J_{N,K}^\dagger(\bI\otimes\ket v)=0$.
Taking the matrix element between $\ket{\psi_g}^{\otimes M}$ and $\ket{\psi_g}^{\otimes N}$ gives $\langle\psi_g^{\otimes K}|v\rangle=0$ for all $g$.
These vectors span $\cH_K$, so $v=0$; the $d_K$ Kraus operators are linearly independent.
\end{proof}

Two features of \cref{eq:structured-stinespring-complement} are worth stressing.
It is not the transpose map: it changes the number of copies, it is noisy, and it transforms in the conjugate representation.
And it prepares all $K$ conjugate outputs from the same outcome $g$, so they are separable but correlated.
At $K=1$ its measurement is the covariant measurement whose outcome distribution $d_N|\langle\psi_g|\psi_h\rangle|^{2N}d\mu(g)$ is the tilted measure of \cref{sec:dimension_law}.
The same ratio $d_N/d_{N+1}$ therefore governs cloning, conjugation, and the dimension law.
Any other Stinespring dilation of $T$ has complement $u\,\widetilde T_{N\rarr K}(\cdot)\,u^\dagger$ for an isometry $u$ of the environment.

\emph{Cloning versus learning.}
The complementary channel is a learning procedure: it estimates $g$ and prepares states.
The theorem therefore compares cloning with learning exactly.

\begin{corollary}[Cloning and estimation]
\label{cor:cloning-estimation}
Under the assumptions of \cref{thm:structured-cloning-complementarity}:
\begin{enumerate}[label=\textup{(\alph*)}]
\item The largest worst-case fidelity with $\psi_g^{\otimes K}$ achievable by measuring $\psi_g^{\otimes N}$ and preparing a state from the outcome is $d_N/d_{N+K}$, the optimal cloning fidelity.
\item If $\cH_{2N+K}$ is also irreducible, every entanglement-breaking cloner from $N$ to $N+K$ copies has fidelity at most $d_N/d_{2N+K}$, and measuring $\{d_N\psi_g^{\otimes N}d\mu(g)\}$ and preparing $\psi_g^{\otimes(N+K)}$ attains this value.
\end{enumerate}
\end{corollary}

\begin{proof}
A measure-and-prepare channel $\rho\mapsto\int\Tr(E(dx)\rho)\,\sigma_x$ has the same pointwise fidelities as the channel $\rho\mapsto\int\Tr(E(dx)\rho)\,\overline{\sigma_x}$ has for the conjugation task, since $\Tr(\overline{\sigma}\,\overline{\psi}^{\otimes K})=\Tr(\sigma\psi^{\otimes K})$.
\Cref{thm:structured-cloning-complementarity} bounds the latter by $d_N/d_{N+K}$.
Measuring $\{d_N\psi_g^{\otimes N}d\mu(g)\}$ and preparing $\psi_g^{\otimes K}$ attains the bound, by the computation at the end of the proof of the theorem.
Part (b) is part (a) with $K$ replaced by $N+K$.
\end{proof}

Part (a) says that adding $K$ copies optimally is exactly as hard as learning the state well enough to prepare $K$ fresh copies.
Part (b) says that learning cannot also regenerate the $N$ copies it consumed.
For the families of \cref{sec:dimension_law} and large $N$, the power law in $K$ following \cref{sec:dimension_law} gives $d_N/d_{2N+K}\approx(N/(2N+K))^{m/2}\leq2^{-m/2}$, where $m$ is the real dimension of the family.
Every learning-based cloner thus has fidelity at most about $2^{-m/2}$, however many copies are available, whereas the fidelity of the optimal cloner approaches one as $1-mK/(2N)$.
This is the quantitative form of the obstruction in \cref{sec:cloning_and_learning}.
For universal qubits, for example, $d_N/d_{2N+1}=1/2$ for every $N$.
The connection between cloning and state estimation goes back to Ref.~\cite{BrussEkertMacchiavello1998}.
As $K\to\infty$ at fixed $N$, the ratio $d_{N+K}/d_{2N+K}$ of the two fidelities in \cref{cor:cloning-estimation} tends to one, consistent with the asymptotic equivalence of global cloning and estimation studied in Ref.~\cite{chiribella2014optimal}.

\begin{remark}[Noncompact orbits]
\label{rmk:noncompact-complement}
The construction of $\mathsf{St}_{N,K}$ and the formula \cref{eq:structured-stinespring-complement} use only the resolution identities at $N$ and $M$.
They therefore hold for the noncompact resolving orbits of \cref{sec:active-bosonic-cloning}, with formal degrees in place of dimensions and an infinite-dimensional environment; square-integrability at $K$ is not needed.
For coherent states, $\widetilde T_{N\rarr K}$ maps $\ket\alpha^{\otimes N}$ to a mixture of $\ket{\overline\beta}^{\otimes K}$ with $\beta-\alpha$ Gaussian of variance $1/N$ per mode.
After combining the $K$ outputs in a passive network, this is a thermal state of mean occupation $K/N$ per mode displaced to $\sqrt K\,\overline\alpha$, which is exactly the idler output of the quantum-limited amplifier in \cref{prop:ab-coherent}.
The optimality part of \cref{thm:structured-cloning-complementarity}, however, uses the normalized Haar measure, and we do not claim it for noncompact groups.
\end{remark}

\subsection{Fermions: conjugation becomes particle--hole conjugation}
\label{subsec:fermionic-environment}

The target $\overline\psi$ of the complementary channel depends on a choice of basis, and complex conjugation is not a physical operation.
For fermionic Gaussian states it becomes a physical target after a fixed unitary decoder.
We use the conventions of \cref{app:ff_background}: $a_j=\frac12(c_{2j-1}-ic_{2j})$, the occupation basis, and the Jordan--Wigner representation, in which $c_{2j-1}$ is real and $c_{2j}$ is imaginary.
Let $\mathcal K$ be complex conjugation in the occupation basis, which is also the product basis used for the environment in \cref{subsec:structured-complement}, and set
\begin{align}
C\defeq c_1c_3\cdots c_{2n-1}=\prod_{j=1}^n\lrp{a_j+a_j^\dagger},
\qquad
\Theta\defeq C\mathcal K,
\label{eq:fermion-particle-hole-antiunitary}
\end{align}
with the product ordered by increasing $j$; $C$ is unitary and $\Theta$ is antiunitary.

\begin{proposition}[Particle--hole decoder]
\label{prop:particle-hole-decoder}
For every Majorana operator $c_\alpha$ and every mode $j$,
\begin{align}
\Theta c_\alpha\Theta^{-1}=(-1)^{n-1}c_\alpha,
\qquad
\Theta a_j\Theta^{-1}=(-1)^{n-1}a_j^\dagger.
\label{eq:particle-hole-majorana-action}
\end{align}
Hence $\Theta$ commutes with every Gaussian unitary, and the particle--hole conjugate $\psi^{\mathrm h}\defeq\Theta\psi\Theta^{-1}=C\,\overline\psi\,C^\dagger$ of a pure Gaussian state $\psi$ is a pure Gaussian state with reversed occupations, $n_j\mapsto\bI-n_j$, and opposite covariance matrix, $\Gamma(\psi^{\mathrm h})=-\Gamma(\psi)$, where $\Gamma_{\alpha\beta}(\psi)=\frac i2\Tr(\psi[c_\alpha,c_\beta])$.
A $q$-particle Slater determinant occupying a subspace $W\subseteq\bC^n$ is mapped to the $(n-q)$-particle Slater determinant occupying $W^\perp$.
For $\cS\in\{\Sff,\Sffq\}$, the decoded environment $\rho\mapsto C^{\otimes K}\widetilde T_{N\rarr K}(\rho)C^{\dagger\otimes K}$ produces $K$ approximate copies of $\psi^{\mathrm h}$ with fidelity $d_{\cS,N}/d_{\cS,N+K}$ at every input, and no channel achieves a larger worst-case fidelity for the task $\psi^{\otimes N}\mapsto(\psi^{\mathrm h})^{\otimes K}$.
\end{proposition}

\begin{proof}
Moving $c_\alpha$ through the $n$ factors of $C$ gives $Cc_\alpha C^\dagger=(-1)^{n-1}c_\alpha$ for odd $\alpha$ and $(-1)^nc_\alpha$ for even $\alpha$, while $\mathcal Kc_\alpha\mathcal K$ equals $c_\alpha$ for odd $\alpha$ and $-c_\alpha$ for even $\alpha$.
Combining the two signs gives the first identity, and the second follows from $a_j=\frac12(c_{2j-1}-ic_{2j})$ and $\Theta i\Theta^{-1}=-i$.
A Gaussian unitary is $\exp(\sum_{\alpha\beta}A_{\alpha\beta}c_\alpha c_\beta)$ with $A$ real, and $\Theta$ leaves its exponent invariant.
Every pure Gaussian state is a Gaussian unitary applied to the vacuum or to $a_n^\dagger\vac$, and $\Theta$ maps these to occupation-basis states with all occupations reversed, which are Gaussian.
Passive Gaussian unitaries also commute with $\Theta$, which gives the statement about Slater determinants.
For the covariance matrix, $\Theta^{-1}[c_\alpha,c_\beta]\Theta=[c_\alpha,c_\beta]$ and antiunitarity give $\Tr(\psi^{\mathrm h}[c_\alpha,c_\beta])=\overline{\Tr(\psi[c_\alpha,c_\beta])}$, and the latter trace is imaginary.
Finally, $C^{\otimes K}$ maps $\overline\psi^{\otimes K}$ to $(\psi^{\mathrm h})^{\otimes K}$.
A channel that beats the decoded environment, followed by $C^{\dagger\otimes K}$, would beat \cref{thm:structured-cloning-complementarity}.
\end{proof}

The map $\psi\mapsto\psi^{\mathrm h}$ is antiunitary, so no physical process implements it exactly.
The environment of the optimal fermionic cloner implements it approximately, and optimally: from $N$ copies of a $q$-particle Slater determinant it produces $K$ copies of the complementary $(n-q)$-particle determinant, with the same fidelity as the clones.
The decoder is unitary because the environment already supplies the antiunitary factor $\mathcal K$.
Since $C$ is a product of $n$ Majorana operators, the target $\psi^{\mathrm h}$ has the parity of $\psi$ when $n$ is even and the opposite parity when $n$ is odd.

\subsection{Howe--Clebsch--Gordan circuit architecture}
\label{subsec:howe-cg-circuit}

For universal pure states, the isometry $\mathsf{St}_{N,K}$ is implemented in Ref.~\cite{Brzic2026} by a Schur transform, a multiplicity-free Clebsch--Gordan coupling, and two inverse Schur transforms.
Tracing out one output register gives the optimal cloner, and tracing out the other gives the optimal transposition channel.
The Howe dualities of \cref{sec:howe-background} provide the same architecture for structured orbits.
Two ingredients are needed: the copy span must be a single physical irreducible representation tensored with a fixed copy-side vector, and the relevant coupling must be multiplicity-free.
The first is \cref{lem:copy-side-line}; the second holds for highest-weight orbits.

\begin{lemma}[Multiplicity one for Cartan couplings]
\label{lem:cartan-multiplicity-one}
Let $G$ be a compact connected Lie group, and let $\ket{\psi_0}$ be a highest-weight vector of weight $\lambda$, so that $\cH_r\cong V_r\defeq V_{r\lambda}$ for every $r\geq1$ by \cref{thm:cartan_product}.
Then
\begin{align}
\dim\operatorname{Hom}_G\lrp{V_N,V_{N+K}\otimes\overline{V_K}}
=\dim\operatorname{Hom}_G\lrp{V_N\otimes V_K,V_{N+K}}
=1.
\end{align}
\end{lemma}

\begin{proof}
The first equality moves the conjugate factor across the intertwiner.
For the second, the multiplicity of $V_{(N+K)\lambda}$ in $V_{N\lambda}\otimes V_{K\lambda}$ is the dimension of the space of highest-weight vectors of weight $(N+K)\lambda$.
All weights of $V_{N\lambda}$ and $V_{K\lambda}$ lie below $N\lambda$ and $K\lambda$ in the dominance order, so the weight $(N+K)\lambda$ occurs only in the one-dimensional space spanned by the product of the two highest-weight vectors.
\end{proof}

For $r\in\{N,K,M\}$, \cref{lem:copy-side-line} writes the copy span as $\cH_r=V_r\otimes\bC w_r$, with a fixed copy-side unit vector $w_r$.
We call the corresponding unitary $U^{(r)}_{\mathrm{Howe}}:\cH_r\to V_r\otimes\bC w_r$ the \emph{Howe transform}.
By \cref{lem:cartan-multiplicity-one} there is, up to a phase, a unique $G$-equivariant isometry
\begin{align}
\mathrm{CG}_{N,K}:V_N\longrightarrow V_M\otimes\overline{V_K};
\label{eq:structured-inverse-cg}
\end{align}
it is the Clebsch--Gordan coupling $V_N\otimes V_K\to V_M$ with its $V_K$ leg bent, rescaled by $(d_N/d_M)^{1/2}$.
With the preparation isometry $\operatorname{Prep}(\ket x\otimes\ket{\overline y})\defeq(\ket x\otimes\ket{w_M})\otimes(\ket{\overline y}\otimes\ket{\overline{w_K}})$,
\begin{align}
\mathsf{St}_{N,K}
=\left[(U_{\mathrm{Howe}}^{(M)})^\dagger\otimes(\overline{U_{\mathrm{Howe}}^{(K)}})^\dagger\right]
\circ\operatorname{Prep}\circ\mathrm{CG}_{N,K}\circ\lrp{\bI_{V_N}\otimes\bra{w_N}}\circ U_{\mathrm{Howe}}^{(N)}
\label{eq:howe-cg-circuit-architecture}
\end{align}
up to a global phase, which affects neither marginal.
Indeed, after the Howe identifications both sides are equivariant isometries from $V_N$ into $V_M\otimes\overline{V_K}$, and the lemma leaves only a phase.
Read from right to left, the circuit applies the Howe transform to the input, removes the known copy-side vector, couples the physical irreducible representations, attaches the copy-side vectors of the two outputs, and applies the inverse Howe transforms.
If the input promise is not trusted, the projection onto the copied sector can be measured first, with a failure flag that extends the channel to the full input space.

To realize the above construction of the Stinespring dilation for the families of this paper, the ingredients of the construction are realized as follows.
\begin{enumerate}
\item \emph{Slater determinants.} The skew Howe duality $\rU(n)\times\rU(r)$ of \cref{eq:fermionic-skew-howe} gives $V_r=V_{(r^q)}$, and the copy-side line carries the character $\det^q$ of $\rU(r)$ (\cref{eq:slater-howe-selection}).
The coupling is the rectangular $\rU(n)$ Clebsch--Gordan map $V_{(N^q)}\otimes V_{(K^q)}\to V_{((N+K)^q)}$, and the particle--hole decoder of \cref{prop:particle-hole-decoder} turns the environment register into a physical $(n-q)$-particle register.
\item \emph{Gaussian states of fixed parity.} The spin--orthogonal duality $\Spin(2n)\times\rO(r)$ of \cref{eq:spin-orthogonal-howe} gives $V_r=V_{r\alpha}$ for even and $V_r=V_{r\beta}$ for odd parity, with $\alpha,\beta$ of \cref{eq:ff_highest_weights}.
The copy-side line carries the trivial character of $\rO(r)$ for even and the determinant character for odd parity.
\item \emph{Gaussian states of unknown parity.} Here $G=\Pin(2n)$ is disconnected and the lemma does not apply directly.
The copy span splits into the two parity blocks, $\cH_r\cong(V_{r\alpha}\otimes\bC\ket{\mathbf 1})\oplus(V_{r\beta}\otimes\bC\ket{\det})$.
A weight comparison shows that the coupling cannot mix the blocks, so it has one $\Spin(2n)$ component per block, and the parity-exchanging element of $\Pin(2n)$ fixes their relative phase.
Accordingly, the removal and preparation of copy-side vectors are controlled on the parity block.
Doing this coherently preserves the superposition of parities in the joint output, whereas either marginal loses it.
\item \emph{Bose condensates.} The passive duality $\rU(n)\times\rU(r)$ of \cref{eq:bosonic-symmetric-howe} applies to $Nq$, $Kq$ and $Mq$ bosons; the copy-side line is a weight line inside a copy-side irreducible representation, and $V_r=\Sym^{rq}(\bC^n)$.
\item \emph{Squeezed vacua.} The oscillator duality $\operatorname{Mp}(2n,\bR)\times\rO(r)$ of \cref{eq:metaplectic-orthogonal-howe} identifies $\cH^{\mathrm{sq}}_r$ with the scalar lowest-weight module $\Theta_{n,r}(\varnothing)$ and the copy-side line with the trivial $\rO(r)$ type (\cref{eq:squeezed-howe-selection}).
The coupling is between scalar lowest-weight modules of parameters $N/2$, $K/2$ and $M/2$.
By \cref{rmk:noncompact-complement}, the construction requires $N\geq2n+1$ but no condition on $K$.
The resulting channel is the canonical cloner of \cref{thm:ab-resolving}, which is optimal only among metaplectic-covariant channels.
A physical implementation also needs a photon-number cutoff.
\end{enumerate}

For the compact families, \cref{eq:howe-cg-circuit-architecture} is an exact specification of both optimal channels, and it exposes both in a single coherent process.
It is not an efficient gate decomposition.
Even for universal states the Clebsch--Gordan block of Ref.~\cite{Brzic2026} is a conceptual construction rather than an efficient circuit, and efficient skew-Howe, spin-Howe, $\Pin$ and metaplectic transforms, with their encodings and precision costs, are separate algorithmic problems.
If only the complementary channel is needed, \cref{eq:structured-stinespring-complement} gives a simpler route: measure the covariant POVM $\{d_N\psi_g^{\otimes N}d\mu(g)\}$ and prepare $K$ conjugate outputs.

%% file: PhaseStateCloner.tex
\section{Optimal phase-state cloning}
\label{sec:phase-main}

We next construct a phase-state cloner that maximizes both the average and the worst-case global fidelity. Here global fidelity is the overlap $\bra{\psi}^{\otimes(N+K)}\rho\ket{\psi}^{\otimes(N+K)}$ between the actual output $\rho$ and the ideal state of $N+K$ copies. Throughout this section, optimality refers to this fidelity criterion, which differs from the trace-distance criterion used in our general definition of an approximate cloner. Optimality for global fidelity does not necessarily imply optimality in trace distance. Consider the phase states
\begin{align}
    \ket{\psi_x} \coloneq \frac{1}{\sqrt d}\sum_{j=0}^{d-1}x_j\ket j, \qquad x\in\C^d, \qquad |x_j|=1. \label{eq:phase-family-main}
\end{align}
These states form an orbit under diagonal phase rotations. However, because this group is abelian, its action on the tensor-power span decomposes into occupation-number sectors and is not irreducible. The phase states therefore do not form a resolving orbit, and \cref{thm:structured_optimal_cloning} does not apply. 

Our construction starts from the occupation-number rule of the known optimal economical cloner~\cite{BuscemiDArianoMacchiavello2005}. When the number of added copies is $K=d$, this cloner adds one to every entry of the occupation vector. More generally, when $K=qd$, it adds $q$ to every entry. We extend this rule to the case where $d$ does not divide $K$ by distributing the added occupations as evenly as possible. Writing $K=qd+s$ with $0\leq s<d$, each balanced shift adds $q+1$ to $s$ entries and $q$ to the remaining entries. Our channel combines all such balanced shifts, with the occupation-dependent normalization specified below.

The construction is thus a direct extension of the equal-increment rule. To prove optimality among all quantum channels, we use phase covariance to reduce the optimization of global fidelity to a semidefinite program (cf.~\cite{KimChitambar2022,hettel2026semidefinite,Brzic2026}). The balanced channel provides a feasible primal point, and we construct a dual certificate with the same value, proving optimality.

Previous exact results cover arbitrary $N\to N+K$ qubit cloners and $1\to 1+K$ qutrit cloners~\cite{DArianoMacchiavello2003}, as well as the economical construction for $K=qd$ in arbitrary dimension~\cite{BuscemiDArianoMacchiavello2005}. The balanced construction and its optimality proof apply to arbitrary $d$, $N$, and $K$.

\subsection{Occupation basis and main result}
\label{subsec:phase-main-result}

For $\ell\geq0$, let
\begin{align}
    \cB_\ell \coloneq \{\mathbf b\in\N^d:|\mathbf b|=\ell\}, \qquad |\mathbf b| \coloneq \sum_{j=0}^{d-1}b_j.
\end{align}
For $\mathbf b\in\cB_\ell$, let $S_{\mathbf b}$ be the set of computational-basis words in which the symbol $j$ occurs exactly $b_j$ times, and define the normalized symmetric occupation state
\begin{align}
    \ket{\Symb{\ell}{\mathbf b}} \coloneq \frac{1}{\sqrt{\binom{\ell}{\mathbf b}}} \sum_{\ket{\phi}\in S_{\mathbf b}}\ket{\phi}, \qquad \binom{\ell}{\mathbf b} \coloneq \frac{\ell!}{\prod_{j=0}^{d-1}b_j!}. \label{eq:sym-state-main}
\end{align}
The phase states have the occupation-basis expansion
\begin{align}
    \ket{\psi_x}^{\otimes\ell} = \frac{1}{d^{\ell/2}} \sum_{\mathbf b\in\cB_\ell} \sqrt{\binom{\ell}{\mathbf b}}\, x^{\mathbf b} \ket{\Symb{\ell}{\mathbf b}}, \qquad x^{\mathbf b} \coloneq \prod_{j=0}^{d-1}x_j^{b_j}. \label{eq:phase-expansion-main}
\end{align}

Write the number of added copies as
\begin{align}
    K=qd+s, \qquad 0\leq s<d,
\end{align}
and define the set of balanced $K$-copy occupations by
\begin{align}
    \cC_{d,K} \coloneq \left\{ \boldsymbol\Gamma\in\N^d: |\boldsymbol\Gamma|=K,\ \max_j\Gamma_j-\min_j\Gamma_j\leq1 \right\}. \label{eq:balanced-occupations-main}
\end{align}
Every $\boldsymbol\Gamma\in\cC_{d,K}$ has $s$ entries equal to $q+1$ and $d-s$ entries equal to $q$, so
\begin{align}
    |\cC_{d,K}|=\binom{d}{s}.
\end{align}
For every input occupation $\mathbf b\in\cB_N$, define
\begin{align}
    Z_{\mathbf b} \coloneq \sum_{\boldsymbol\Theta\in\cC_{d,K}} \binom{N+K}{\mathbf b+\boldsymbol\Theta}. \label{eq:Z-main}
\end{align}

\begin{theorem}[Optimal phase-state cloner]
\label{thm:phase-main}
The channel from $\Sym^N(\C^d)$ to $\Sym^{N+K}(\C^d)$ defined by
\begin{align}
    \Phi_{\mathrm{ph}}(X) = \sum_{\boldsymbol\Gamma\in\cC_{d,K}} A_{\boldsymbol\Gamma}X A_{\boldsymbol\Gamma}^{\dagger}, \label{eq:phase-channel-main}
\end{align}
with Kraus operators
\begin{align}
    A_{\boldsymbol\Gamma} = \sum_{\mathbf b\in\cB_N} \sqrt{ \frac{ \binom{N+K}{\mathbf b+\boldsymbol\Gamma} }{ Z_{\mathbf b} } } \ket{\Symb{N+K}{\mathbf b+\boldsymbol\Gamma}} \bra{\Symb{N}{\mathbf b}}, \label{eq:phase-kraus-main}
\end{align}
maximizes the worst-case global fidelity for $N\to N+K$ cloning of the phase states in \cref{eq:phase-family-main}. Its global fidelity is independent of $x$ and equals
\begin{align}
    F_{\mathrm{ph}}^\star(d,N,K) = \frac{|\cC_{d,K}|}{d^{2N+K}} \left( \sum_{\mathbf b\in\cB_N} \frac{ \sqrt{\binom{N}{\mathbf b}}\, \binom{N+K}{\mathbf b+\boldsymbol\Gamma_0} }{ \sqrt{Z_{\mathbf b}} } \right)^2, \label{eq:phase-fidelity-main}
\end{align}
where $\boldsymbol\Gamma_0$ is any element of $\cC_{d,K}$.
\end{theorem}

When $K=qd$, the balanced shift is unique and \cref{thm:phase-main} recovers the economical cloner in \cite{BuscemiDArianoMacchiavello2005}. When $d$ does not divide $K$, there are $\binom{d}{s}$ equally balanced shifts and the optimal channel combines all of them.

\subsection{Covariance reduction and the block SDP}
\label{subsec:phase-sdp}

For a channel $\Phi$, define its orbit-averaged global fidelity by
\begin{align}
    F_{\mathrm{av}}(\Phi) \coloneq \int \bra{\psi_x}^{\otimes(N+K)} \Phi\!\left( \left(\ketbra{\psi_x}{\psi_x}\right)^{\otimes N} \right) \ket{\psi_x}^{\otimes(N+K)} d\mu(x), \label{eq:average-fidelity-main}
\end{align}
where $\mu$ denotes the probability measure obtained by choosing the phases uniformly on the unit circle. Let $U_z\ket j=z_j\ket j$, where $|z_j|=1$. Twirling $\Phi$ according to
\begin{align}
    \overline\Phi(X) \coloneq \int (U_z^\dagger)^{\otimes(N+K)} \Phi\!\left( U_z^{\otimes N}X (U_z^\dagger)^{\otimes N} \right) U_z^{\otimes(N+K)} d\mu(z) \label{eq:phase-twirl-main}
\end{align}
produces a phase-covariant CPTP map. Its fidelity is constant on the phase orbit and equals $F_{\mathrm{av}}(\Phi)$, which is at least the worst-case fidelity of $\Phi$. Hence an optimal worst-case cloner may be taken to be phase covariant.

We may also restrict the output to $\Sym^{N+K}(\C^d)$. Indeed, if $P_{\mathrm{sym}}$ denotes the projector onto this subspace and $\tau_{\mathrm{sym}}$ is any phase-invariant state supported on it, then
\begin{align}
    \mathcal C(\rho) \coloneq P_{\mathrm{sym}}\rho P_{\mathrm{sym}} + \Tr\!\left[ (\Id-P_{\mathrm{sym}})\rho \right]\tau_{\mathrm{sym}}
\end{align}
is CPTP and cannot decrease the overlap with any target $\ket{\psi_x}^{\otimes(N+K)}$.

Substituting the occupation expansion in \cref{eq:phase-expansion-main} into \cref{eq:average-fidelity-main} gives
\begin{align}
    F_{\mathrm{av}}(\Phi) ={}&
    \frac{1}{d^{2N+K}}
    \sum_{\substack{
        \mathbf b,\mathbf c\in\cB_N\\
        \mathbf a,\mathbf a'\in\cB_{N+K}
    }}
    \sqrt{
        \binom{N}{\mathbf b}
        \binom{N}{\mathbf c}
        \binom{N+K}{\mathbf a}
        \binom{N+K}{\mathbf a'}
    }
    \int x^{\mathbf b-\mathbf c-\mathbf a+\mathbf a'}d\mu(x)
    \nonumber\\
    &\times
    \bra{\Symb{N+K}{\mathbf a}}
    \Phi\!\left(
        \ket{\Symb{N}{\mathbf b}}
        \bra{\Symb{N}{\mathbf c}}
    \right)
    \ket{\Symb{N+K}{\mathbf a'}}.
    \label{eq:average-fidelity-expanded-main}
\end{align}
For $\mathbf q\in\mathbb Z^d$, write $z^{\mathbf q}\coloneq\prod_jz_j^{q_j}$. Character orthogonality gives
\begin{align}
    \int z^{\mathbf q}d\mu(z)
    =
    \begin{cases}
        1,&\mathbf q=\mathbf0,\\
        0,&\mathbf q\neq\mathbf0.
    \end{cases}
    \label{eq:character-orthogonality-main}
\end{align}
Consequently, the integral in \cref{eq:average-fidelity-expanded-main} vanishes unless
\begin{align}
    \mathbf a-\mathbf b=\mathbf a'-\mathbf c.
    \label{eq:charge-selection-main}
\end{align}
Write this common occupation shift as $\boldsymbol\Delta$. It is an integer vector with $\sum_j\Delta_j=K$; its individual entries may be negative. Define the set of admissible input occupations for this shift by
\begin{align}
    \mathcal D_{\boldsymbol\Delta}
    \coloneq
    \left\{
        \mathbf b\in\cB_N:
        \mathbf b+\boldsymbol\Delta\in\cB_{N+K}
    \right\}.
    \label{eq:shift-domain-main}
\end{align}
For each shift, define the matrix indexed by $\mathbf b,\mathbf c\in\mathcal D_{\boldsymbol\Delta}$ by
\begin{align}
    (J_{\boldsymbol\Delta})_{\mathbf b,\mathbf c}
    \coloneq
    \bra{\Symb{N+K}{\mathbf b+\boldsymbol\Delta}}
    \Phi\!\left(
        \ket{\Symb{N}{\mathbf b}}
        \bra{\Symb{N}{\mathbf c}}
    \right)
    \ket{\Symb{N+K}{\mathbf c+\boldsymbol\Delta}}.
    \label{eq:J-block-main}
\end{align}
Writing $\mathbf a=\mathbf b+\boldsymbol\Delta$ and $\mathbf a'=\mathbf c+\boldsymbol\Delta$ in \cref{eq:average-fidelity-expanded-main} therefore gives
\begin{align*}
    F_{\mathrm{av}}(\Phi)
    =
    \frac{1}{d^{2N+K}}
    \sum_{\boldsymbol\Delta}
    \sum_{\mathbf b,\mathbf c\in\mathcal D_{\boldsymbol\Delta}}
    \sqrt{
        \binom{N}{\mathbf b}
        \binom{N}{\mathbf c}
        \binom{N+K}{\mathbf b+\boldsymbol\Delta}
        \binom{N+K}{\mathbf c+\boldsymbol\Delta}
    }
    (J_{\boldsymbol\Delta})_{\mathbf b,\mathbf c}.
\end{align*}

We now express the optimization over quantum channels in terms of their Choi operators. The requirements of complete positivity and trace preservation then become constraints on the blocks $J_{\boldsymbol\Delta}$. The channel matrix elements in \cref{eq:average-fidelity-expanded-main} are the entries of the unnormalized Choi operator
\begin{align*}
    J_\Phi
    =
    \sum_{\mathbf b,\mathbf c\in\cB_N}
    \ket{\Symb{N}{\mathbf b}}
    \bra{\Symb{N}{\mathbf c}}
    \otimes
    \Phi\!\left(
        \ket{\Symb{N}{\mathbf b}}
        \bra{\Symb{N}{\mathbf c}}
    \right).
\end{align*}
As established above, we may take $\Phi$ to be phase covariant, so twirling leaves it unchanged. By \cref{eq:phase-twirl-main}, twirling multiplies each channel matrix element by the same phase integral as in \cref{eq:average-fidelity-expanded-main}. Thus the matrix elements themselves vanish unless \cref{eq:charge-selection-main} holds. The Choi operator therefore decomposes into blocks labeled by the common occupation shift $\boldsymbol\Delta$, and these blocks are precisely the matrices $J_{\boldsymbol\Delta}$.

Complete positivity is equivalent to positivity of the Choi operator. Since its blocks act on mutually orthogonal subspaces, this is equivalent to
\begin{align}
    J_{\boldsymbol\Delta}\succeq0
    \qquad
    \textnormal{for every }\boldsymbol\Delta.
    \label{eq:J-positive-main}
\end{align}
For trace preservation, taking the partial trace over the output gives
\begin{align*}
    \Tr_{\mathrm{out}}J_\Phi
    ={}&
    \sum_{\boldsymbol\Delta}
    \sum_{\mathbf b,\mathbf c\in\mathcal D_{\boldsymbol\Delta}}
    (J_{\boldsymbol\Delta})_{\mathbf b,\mathbf c}
    \ket{\Symb{N}{\mathbf b}}
    \bra{\Symb{N}{\mathbf c}}
    \braket{
        \Symb{N+K}{\mathbf c+\boldsymbol\Delta}
    }{
        \Symb{N+K}{\mathbf b+\boldsymbol\Delta}
    }
    \\
    ={}&
    \sum_{\mathbf b\in\cB_N}
    \left(
        \sum_{\boldsymbol\Delta:
            \mathbf b\in\mathcal D_{\boldsymbol\Delta}}
        (J_{\boldsymbol\Delta})_{\mathbf b,\mathbf b}
    \right)
    \ket{\Symb{N}{\mathbf b}}
    \bra{\Symb{N}{\mathbf b}}.
\end{align*}
The second equality uses orthonormality of the occupation basis. Requiring this partial trace to equal the identity on $\Sym^N(\C^d)$ is therefore equivalent to
\begin{align}
    \sum_{\boldsymbol\Delta:
        \mathbf b\in\mathcal D_{\boldsymbol\Delta}}
    (J_{\boldsymbol\Delta})_{\mathbf b,\mathbf b}
    =
    1
    \qquad
    (\mathbf b\in\cB_N).
    \label{eq:J-trace-main}
\end{align}
Conversely, any collection of matrices satisfying \cref{eq:J-positive-main,eq:J-trace-main} can be assembled into a block-diagonal Choi operator with these occupation-shift blocks. It is positive and has partial trace equal to the identity, so it defines a CPTP map. Its block structure ensures phase covariance. Thus these constraints describe exactly the channels over which we need to optimize.

\begin{proposition}[Phase-covariant primal and dual SDPs]
\label{prop:phase-sdp-main}
For every shift $\boldsymbol\Delta$, let $v_{\boldsymbol\Delta}\in\mathbb R^{\mathcal D_{\boldsymbol\Delta}}$ be the vector with components
\begin{align}
    v_{\boldsymbol\Delta}(\mathbf b) \coloneq \frac{1}{d^{(2N+K)/2}} \sqrt{ \binom{N}{\mathbf b} \binom{N+K}{\mathbf b+\boldsymbol\Delta} }. \label{eq:v-delta-main}
\end{align}
The optimal worst-case phase-state cloning fidelity is the value of the primal SDP
\begin{align}
    \underset{\{J_{\boldsymbol\Delta}\}}{\textnormal{maximize}} \quad& \sum_{\boldsymbol\Delta} v_{\boldsymbol\Delta}^{\dagger} J_{\boldsymbol\Delta} v_{\boldsymbol\Delta} \label{eq:phase-primal-main}\\
    \textnormal{subject to} \quad& J_{\boldsymbol\Delta}\succeq0 \qquad \textnormal{for every }\boldsymbol\Delta, \nonumber\\
    & \sum_{\boldsymbol\Delta: \mathbf b\in\mathcal D_{\boldsymbol\Delta}} (J_{\boldsymbol\Delta})_{\mathbf b,\mathbf b} =1 \qquad (\mathbf b\in\cB_N). \nonumber
\end{align}
Its dual is
\begin{align}
    \underset{\{y_{\mathbf b}\}}{\textnormal{minimize}} \quad& \sum_{\mathbf b\in\cB_N}y_{\mathbf b} \label{eq:phase-dual-main}\\
    \textnormal{subject to} \quad& \operatorname{diag}(y_{\mathbf b})_ {\mathbf b\in\mathcal D_{\boldsymbol\Delta}} \succeq v_{\boldsymbol\Delta} v_{\boldsymbol\Delta}^{\dagger} \qquad \textnormal{for every }\boldsymbol\Delta. \nonumber
\end{align}
\end{proposition}

\begin{proof}
With $v_{\boldsymbol\Delta}$ defined in \cref{eq:v-delta-main}, the preceding fidelity calculation reads
\begin{align*}
    F_{\mathrm{av}}(\Phi)
    =
    \sum_{\boldsymbol\Delta}
    v_{\boldsymbol\Delta}^{\dagger}
    J_{\boldsymbol\Delta}
    v_{\boldsymbol\Delta}.
\end{align*}
Together with \cref{eq:J-positive-main,eq:J-trace-main}, this proves the primal formulation. For the dual, introduce a multiplier $y_{\mathbf b}$ for each constraint in \cref{eq:J-trace-main}. The Lagrangian is
\begin{align}
    \sum_{\mathbf b\in\cB_N}y_{\mathbf b} + \sum_{\boldsymbol\Delta} \Tr\!\left[ J_{\boldsymbol\Delta} \left( v_{\boldsymbol\Delta} v_{\boldsymbol\Delta}^{\dagger} - \operatorname{diag}(y_{\mathbf b})_ {\mathbf b\in\mathcal D_{\boldsymbol\Delta}} \right) \right].
\end{align}
Its supremum over $J_{\boldsymbol\Delta}\succeq0$ is finite precisely when the matrix inequalities in \cref{eq:phase-dual-main} hold, which proves the dual formulation.
\end{proof}

For later use, observe that if every $y_{\mathbf b}>0$, then the rank-one matrix constraint for a fixed $\boldsymbol\Delta$ is equivalent to the scalar inequality
\begin{align}
    \sum_{\mathbf b\in\mathcal D_{\boldsymbol\Delta}} \frac{ |v_{\boldsymbol\Delta}(\mathbf b)|^2 }{ y_{\mathbf b} } \leq1. \label{eq:rank-one-dual-main}
\end{align}
Indeed, conjugating the matrix inequality by the inverse square root of its diagonal left-hand side reduces it to $\Id\succeq uu^\dagger$, where $u_{\mathbf b}=v_{\boldsymbol\Delta}(\mathbf b)/\sqrt{y_{\mathbf b}}$. Since $uu^\dagger$ is rank 1, it has the single nonzero eigenvalue $\|u\|_2^2$, the matrix inequality is equivalent to \cref{eq:rank-one-dual-main}.

\subsection{Balanced channel and its moment representation}
\label{subsec:phase-balanced-channel}

The phase moments are diagonal in the occupation basis:
\begin{align}
    D_\ell &\coloneq \int \left( \ketbra{\psi_x}{\psi_x} \right)^{\otimes\ell} d\mu(x) \nonumber\\
    &= \frac{1}{d^\ell} \sum_{\mathbf a\in\cB_\ell} \binom{\ell}{\mathbf a} \ketbra{ \Symb{\ell}{\mathbf a} }{ \Symb{\ell}{\mathbf a} }. \label{eq:phase-moment-main}
\end{align}
Define the projector onto balanced $K$-copy occupations by
\begin{align}
    P_{\mathrm{bal},K} \coloneq \sum_{\boldsymbol\Gamma\in\cC_{d,K}} \ketbra{ \Symb{K}{\boldsymbol\Gamma} }{ \Symb{K}{\boldsymbol\Gamma} } \label{eq:balanced-projector-main}
\end{align}
and the normalization operator
\begin{align}
    R_{N,K}^{\mathrm{bal}} &\coloneq \Tr_K\!\left[ (\Id^{\otimes N}\otimes P_{\mathrm{bal},K}) D_{N+K}^2 (\Id^{\otimes N}\otimes P_{\mathrm{bal},K}) \right] \nonumber\\
    &= \Tr_K\!\left[ (\Id^{\otimes N}\otimes P_{\mathrm{bal},K}) D_{N+K}^2 \right]. \label{eq:petz-normalization-main}
\end{align}
The second equality follows from $P_{\mathrm{bal},K}^2=P_{\mathrm{bal},K}$ and cyclicity on the subsystem being traced out. The same channel as in \cref{thm:phase-main} can then be written as the Petz-type transpose map
\begin{align}
    \Phi_{\mathrm{ph}}(X) ={}& D_{N+K} \Bigl( (R_{N,K}^{\mathrm{bal}})^{-1/2} X (R_{N,K}^{\mathrm{bal}})^{-1/2} \otimes P_{\mathrm{bal},K} \Bigr) D_{N+K}. \label{eq:petz-channel-main}
\end{align}
Here Petz-type refers to the transpose-channel structure associated with the reference operator $D_{N+K}^2$ and the restriction of the last $K$ systems to the balanced sector.

\begin{proposition}[Normalization and Kraus form]
\label{prop:phase-normalization-main}
For any fixed $\boldsymbol\Gamma_0\in\cC_{d,K}$,
\begin{align}
    R_{N,K}^{\mathrm{bal}} = \frac{ \binom{K}{\boldsymbol\Gamma_0} }{ d^{2(N+K)} } \sum_{\mathbf b\in\cB_N} \binom{N}{\mathbf b} Z_{\mathbf b} \ketbra{ \Symb{N}{\mathbf b} }{ \Symb{N}{\mathbf b} }. \label{eq:petz-R-diagonal-main}
\end{align}
Consequently, the map in \cref{eq:petz-channel-main} has the Kraus representation in \cref{eq:phase-channel-main,eq:phase-kraus-main} and is trace preserving on $\Sym^N(\C^d)$.
\end{proposition}

\begin{proof}
The standard splitting of a symmetric occupation state is
\begin{align}
    \ket{\Symb{N+K}{\mathbf a}} = \sum_{\substack{ \mathbf b\in\cB_N\\ \mathbf a-\mathbf b\in\cB_K}} \sqrt{ \frac{ \binom{N}{\mathbf b} \binom{K}{\mathbf a-\mathbf b} }{ \binom{N+K}{\mathbf a} } } \ket{\Symb{N}{\mathbf b}} \otimes \ket{\Symb{K}{\mathbf a-\mathbf b}}. \label{eq:symmetric-splitting-main}
\end{align}
Insert the diagonal form of $D_{N+K}^2$ from \cref{eq:phase-moment-main}, apply the balanced projector on the last $K$ systems, and use \cref{eq:symmetric-splitting-main}. The contribution indexed by $\mathbf a=\mathbf b+\boldsymbol\Gamma$ is
\begin{align}
    \frac{1}{d^{2(N+K)}} \binom{N+K}{\mathbf b+\boldsymbol\Gamma} \binom{N}{\mathbf b} \binom{K}{\boldsymbol\Gamma}.
\end{align}
All balanced $\boldsymbol\Gamma$ have the same multinomial coefficient $\binom{K}{\boldsymbol\Gamma_0}$. Summing over $\boldsymbol\Gamma\in\cC_{d,K}$ gives \cref{eq:petz-R-diagonal-main}.

Expanding $P_{\mathrm{bal},K}$ using \cref{eq:balanced-projector-main}, we can write \cref{eq:petz-channel-main} as
$$
\begin{aligned}
\Phi_{\mathrm{ph}}(X)
&= \sum_{\boldsymbol\Gamma\in\cC_{d,K}}
D_{N+K}\Bigl[
(R_{N,K}^{\mathrm{bal}})^{-1/2}X(R_{N,K}^{\mathrm{bal}})^{-1/2}
\otimes \ketbra{\Symb{K}{\boldsymbol\Gamma}}{\Symb{K}{\boldsymbol\Gamma}}
\Bigr]D_{N+K} \\
&= \sum_{\boldsymbol\Gamma\in\cC_{d,K}}
\Bigl[D_{N+K}\bigl(\Id^{\otimes N}\otimes\ket{\Symb{K}{\boldsymbol\Gamma}}\bigr)(R_{N,K}^{\mathrm{bal}})^{-1/2}\Bigr]X \\
&\qquad {}\times
\Bigl[(R_{N,K}^{\mathrm{bal}})^{-1/2}\bigl(\Id^{\otimes N}\otimes\bra{\Symb{K}{\boldsymbol\Gamma}}\bigr)D_{N+K}\Bigr] \\
&= \sum_{\boldsymbol\Gamma\in\cC_{d,K}}
A_{\boldsymbol\Gamma}X A_{\boldsymbol\Gamma}^{\dagger},
\end{aligned}
$$
where we used the self-adjointness of $D_{N+K}$ and $(R_{N,K}^{\mathrm{bal}})^{-1/2}$. Thus there is one Kraus operator for each balanced shift:
$$
A_{\boldsymbol\Gamma}
= D_{N+K}
\left(\Id^{\otimes N}\otimes\ket{\Symb{K}{\boldsymbol\Gamma}}\right)
(R_{N,K}^{\mathrm{bal}})^{-1/2}.
$$

Using \cref{eq:phase-moment-main,eq:symmetric-splitting-main}, we obtain
\begin{align}
    D_{N+K}\left(\ket{\Symb{N}{\mathbf b}}\otimes\ket{\Symb{K}{\boldsymbol\Gamma}}\right)
    = \frac{\sqrt{\binom{N}{\mathbf b}\binom{K}{\boldsymbol\Gamma}\binom{N+K}{\mathbf b+\boldsymbol\Gamma}}}{d^{N+K}}
    \ket{\Symb{N+K}{\mathbf b+\boldsymbol\Gamma}}.
\end{align}
Combining this with the diagonal normalization in \cref{eq:petz-R-diagonal-main} gives
\begin{align}
    A_{\boldsymbol\Gamma}\ket{\Symb{N}{\mathbf b}}
    &= \sqrt{\frac{\binom{N}{\mathbf b}\binom{K}{\boldsymbol\Gamma}\binom{N+K}{\mathbf b+\boldsymbol\Gamma}}{\binom{K}{\boldsymbol\Gamma_0}\binom{N}{\mathbf b}Z_{\mathbf b}}}
    \ket{\Symb{N+K}{\mathbf b+\boldsymbol\Gamma}} \nonumber\\
    &= \sqrt{\frac{\binom{N+K}{\mathbf b+\boldsymbol\Gamma}}{Z_{\mathbf b}}}
    \ket{\Symb{N+K}{\mathbf b+\boldsymbol\Gamma}},
\end{align}
Here the powers of $d$ cancel, and we used $\binom{K}{\boldsymbol\Gamma}=\binom{K}{\boldsymbol\Gamma_0}$ because every balanced shift has $s$ entries equal to $q+1$ and $d-s$ entries equal to $q$, so the product of factorials in the multinomial denominator is the same for all balanced shifts. We therefore recover \cref{eq:phase-kraus-main} explicitly:
\begin{align}
    A_{\boldsymbol\Gamma}
    = \sum_{\mathbf b\in\cB_N}
    \sqrt{\frac{\binom{N+K}{\mathbf b+\boldsymbol\Gamma}}{Z_{\mathbf b}}}
    \ket{\Symb{N+K}{\mathbf b+\boldsymbol\Gamma}}
    \bra{\Symb{N}{\mathbf b}}.
\end{align}
Taking the adjoint and multiplying gives
\begin{align}
    A_{\boldsymbol\Gamma}^{\dagger}A_{\boldsymbol\Gamma}
    &= \sum_{\mathbf b,\mathbf c\in\cB_N}
    \sqrt{\frac{\binom{N+K}{\mathbf b+\boldsymbol\Gamma}\binom{N+K}{\mathbf c+\boldsymbol\Gamma}}{Z_{\mathbf b}Z_{\mathbf c}}}
    \braket{\Symb{N+K}{\mathbf b+\boldsymbol\Gamma}}{\Symb{N+K}{\mathbf c+\boldsymbol\Gamma}}
    \ket{\Symb{N}{\mathbf b}}\bra{\Symb{N}{\mathbf c}} \nonumber\\
    &= \sum_{\mathbf b\in\cB_N}
    \frac{\binom{N+K}{\mathbf b+\boldsymbol\Gamma}}{Z_{\mathbf b}}
    \ketbra{\Symb{N}{\mathbf b}}{\Symb{N}{\mathbf b}},
\end{align}

Summing over the balanced shifts and using the definition of $Z_{\mathbf b}$, we find
\begin{align}
    \sum_{\boldsymbol\Gamma\in\cC_{d,K}}
    A_{\boldsymbol\Gamma}^{\dagger}A_{\boldsymbol\Gamma}
    &= \sum_{\mathbf b\in\cB_N}
    \frac{\sum_{\boldsymbol\Gamma\in\cC_{d,K}}\binom{N+K}{\mathbf b+\boldsymbol\Gamma}}{Z_{\mathbf b}}
    \ketbra{\Symb{N}{\mathbf b}}{\Symb{N}{\mathbf b}} \nonumber\\
    &= \sum_{\mathbf b\in\cB_N}
    \ketbra{\Symb{N}{\mathbf b}}{\Symb{N}{\mathbf b}} \nonumber\\
    &= \Id_{\Sym^N(\C^d)}.
\end{align}
Thus the channel is trace preserving on $\Sym^N(\C^d)$. 
\end{proof}

The similarity between \cref{eq:petz-channel-main} and the resolving-orbit cloner reflects a common moment-normalized transpose construction. It does not by itself imply a general optimality principle: for a resolving orbit, the relevant moment is proportional to a projector and the normalization simplifies accordingly, whereas the phase moments are nonuniform and optimality requires the additional restriction $P_{\mathrm{bal},K}$ together with the SDP certificate below.

\subsection{Performance operator and fidelity}
\label{subsec:phase-performance}

The Choi picture makes the fidelity calculation particularly transparent. Writing $\rho_x=\ketbra{\psi_x}{\psi_x}$, define the performance operator
\begin{align}
    \Omega_{N,K} \coloneq \int (\rho_x^{\otimes N})^T \otimes \rho_x^{\otimes(N+K)} d\mu(x), \label{eq:performance-operator-main}
\end{align}
where the transpose is taken in the occupation basis. For the unnormalized Choi operator $J_\Phi$ of a channel $\Phi$,
\begin{align}
    F_{\mathrm{av}}(\Phi) = \Tr(J_\Phi\Omega_{N,K}). \label{eq:performance-overlap-main}
\end{align}
Character orthogonality gives the rank-one decomposition
\begin{align}
    \Omega_{N,K} &= \sum_{\boldsymbol\Delta} \ketbra{ v_{\boldsymbol\Delta} }{ v_{\boldsymbol\Delta} }, \label{eq:performance-decomposition-main}\\
    \ket{v_{\boldsymbol\Delta}} &= \frac{1}{d^{(2N+K)/2}} \sum_{\mathbf b\in\mathcal D_{\boldsymbol\Delta}} \sqrt{ \binom{N}{\mathbf b} \binom{N+K}{\mathbf b+\boldsymbol\Delta} } \ket{\Symb{N}{\mathbf b}} \otimes \ket{\Symb{N+K}{\mathbf b + \mathbf \Delta}}. \label{eq:performance-vector-main}
\end{align}
Thus the vectors $v_{\boldsymbol\Delta}$ in the block SDP are precisely the rank-one vectors of the ensemble performance operator.

Vectorizing the Kraus operators in \cref{eq:phase-kraus-main} gives
\begin{align}
    \ket{u_{\boldsymbol\Gamma}} \coloneq \operatorname{vec}(A_{\boldsymbol\Gamma}) = \sum_{\mathbf b\in\cB_N} \sqrt{ \frac{ \binom{N+K}{\mathbf b+\boldsymbol\Gamma} }{ Z_{\mathbf b} } } \ket{\Symb{N}{\mathbf b}} \otimes \ket{\Symb{N+K}{\mathbf b + \mathbf \Gamma}}, \label{eq:kraus-vector-main}
\end{align}
and hence
\begin{align}
    J_{\mathrm{ph}} = \sum_{\boldsymbol\Gamma\in\cC_{d,K}} \ketbra{ u_{\boldsymbol\Gamma} }{ u_{\boldsymbol\Gamma} }.
\end{align}
The vectors $\ket{v_{\boldsymbol\Gamma}}$ and $\ket{u_{\boldsymbol\Gamma}}$ have different input-dependent weights but lie in the same occupation-shift sector. Since distinct sectors are orthogonal,
\begin{align}
    F_{\mathrm{ph}} &= \Tr(J_{\mathrm{ph}}\Omega_{N,K}) = \sum_{\boldsymbol\Gamma\in\cC_{d,K}} \left| \braket{v_{\boldsymbol\Gamma}}{u_{\boldsymbol\Gamma}} \right|^2, \label{eq:performance-sector-overlap-main}
\end{align}
where
\begin{align}
    \braket{v_{\boldsymbol\Gamma}}{u_{\boldsymbol\Gamma}} = \frac{1}{d^{(2N+K)/2}} \sum_{\mathbf b\in\cB_N} \frac{ \sqrt{\binom{N}{\mathbf b}}\, \binom{N+K}{\mathbf b+\boldsymbol\Gamma} }{ \sqrt{Z_{\mathbf b}} }. \label{eq:performance-sector-coefficient-main}
\end{align}
Equivalently, applying one Kraus operator directly to the phase state gives
\begin{align}
    \bra{\psi_x}^{\otimes(N+K)} A_{\boldsymbol\Gamma} \ket{\psi_x}^{\otimes N} = \frac{ \overline{x}^{\boldsymbol\Gamma} }{ d^{(2N+K)/2} } \sum_{\mathbf b\in\cB_N} \frac{ \sqrt{\binom{N}{\mathbf b}}\, \binom{N+K}{\mathbf b+\boldsymbol\Gamma} }{ \sqrt{Z_{\mathbf b}} }. \label{eq:phase-kraus-overlap-main}
\end{align}
The phase factor has modulus one, which proves that the fidelity is independent of $x$. All balanced shifts are coordinate permutations of one another. Since the remaining sum is invariant under coordinate permutations, it has the same value for every balanced shift. Squaring \cref{eq:phase-kraus-overlap-main} and summing over $\boldsymbol\Gamma\in\cC_{d,K}$ proves the fidelity formula in \cref{thm:phase-main}.

\subsection{Balanced shifts and optimality}
\label{subsec:phase-optimality}

It remains to prove that no other occupation-shift sectors can produce a larger fidelity. Our strategy is to show that the quantity $\Lambda(\boldsymbol\Delta)$ defined below cannot decrease as the shift becomes more balanced: transferring one unit from a larger entry to a smaller entry cannot decrease its value whenever the two entries differ by at least two. Iterating this step shows that balanced shifts maximize $\Lambda$, which allows us to construct a dual certificate matching the fidelity of our channel and thereby prove optimality. For any integer shift $\boldsymbol\Delta$ with $|\boldsymbol\Delta|=K$, define
\begin{align}
    \Lambda(\boldsymbol\Delta) \coloneq \sum_{\mathbf b\in\mathcal D_{\boldsymbol\Delta}} \frac{ \sqrt{\binom{N}{\mathbf b}}\, \binom{N+K}{\mathbf b+\boldsymbol\Delta} }{ \sqrt{Z_{\mathbf b}} }. \label{eq:Lambda-main}
\end{align}
We use the standard convention that a multinomial coefficient is zero if one of its lower entries is negative. This allows us to write the same quantity over the shift-independent domain $\cB_N$:
\begin{align}
    \Lambda(\boldsymbol\Delta) = \sum_{\mathbf b\in\cB_N} \frac{ \sqrt{\binom{N}{\mathbf b}}\, \binom{N+K}{\mathbf b+\boldsymbol\Delta} }{ \sqrt{Z_{\mathbf b}} }. \label{eq:Lambda-common-domain-main}
\end{align}

\begin{lemma}[Two-coordinate smoothing]
\label{lem:phase-smoothing-main}
Suppose that $\Delta_i\geq\Delta_j+2$, and define
\begin{align}
    \boldsymbol\Delta' \coloneq \boldsymbol\Delta-\mathbf e_i+\mathbf e_j.
\end{align}
Then
\begin{align}
    \Lambda(\boldsymbol\Delta') \geq \Lambda(\boldsymbol\Delta). \label{eq:phase-smoothing-main}
\end{align}
\end{lemma}

The proof is given in \cref{app:phase-smoothing}. Its main idea is to partition the input occupations into lines on which only coordinates $i$ and $j$ vary. On each line, the weight $\sqrt{\binom{N}{\mathbf b}/Z_{\mathbf b}}$ is symmetric and log-concave, while the remaining multinomial coefficient becomes a translated binomial sequence. Hoggar's convolution theorem shows that the complete line contribution is log-concave in the translation parameter \cite{Hoggar1974}, and a direct reflection argument shows that it is symmetric. Replacing $(\Delta_i,\Delta_j)$ by $(\Delta_i-1,\Delta_j+1)$ moves this parameter toward the center and therefore cannot decrease the line contribution.

Before comparing arbitrary shifts with balanced ones, we verify that $\Lambda$ takes the same value on every balanced shift. Fix $\boldsymbol\Gamma_0,\boldsymbol\Gamma\in\cC_{d,K}$ and choose a coordinate permutation $\sigma$ such that $\boldsymbol\Gamma=\sigma\boldsymbol\Gamma_0$. Since $\sigma$ maps $\cC_{d,K}$ bijectively onto itself and multinomial coefficients are invariant under coordinate permutations, we have
\begin{align}
    Z_{\sigma\mathbf b}
    &= \sum_{\boldsymbol\Theta\in\cC_{d,K}}
    \binom{N+K}{\sigma\mathbf b+\sigma\boldsymbol\Theta} \nonumber\\
    &= \sum_{\boldsymbol\Theta\in\cC_{d,K}}
    \binom{N+K}{\mathbf b+\boldsymbol\Theta}
    = Z_{\mathbf b}.
\end{align}
Multinomial coefficients are invariant under coordinate permutations. Together with $Z_{\sigma\mathbf b}=Z_{\mathbf b}$ and the invariance of $\cB_N$ under $\sigma$, this implies that the sum in \cref{eq:Lambda-common-domain-main} satisfies $\Lambda(\boldsymbol\Gamma)=\Lambda(\sigma\boldsymbol\Gamma_0)=\Lambda(\boldsymbol\Gamma_0)$.
Thus all balanced shifts have the same value of $\Lambda$.

\begin{proposition}[Balanced shifts maximize the overlap]
\label{prop:balanced-max-main}
For every integer shift $\boldsymbol\Delta$ with $|\boldsymbol\Delta|=K$ and every $\boldsymbol\Gamma_0\in\cC_{d,K}$,
\begin{align}
    \Lambda(\boldsymbol\Delta) \leq \Lambda(\boldsymbol\Gamma_0). \label{eq:balanced-max-main}
\end{align}
\end{proposition}

\begin{proof}
If $\boldsymbol\Delta$ is not balanced, then some pair of coordinates satisfies $\Delta_i\geq\Delta_j+2$. By \cref{lem:phase-smoothing-main}, transferring one unit from coordinate $i$ to coordinate $j$ does not decrease $\Lambda$. Each transfer strictly decreases $\sum_j\Delta_j^2$, so iteration terminates at a vector whose coordinates differ by at most one. Since its coordinates sum to $K$, the terminal vector belongs to $\cC_{d,K}$. All elements of $\cC_{d,K}$ have the same value of $\Lambda$ by coordinate-permutation symmetry.
\end{proof}

We can now give a dual certificate that matches the balanced channel. Fix $\boldsymbol\Gamma_0\in\cC_{d,K}$ and define
\begin{align}
    y_{\mathbf b} \coloneq \frac{ \Lambda(\boldsymbol\Gamma_0) }{ d^{2N+K} } \sqrt{ \binom{N}{\mathbf b} Z_{\mathbf b} }. \label{eq:dual-certificate-main}
\end{align}
Since $y_{\mathbf b}>0$ for every $\mathbf b\in\cB_N$, it suffices to verify the scalar dual constraint in \cref{eq:rank-one-dual-main}. Substituting \cref{eq:v-delta-main,eq:dual-certificate-main}, we obtain
\begin{align}
    \sum_{\mathbf b\in\mathcal D_{\boldsymbol\Delta}}
    \frac{|v_{\boldsymbol\Delta}(\mathbf b)|^2}{y_{\mathbf b}}
    &= \sum_{\mathbf b\in\mathcal D_{\boldsymbol\Delta}}
    \frac{
        d^{-(2N+K)}
        \binom{N}{\mathbf b}
        \binom{N+K}{\mathbf b+\boldsymbol\Delta}
    }{
        \Lambda(\boldsymbol\Gamma_0)d^{-(2N+K)}
        \sqrt{\binom{N}{\mathbf b}Z_{\mathbf b}}
    } \nonumber\\
    &= \frac{1}{\Lambda(\boldsymbol\Gamma_0)}
    \sum_{\mathbf b\in\mathcal D_{\boldsymbol\Delta}}
    \frac{
        \sqrt{\binom{N}{\mathbf b}}\,
        \binom{N+K}{\mathbf b+\boldsymbol\Delta}
    }{\sqrt{Z_{\mathbf b}}} \nonumber\\
    &= \frac{\Lambda(\boldsymbol\Delta)}{\Lambda(\boldsymbol\Gamma_0)}
    \leq 1.
    \label{eq:dual-feasible-main}
\end{align}
The remaining sum is exactly $\Lambda(\boldsymbol\Delta)$ by \cref{eq:Lambda-main}, and the inequality follows from \cref{prop:balanced-max-main}. Thus every shift satisfies its dual constraint, so the dual point is feasible. Using the definition of $Z_{\mathbf b}$ and exchanging the order of summation, we obtain
\begin{align}
    \sum_{\mathbf b\in\cB_N}\sqrt{\binom{N}{\mathbf b}Z_{\mathbf b}}
    &=\sum_{\mathbf b\in\cB_N}\frac{\sqrt{\binom{N}{\mathbf b}}}{\sqrt{Z_{\mathbf b}}}Z_{\mathbf b} \nonumber\\
    &=\sum_{\mathbf b\in\cB_N}\sum_{\boldsymbol\Gamma\in\cC_{d,K}}\frac{\sqrt{\binom{N}{\mathbf b}}\,\binom{N+K}{\mathbf b+\boldsymbol\Gamma}}{\sqrt{Z_{\mathbf b}}} \nonumber\\
    &=\sum_{\boldsymbol\Gamma\in\cC_{d,K}}\sum_{\mathbf b\in\cB_N}\frac{\sqrt{\binom{N}{\mathbf b}}\,\binom{N+K}{\mathbf b+\boldsymbol\Gamma}}{\sqrt{Z_{\mathbf b}}} \nonumber\\
    &=\sum_{\boldsymbol\Gamma\in\cC_{d,K}}\Lambda(\boldsymbol\Gamma) \nonumber\\
    &=|\cC_{d,K}|\Lambda(\boldsymbol\Gamma_0).
    \label{eq:balanced-lambda-sum-main}
\end{align}
The final equality follows because every $\boldsymbol\Gamma\in\cC_{d,K}$ is a coordinate permutation of $\boldsymbol\Gamma_0$, while the definition of $\Lambda$ is invariant under coordinate permutations. Therefore the objective value of the dual point is
\begin{align}
    \sum_{\mathbf b\in\cB_N}y_{\mathbf b}
    &= \frac{\Lambda(\boldsymbol\Gamma_0)}{d^{2N+K}}
    \sum_{\mathbf b\in\cB_N}
    \sqrt{\binom{N}{\mathbf b}Z_{\mathbf b}} \nonumber\\
    &= \frac{|\cC_{d,K}|}{d^{2N+K}}
    \Lambda(\boldsymbol\Gamma_0)^2 \nonumber\\
    &= \frac{|\cC_{d,K}|}{d^{2N+K}}
    \left(
        \sum_{\mathbf b\in\cB_N}
        \frac{
            \sqrt{\binom{N}{\mathbf b}}\,
            \binom{N+K}{\mathbf b+\boldsymbol\Gamma_0}
        }{\sqrt{Z_{\mathbf b}}}
    \right)^2 \nonumber\\
    &= F_{\mathrm{ph}}^\star(d,N,K).
    \label{eq:dual-matches-main}
\end{align}
The balanced channel is primal feasible by \cref{prop:phase-normalization-main}, and its fidelity equals the value of this dual feasible point. Weak duality therefore proves the optimality claim in \cref{thm:phase-main}.

\subsection{Uniform fidelity bounds for adding one copy}
\label{subsec:phase-one-copy-asymptotics}

We now extract the sample complexity from the exact optimal-fidelity formula for $K=1$. Throughout, $d\ge2$ and $N\ge1$. Recall that $\ket{\psi_x}=d^{-1/2}\sum_{j=0}^{d-1}x_j\ket j$ with $|x_j|=1$. Let $\cS_{\mathrm{ph}}=\{\ket{\psi_x}\bra{\psi_x}:|x_j|=1\}$ denote the phase-state family, and put
\begin{align}
    \delta_{N,d}\coloneq1-F_{\mathrm{ph}}^\star(d,N,1).
    \label{eq:phase-one-copy-infidelity}
\end{align}
Here the global fidelity is the overlap of the channel output with the entire ideal state $\ket{\psi_x}^{\otimes(N+1)}$. The optimal channel has input-independent fidelity, so the average and worst-case optima coincide.

\begin{theorem}[Uniform infidelity bounds]
\label{thm:phase-one-copy-uniform}
There are absolute constants $c,C>0$ such that, for every $d\ge2$ and $N\ge1$,
\begin{align}
    c\left(\left(1-\frac1d\right)^{N+1}+\frac{d-1}{(N+d)^2}\right)
    \le\delta_{N,d}
    \le C\left(\left(1-\frac1d\right)^{N+1}+\frac{d-1}{(N+d)^2}\right).
    \label{eq:phase-one-copy-uniform}
\end{align}
\end{theorem}

For a target global infidelity $\delta$, define
\begin{align}
    N_{\min,\cS_{\mathrm{ph}}}^{F}(\delta)
    \coloneq\min\{N\ge1:F_{\mathrm{ph}}^\star(d,N,1)\ge1-\delta\}.
    \label{eq:phase-one-copy-complexity-definition}
\end{align}
The superscript $F$ distinguishes this fidelity criterion from the trace-distance complexity $N_{\min,\cS_{\mathrm{ph}}}(\epsilon)$ in \cref{def:approximate_cloning}.

\begin{corollary}[Fidelity sample complexity]
\label{cor:phase-one-copy-complexity}
Uniformly for $d\ge2$ and $0<\delta\le1/4$,
\begin{align}
    N_{\min,\cS_{\mathrm{ph}}}^{F}(\delta)
    =\Theta\!\left(d\log\frac1\delta+\sqrt{\frac d\delta}\right).
    \label{eq:phase-one-copy-complexity}
\end{align}
\end{corollary}

For the trace-distance convention of \cref{def:approximate_cloning}, the pure-target inequalities $1-F\le\frac12\|\rho-\ket{\psi}\bra{\psi}\|_1\le\sqrt{1-F}$ imply
\begin{align}
    N_{\min,\cS_{\mathrm{ph}}}^{F}(\epsilon)
    \le N_{\min,\cS_{\mathrm{ph}}}(\epsilon)
    \le N_{\min,\cS_{\mathrm{ph}}}^{F}(\epsilon^2).
    \label{eq:phase-one-copy-trace-comparison}
\end{align}
In particular, the trace-distance complexity is $\Theta(d)$ at every fixed $0<\epsilon\le1/4$. Combining \cref{eq:phase-one-copy-trace-comparison} with \cref{cor:phase-one-copy-complexity}, there are absolute constants $c_1,C_1>0$ such that, for $d\ge2$ and $0<\epsilon\le1/4$,
\begin{align}
c_1\left(d\log\frac1\epsilon+\sqrt{\frac d\epsilon}\right)
&\le N_{\min,\cS_{\mathrm{ph}}}(\epsilon)
\le C_1\left(d\log\frac1\epsilon+\frac{\sqrt d}{\epsilon}\right).
\label{eq:phase-one-copy-trace-bounds}
\end{align}
If $d\ge\epsilon^{-2}$, then $\sqrt d/\epsilon\le d$, so these bounds imply
\begin{align}
N_{\min,\cS_{\mathrm{ph}}}(\epsilon)
&=\Theta\left(d\log\frac1\epsilon\right).
\label{eq:phase-one-copy-trace-high-d}
\end{align}
For comparison, applying \cref{thm:structured_optimal_cloning} to the unstructured family $\cS_{\mathcal H}$ with $\mathcal H=\C^d$ gives optimal trace error
\begin{align}
1-\frac{\dim\operatorname{Sym}^N(\C^d)}{\dim\operatorname{Sym}^{N+1}(\C^d)}
&=1-\frac{N+1}{N+d}
=\frac{d-1}{N+d},
\end{align}
and hence $N_{\min,\cS_{\mathcal H}}(\epsilon)=\Theta(d/\epsilon)$. Thus, in the stated regime, phase states have logarithmic rather than inverse dependence on the target trace error.

We now prove \cref{thm:phase-one-copy-uniform}, starting with an exact decomposition of the infidelity into an explicit term and a variance. Let $J_1,\ldots,J_N$ be independent uniform elements of $\{0,\ldots,d-1\}$, and let $\mathbf{B}=(B_0,\ldots,B_{d-1})$ be their occupation vector, with $B_j=\sum_{r=1}^{N}\mathbf1_{\{J_r=j\}}$. Thus
\begin{align}
    \mathbb{P}(\mathbf{B}=\mathbf{b})=d^{-N}\binom N{\mathbf{b}},
    \qquad
    X(\mathbf{b})=\frac{N+1}{d^2}\sum_{j=0}^{d-1}\frac1{b_j+1}.
    \label{eq:phase-one-copy-auxiliary-function}
\end{align}
We use the existing occupation notation $\cB_N=\{\mathbf{b}\in\mathbb Z_{\ge0}^{d}:\sum_j b_j=N\}$ and $\binom N{\mathbf{b}}=N!/\prod_j b_j!$, and write $X=X(\mathbf{B})$.

\begin{lemma}[Exact decomposition]
\label{lem:phase-one-copy-decomposition}
The optimal fidelity and infidelity satisfy
\begin{align}
    F_{\mathrm{ph}}^\star(d,N,1)=(\E[\sqrt X])^2,
    \qquad
    \delta_{N,d}=\left(1-\frac1d\right)^{N+1}+\operatorname{Var}(\sqrt X).
    \label{eq:phase-one-copy-decomposition}
\end{align}
\end{lemma}
\begin{proof}
For $K=1$, the balanced shifts are the unit vectors $\mathbf{e}_0,\ldots,\mathbf{e}_{d-1}$. Choosing $\boldsymbol\Gamma_0=\mathbf{e}_0$, the exact fidelity formula specializes to
\begin{align}
    F_{\mathrm{ph}}^\star(d,N,1)
    =\frac1{d^{2N}}\left(
    \sum_{\mathbf{b}\in\cB_N}
    \frac{\sqrt{\binom N{\mathbf{b}}}\binom{N+1}{\mathbf{b}+\mathbf{e}_0}}{\sqrt{Z_{\mathbf{b}}}}
    \right)^2,
    \qquad
    Z_{\mathbf{b}}=\sum_{j=0}^{d-1}\binom{N+1}{\mathbf{b}+\mathbf{e}_j}.
    \label{eq:phase-one-copy-exact-fidelity}
\end{align}
The identity
\begin{align}
    \binom{N+1}{\mathbf{b}+\mathbf{e}_j}
    =\frac{N+1}{b_j+1}\binom N{\mathbf{b}}
\end{align}
gives
\begin{align}
    Z_{\mathbf{b}}
    &=(N+1)\binom N{\mathbf{b}}\sum_{j=0}^{d-1}\frac1{b_j+1}
    =d^2\binom N{\mathbf{b}}X(\mathbf{b}).
    \label{eq:phase-one-copy-normalization-X}
\end{align}
Recalling \cref{eq:phase-one-copy-exact-fidelity} with $\boldsymbol\Gamma_0=\mathbf{e}_0$ and substituting these expressions for the multinomial coefficient and $Z_{\mathbf{b}}$, we obtain
\begin{align}
    F_{\mathrm{ph}}^\star(d,N,1)
    &=\frac1{d^{2N}}\left(
    \sum_{\mathbf{b}\in\cB_N}
    \frac{\sqrt{\binom N{\mathbf{b}}}\binom{N+1}{\mathbf{b}+\mathbf{e}_0}}
    {\sqrt{Z_{\mathbf{b}}}}
    \right)^2 \nonumber\\
    &=\frac1{d^{2N}}\left(
    \sum_{\mathbf{b}\in\cB_N}
    \frac{\sqrt{\binom N{\mathbf{b}}}\,(N+1)\binom N{\mathbf{b}}}
    {d(b_0+1)\sqrt{\binom N{\mathbf{b}}X(\mathbf{b})}}
    \right)^2 \nonumber\\
    &=\frac1{d^{2N}}\left(
    \frac{N+1}{d}\sum_{\mathbf{b}\in\cB_N}
    \binom N{\mathbf{b}}\frac1{(b_0+1)\sqrt{X(\mathbf{b})}}
    \right)^2 \nonumber\\
    &=\frac1{d^{2N}}\left(
    (N+1)d^{N-1}\,
    \E\left[\frac1{(B_0+1)\sqrt{X(\mathbf{B})}}\right]
    \right)^2.
    \label{eq:phase-one-copy-fidelity-expectation}
\end{align}
Both the law of $\mathbf{B}$ and the function $X(\mathbf{B})$ are invariant under permutations of the coordinates. We may therefore average the expectation over all $d$ coordinates:
\begin{align}
    \E\left[\frac1{(B_0+1)\sqrt{X(\mathbf{B})}}\right]
    &=\frac1d\sum_{j=0}^{d-1}
    \E\left[\frac1{(B_j+1)\sqrt{X(\mathbf{B})}}\right] \nonumber\\
    &=\frac1d\,\E\left[
    \frac1{\sqrt{X(\mathbf{B})}}\sum_{j=0}^{d-1}\frac1{B_j+1}
    \right] \nonumber\\
    &=\frac1d\,\E\left[
    \frac{d^2X(\mathbf{B})}{(N+1)\sqrt{X(\mathbf{B})}}
    \right]
    =\frac d{N+1}\E[\sqrt{X(\mathbf{B})}].
    \label{eq:phase-one-copy-coordinate-average}
\end{align}
In the third equality we used the definition of $X$. Substituting \cref{eq:phase-one-copy-coordinate-average} into \cref{eq:phase-one-copy-fidelity-expectation} now gives
\begin{align}
    F_{\mathrm{ph}}^\star(d,N,1)
    &=\frac1{d^{2N}}\left(
    (N+1)d^{N-1}\frac d{N+1}\E[\sqrt{X(\mathbf{B})}]
    \right)^2 \nonumber\\
    &=\frac1{d^{2N}}\left(d^N\E[\sqrt{X(\mathbf{B})}]\right)^2
    =\left(\E[\sqrt X]\right)^2.
    \label{eq:phase-one-copy-fidelity-simplified}
\end{align}

The standard binomial reciprocal-moment identity \cite{ChaoStrawderman1972} gives $\E[X]=1-(1-1/d)^{N+1}$, since $B_j\sim\Bin(N,1/d)$. Consequently
\begin{align}
    \delta_{N,d}
    &=1-F_{\mathrm{ph}}^\star(d,N,1)
    =1-(\E[\sqrt X])^2 \nonumber\\
    &=1-\E[X]+\operatorname{Var}(\sqrt X)
    =\left(1-\frac1d\right)^{N+1}+\operatorname{Var}(\sqrt X).
    \qquad\qedhere
\end{align}
\end{proof}

\begin{lemma}[Uniform variance estimate]
\label{lem:phase-one-copy-variance-main}
There are absolute constants $c_0,C_0>0$ such that, for $d\ge2$ and $N\ge2d$,
\begin{align}
    c_0\frac{d-1}{N^2}
    \le\operatorname{Var}(\sqrt X)
    \le C_0\frac{d-1}{N^2}.
    \label{eq:phase-one-copy-variance-main}
\end{align}
\end{lemma}

The upper bound controls how $X$ changes when one sample $J_i$ is replaced by an independent copy. For the lower bound, we compare $\sqrt X$ with the number of pairs of equal samples, and bound their covariance from below. The proof, including the resampling calculation and the lower bound on the second derivative, is given in \cref{app:phase-one-copy-upper,app:phase-one-copy-lower}.

Combining \cref{lem:phase-one-copy-decomposition,lem:phase-one-copy-variance-main} gives \cref{thm:phase-one-copy-uniform} for $N\ge 2d$. The remaining range $1\le N<2d$ is treated in \cref{app:phase-one-copy-completion}.

The two terms in \cref{eq:phase-one-copy-complexity} come from the two contributions in \cref{eq:phase-one-copy-decomposition}. Suppressing $(1-1/d)^{N+1}$ to order $\delta$ requires order $d\log(1/\delta)$ copies. Suppressing the variance to order $\delta$ requires order $\sqrt{d/\delta}$ copies in the range $N\ge2d$. The uniform bounds account for the remaining range. The complete inversion proving \cref{cor:phase-one-copy-complexity} is given in \cref{app:phase-one-copy-completion}.

%% file: SimpleOrbits.tex
\section{Simple examples of resolving orbits}
\label{sec:other_resolving_orbits}
Many structured state families of interest can be shown to fall under the definition of resolving orbits. Here we consider a number of example families, for which this property is straightforward to demonstrate, and discuss the conclusions for their optimal cloning; a part of these cases were previously known in the cloning literature. In the main text, we focus on the cases of multi-mode bosonic and fermionic Gaussian states, which are both new to the literature, and where proving the resolving orbit property is less straightforward.

Let us first take an inventory of the simplest example, which is also the comparison point for the performance of any structured state cloning --- the unstructured case $\cS=\cS_{\cH}$ of Werner for finite-dimensional $\cH$ \cite{werner1998optimal}. In that case, the optimal cloning fidelity is $f_{\cH}=\f{\dim(\Sym^N(\cH))}{\dim(\Sym^{N+K}(\cH))}$, where $\dim(\Sym^{N}(\cH))=\binom{\dim(\cH) + N - 1}{N}$. This in particular implies that, for large $\dim (\cH)$, cloning sample complexity scales as $ N_{\mr{min},\cH}=\Theta(\dim (\cH)/\epsilon)$. 

At the other extreme, let us consider the set $\cS_{\mr{orth}}$ of $d$ mutually orthogonal states $\{\ket{\psi_i}\}$, which is well-known to allow for perfect $1\rarr 2$ cloning. 
In fact, it also forms a resolving orbit. In particular, $\cS_{\mr{orth}}$ is the orbit of the so-called unitary monomial group $G_{\mr{mon}}=\{g_{\boldsymbol{\theta},\sigma}\}$, formed by permutation matrices $\sigma\in \rS_d$ decorated by arbitrary phases $\boldsymbol{\theta}=\{\theta_i\in[0,2\pi)\}$: $g_{\boldsymbol{\theta},\sigma} \ket{\psi_i}=e^{i\theta_i}\ket{\psi_{\sigma(i)}}$ ($G_{\mr{mon}}$ is a semi-direct product $S_d\ltimes U(1)^d$). By Schur's lemma, the defining representation $\phi$ of this group is irreducible on the $d$-dimensional span of $\{\ket{\psi_i}\}$. Indeed, consider a matrix $A$ that commutes with any element in the image $\phi(G_{\mr{mon}})$. Commuting with any $\phi(g_{\boldsymbol{\theta},\sigma=\mr{id}})$ diagonal matrix implies that $A$ is diagonal; and commuting with any permutation $\phi(g_{\boldsymbol{\theta}=\boldsymbol{0},\sigma})$ further implies that $A$ is indeed proportional to the identity --- finishing the proof of irreducibility of $\phi$. The same argument implies the irreducibility of $G_{\mr{mon}}$ as represented on the $N$-copy span $\cH^{(N)}_{\cS_{\mr{orth}}}=\Span(\{\ket{\psi_i}^{\otimes N}\})$. Indeed, under the natural isomorphism $\cH^{(N)}_{\cS_{\mr{orth}}} \cong \cH^{(1)}_{\cS_{\mr{orth}}}$ (given by $\ket{\psi_i}^{\otimes N} \mapsto \ket{\psi_i}$), the image of the tensor power representation $\phi^{\otimes N}$ restricted to this span is identical to the image of the defining representation $\phi$ ($\phi^{\otimes N}(g_{\boldsymbol{\theta},\sigma}) \cong \phi(g_{N\boldsymbol{\theta}\pmod{2\pi},\sigma})$), thus generating the same full matrix algebra and guaranteeing irreducibility for any integer $N$. This implies that the channel presented in \cref{thm:structured_optimal_cloning} is optimal, and indeed, it yields perfect fidelity $f_{\cS_{\mr{orth}}}=\f{d_N}{d_{N+K}}=\f{d}{d}=1$ (because $\dim(\cH^{(N)}_{\cS_{\mr{orth}}})=d$), as expected for this classical case of mutually orthogonal states.

Consider cloning of pure product states of $k$ qudits, i.e., states from the set $\cS_{\mr{prod}}=\{\psi_1\otimes\psi_2\otimes..\otimes\psi_k,~ \ket{\psi_j}\in\bC^d\}$. For this task, it is natural to propose a channel $T=T_d^{\otimes k}$ where $T_d$ is the Werner cloner for $\cH=\bC^d$. In fact, this channel is optimal -- which follows from the fact that $\cS_{\mr{prod}}$ is itself a resolving orbit\footnote{Proving more generally (without using coherent orbits) that a tensor product of optimal cloners is itself optimal is surprisingly tedious, and we skip it in this version of the text.}. This allows to apply \cref{thm:structured_optimal_cloning} to give an optimal cloning channel, which for this case indeed has form $T=T_d^{\otimes k}$. This, in turn, means that the cloning error is proportional to the size of the system, and therefore the cloning sampling complexity of product states is $N_{\mr{min},\cS_{\mr{prod}}}(\epsilon)=\Theta(k(d-1)/\epsilon)$.

Why is $\cS_{\mr{prod}}$ a resolving orbit? It can be shown that a product $\cS=\bigotimes^k_{j=1}\cS_j$ of any $k$ resolving orbits $\cS_j$ (structured or not, finite or infinite dimensional) is again a resolving orbit. To demonstrate this, it is sufficient to consider $k=2$. Suppose
$\cH=\cH_1\otimes\cH_2$ and $\cS= \cS_1\otimes\cS_2$
for $\cS_{1,2}\subset\cD(\cH_{1,2})$, where $\cS_{1,2}$ are resolving orbits of a group $G_{1,2}$ under
representations $\phi_{1,2}:G_{1,2}\rarr\rU(\cH_{1,2})$. $\cS$ is an
orbit of the direct product $G_1\times G_2$, under a representation $\phi=\phi_1\otimes\phi_2$. Consider an operator $A$ that commutes with any $\phi(g)$: it then commutes with any operator of the form $\phi_1(g_1)\otimes\bI_2$ or $\bI_1\otimes\phi_2(g_2)$, which by irreducibility of $\phi_{1,2}$ forces $A=\bI_1\otimes\bI_2$; it follows that $\phi$ is irreducible. The same argument applies to $\phi^{\otimes N}$ restricted to the $N$-copy span of $\cS$, the key requirement for a resolving orbit (cf. \cref{def:resolving_orbit}). The square-integrability requirement follows by Tonnelli's theorem.

Bosonic coherent states provide an infinite-dimensional resolving orbit.
For $n$ modes, the degree is $d_N=N^n$ with normalized phase-space measure, so the optimal global fidelity is $(N/(N+K))^n$ and the optimal trace error is $1-(N/(N+K))^n$.
The displacement group is amenable, making the optimum unrestricted.
We give the collective-mode derivation and the amplifier implementation in \cref{sec:active-bosonic-cloning}, together with the larger squeezed and displaced Gaussian families.

A physically curious example (and to the best of our knowledge not previously considered in the literature) is the family $\cS_{\mr{cond}}$ of Bose condensate states of $q$ bosons on $n$ modes. In particular, these are states of the form
\begin{align}
    \ket{\psi(v,q)}=\f{1}{\sqrt{q!}}\lrp{\sum^n_{j=1} v_j \,a^\dag_j}^{q}\ket{0}.
\end{align}
These states have a physical interpretation as ground states of non-interacting Hamiltonians in the $q$-particle sector. $\cS_{\mr{cond}}=\{\ket{\psi(v,q)}\}$ forms a resolving orbit of $\rU(n)$. This can be seen by identifying the Hilbert space of $q$ bosons on $n$ modes with the symmetric space $\Sym^{q}(\bC^{n})$, where Bose condensate states become simply $\ket{\psi(v,q)}=v^{\vee q}$, and $\rU(n)$ is represented as
\begin{align}
    \phi(U)\ket{\psi(v,q)}=(Uv)^{\vee q}.
\end{align}
(Note that $\ket{\psi(v,q)}$ forms an overcomplete basis on $\Sym^{q}(\bC^{n})$.) Restricted to the $N$-copy span $\cH^{(N)}_{\cS_{\mr{cond}}}=\Span(\{\ket{\psi(v,q)}^{\otimes N}\})=\Sym^{Nq}(\bC^{n})$, the representation $\phi^{\otimes N}$ acts simply as $\phi^{\otimes N}(U)\ket{\psi(v,q)}^{\otimes N}=(Uv)^{\vee qN}$, a standard fully symmetric representation of $\rU(n)$, which is irreducible for any $N$ -- thus $\cS_{\mr{cond}}$ is resolving. Note that the case of Bose condensate states is mathematically equivalent to the Werner's setting of unstructured states $\cS_{\cH}$, by simply replacing $\cH$ with $\bC^n$ and $N$ with $qN$. However, from a physical point of view, states in $\cS_{\mr{cond}}$ are structured, relative to the general $q$-boson states on $n$ modes $\cS_{q-\mr{bos}}$. In particular, the fidelity of $N\rarr N+1$ Bose condensate cloning, as follows from $\cS_{\mr{cond}}$ being a resolving orbit and from \cref{thm:structured_optimal_cloning}, equals
\begin{align}
    f_{\cS_{\mr{cond}}} &= \f{\dim(\Sym^{qN}(\bC^{n}))}{\dim(\Sym^{q(N+1)}(\bC^{n}))} \\ &= \f{\binom{n+qN-1}{n-1}}{\binom{n+q(N+1)-1}{n-1}} \\ &= \prod_{j=1}^{q} \f{qN+j}{n+qN+j-1}
\end{align}
which is much larger than that for general $q$-boson states,
\begin{align}
    f_{\cS_{q-\mr{bos}}} &= \f{\dim(\Sym^{N}(\Sym^{q}(\bC^{n})))}{\dim(\Sym^{N+1}(\Sym^{q}(\bC^{n})))} \\ &= \f{\binom{D+N-1}{N}}{\binom{D+N}{N+1}} \\ &= \f{N+1}{D+N} = \f{N+1}{\binom{n+q-1}{q}+N},
\end{align}
where $D=\dim \Sym^{q}(\bC^{n}))=\binom{n+q-1}{q}$. In particular, for condensate states, $f_{\cS_{\mr{cond}}}=1-\epsilon$ can be achieved starting from $N_{\mr{min}}=\Theta(n/\epsilon)$ copies, independent of $q$. In contrast, for general $q$-boson states, near-exact cloning requires $N_{\mr{min}}=\Theta(\binom{n+q-1}{q}/\epsilon)$, which for constant $q$ scales as $N_{\mr{min}}=\Theta(n^q/\epsilon)$, and for $q=\Theta (n)$ (constant particle density) as $N_{\mr{min}}=2^{\Theta(n)}/\epsilon$. Thus states with Bose-condensate structure yields an exponential or large polynomial (depending on the scaling of $q$) advantage in cloning sample complexity.

%% file: AppCartanProduct.tex
\section{Cartan product theorem}
\label{app:cartan}

This appendix clarifies the statement of the Cartan product theorem, used in the main text to prove irreducibility of some $N$-copy spans. Let us first recall some definitions from standard representation theory \cite{FultonHarris}. Given a Cartan subalgebra and a choice of positive roots of a semisimple Lie algebra $\mathfrak g$, a \emph{highest-weight vector} of a representation is a weight vector annihilated by all raising operators (the root vectors of the positive roots). We have

\begin{theorem}[Cartan product \cite{eastwood2005cartan}, {\cite[Obs.~13.2]{FultonHarris}}]
\label{thm:cartan_product}
    Let $V_1,V_2$ be finite-dimensional representations of a semisimple Lie algebra $\mathfrak g$, and let $v_1\in V_1$, $v_2\in V_2$ be highest-weight vectors of weights $\lambda_1,\lambda_2$. Then $v_1\otimes v_2$ is a highest-weight vector of $V_1\otimes V_2$, of weight $\lambda_1+\lambda_2$. Consequently, $v_1\otimes v_2$ generates an irreducible subrepresentation of $V_1\otimes V_2$ (the \emph{Cartan product}) \cite[Prop.~14.13]{FultonHarris}, whose isomorphism class is determined by its highest weight $\lambda_1+\lambda_2$ \cite[Thm.~14.18]{FultonHarris}.
\end{theorem}
The first claim is elementary, and explains why the highest-weight property survives tensoring: a raising operator $E$ acts on $V_1\otimes V_2$ by the Leibniz rule, $E\otimes\idx+\idx\otimes E$, and thus annihilates $v_1\otimes v_2$ whenever it annihilates each factor.
Iterating, if $v$ is a highest-weight vector of weight $\lambda$ in $V$, then $v^{\otimes N}$ is a highest-weight vector of weight $N\lambda$ in $V^{\otimes N}$, and it generates an irreducible subrepresentation.
In the main text the theorem is applied to representations of the connected compact groups $\mathrm{SU}(n)$ and $\Spin(2n)$; since these groups are connected, a subspace is invariant under the group if and only if it is invariant under the (complexified) Lie algebra, so the irreducibility statement transfers verbatim.

%% file: AppBosonicCloningProofs.tex
\section{Proofs for cloning of Bosonic Gaussian states}
\label{app:bosonic-cloning-proofs}

\subsection{Bosonic notation and passive linear optics}
\label{subsec:ab-proof-primer}

Each bosonic mode has an orthonormal photon-number basis $\{\ket m:m=0,1,\ldots\}$, and $\cF_n$ is the tensor product of $n$ such spaces.
The vacuum $\ket0$ contains no photons; creation and annihilation operators add and remove photons, with $[a_i,a_j^\dagger]=\delta_{ij}\bI$, $[a_i,a_j]=0$, and $a_j\ket0=0$.
The operator $\hat n=\sum_ja_j^\dagger a_j$ counts total photon number.

Displacements $D(\alpha)=\exp\bigl(\sum_j(\alpha_ja_j^\dagger-\overline{\alpha_j}a_j)\bigr)$ shift mean amplitudes and produce coherent states $\ket\alpha=D(\alpha)\ket0$.
Squeezing changes the vacuum fluctuations while retaining zero mean amplitudes.
Every pure Gaussian state is a displaced centered state $D(\alpha)\ket{\Omega_W}$, with $W$ describing squeezing and correlations between modes.

The Bargmann representation turns states into entire functions: $\ket{m_1,\ldots,m_n}$ corresponds to $\prod_j(z_j^{m_j}/\sqrt{m_j!})$, and the vacuum to $1$.
Creation acts by multiplication $a_j^\dagger=z_j$, annihilation by differentiation $a_j=\partial_{z_j}$, with inner product $\pi^{-n}\int_{\bC^n}\overline{f(z)}g(z)e^{-|z|^2}\,d^{2n}z$.
Centered pure Gaussian states become
\begin{equation}
\Omega_W(z)=\Delta(W)^{1/4}\exp\!\left(\tfrac12z^{\mathsf T}Wz\right),
\qquad\Delta(W)=\det(\bI-W^\dagger W).
\label{eq:ab-squeezed-vector}
\end{equation}
Here $W$ ranges over the Siegel disk $\mathbb D_n=\{W\in\bC^{n\times n}:W=W^{\mathsf T},\ \|W\|<1\}$, where $\|\cdot\|$ is the operator norm.
The norm condition ensures normalizability, and the determinant factor gives unit norm.

Passive linear optics mixes modes while conserving photon number.
Let $C_r$ mix the $r$ copies using the same real orthogonal matrix for each mode, with first row $(1,\ldots,1)/\sqrt r$.
Displacements combine in the first copy, while orthogonality preserves the squeezing form $\sum_cz_c^{\mathsf T}Wz_c$ in the copy variables $z_c$, giving
\begin{equation}
C_r\ket{D(\alpha)\Omega_W}^{\otimes r}
=\ket{D(\sqrt r\alpha)\Omega_W}\otimes\ket{\Omega_W}^{\otimes(r-1)}.
\label{eq:ab-copy-centre}
\end{equation}
For coherent inputs ($W=0$), the remaining copies are vacua.
Coherent states resolve the identity: $\int\ketbra\beta\,d^{2n}\beta=\pi^n\bI$.
Using $|\langle0|D(\beta)\Omega_W\rangle|=|\langle-\beta|\Omega_W\rangle|$ and setting $\beta=\sqrt r\alpha$ therefore gives
\begin{equation}
\int_{\bC^n}|\langle0|D(\sqrt r\alpha)\Omega_W\rangle|^2\,d^{2n}\alpha
=\left(\frac\pi r\right)^n.
\label{eq:ab-displacement-integral}
\end{equation}

\subsection{Coherent and canonical Gaussian cloners}
\label{subsec:ab-proof-canonical}

\begin{proof}[Proof of \cref{prop:ab-coherent}]
By \cref{eq:ab-copy-centre}, $C_r\cH_r^{\mathrm{coh}}=\cF_n\otimes\ket0^{\otimes(r-1)}$, with displacements acting through the irreducible Heisenberg representation.
The Heisenberg group with compact phase center is unimodular and amenable.
For Haar measure $d^{2n}\alpha$ times uniform probability on the phase circle, \cref{eq:ab-displacement-integral} at $W=0$ gives formal degree $d_{\mathrm{coh},r}=(r/\pi)^n$.
Thus \cref{thm:structured_optimal_cloning} gives fidelity $d_{\mathrm{coh},N}/d_{\mathrm{coh},M}=(N/M)^n$ and trace error $1-(N/M)^n$ over all channels.

To identify the canonical channel with the amplifier circuit, set $g=M/N$, $K=M-N$ and $J_r\ket\varphi=C_r^\dagger(\ket\varphi\otimes\ket0^{\otimes(r-1)})$.
Write $\tau_g=g^{-1}\sum_{\ell\geq0}(1-g^{-1})^\ell\ketbra\ell$ for the one-mode thermal state of mean occupation $g-1$.
The collective creation operator for mode $j$ is $M^{-1/2}\sum_{c=1}^M a_{j,c}^\dagger$.
Projecting the first $N$ copies onto vacuum removes their creation terms; the remaining sum is $\sqrt{K/M}$ times the collective creation operator for $K$ copies.
Thus vacuum projection of $J_M\ket{\boldsymbol\ell}$ leaves $(K/M)^{|\boldsymbol\ell|/2}J_K\ket{\boldsymbol\ell}$, where $|\boldsymbol\ell|$ is the total photon number.
Substituting this projection into \cref{eq:optimal_cloner} gives
\begin{equation}
J_M^\dagger T(\ketbra0^{\otimes N})J_M
=g^{-n}(1-g^{-1})^{\hat n}=\tau_g^{\otimes n}.
\label{eq:ab-coherent-canonical-thermal}
\end{equation}
A one-mode quantum-limited amplifier sends $\ket\beta$ to $D(\sqrt g\beta)\tau_gD(\sqrt g\beta)^\dagger$~\cite{braunstein2001coherentCloning}.
The passive input transformation $C_N$, amplification of each mode, and redistribution therefore reproduce the canonical vacuum output and, by displacement covariance, every coherent input output.
Equality on coherent-state projectors determines the channel: expanding the coherent-state diagonal as a power series recovers every matrix element of a bounded observable.
For $K=1$, solving $(N/(N+1))^n\geq1-\epsilon$ gives \cref{eq:ab-coherent-sample-complexity}.
With $L=-\log(1-\epsilon)$, the inequalities $n/L-1/2\leq(e^{L/n}-1)^{-1}\leq n/L$ give the lower bound used below and the uniform relative correction $O(\epsilon)$.
For the displacement orbit of a known pure Gaussian state $V\ket0$, use $V^\dagger D(\alpha)V=D(S\alpha)$ up to phase, with $S$ a real-linear bijection.
Conjugating the inputs and outputs by $V$ therefore identifies this promise with the coherent family and preserves both optima.
\end{proof}

Centered squeezed vacua form the vacuum orbit of $\operatorname{Mp}(2n,\bR)$; adding displacements gives the Jacobi group $\cG_n=\operatorname{Mp}(2n,\bR)\ltimes H_n$.
Here $H_n$ has its phase center reduced to a circle, so it has finite phase volume.

\begin{lemma}[Gaussian resolving orbits]
\label{lem:ab-orbit-degrees}
\label{thm:ab-resolving}
For every $r\geq1$, the metaplectic and Jacobi actions on $\cH_r^{\mathrm{sq}}$ and $\cH_r^{\mathrm{bg}}$ are irreducible.
Their square-integrable ranges are respectively $r\geq2n+1$ and $r\geq2n+2$.
Normalize compact vacuum-ray stabilizers to unit volume and use the invariant quotient measures
\begin{equation}
d\nu_{\mathrm{sq}}=\Delta(W)^{-(n+1)}dW,
\qquad d\nu_{\mathrm{bg}}=d\nu_{\mathrm{sq}}\,d^{2n}\alpha,
\label{eq:ab-orbit-measures}
\end{equation}
where $dW=\prod_{i\leq j}d\Re W_{ij}\,d\Im W_{ij}$.
Then their formal degrees are
\begin{align}
d_{\mathrm{sq},r}
&=\frac{\prod_{i\leq j}(r-i-j)}{2^n\pi^{n(n+1)/2}},
&&r\geq2n+1,
\label{eq:ab-centered-degree}\\
d_{\mathrm{bg},r}
&=\left(\frac r\pi\right)^nd_{\mathrm{sq},r-1},
&&r\geq2n+2.
\label{eq:ab-full-degree}
\end{align}
\end{lemma}

\begin{proof}
To prove irreducibility, let $P$ project onto an invariant subspace of the copy span.
It commutes with photon-number rotations, so it sends the unique joint vacuum to a multiple of itself; $P^2=P$ makes that multiple zero or one.
Commutation with the group propagates this action to every vacuum-orbit vector, whose closed span is the copy space.
Thus $P=0$ or $P=\bI$.
The groups are unimodular; the disk measure is symplectically invariant, and the Jacobi action also preserves displacement volume because $\det S=1$.
Thus square-integrability is equivalent to finiteness of the vacuum overlap integral.
For centered states, $|\langle0|\Omega_W\rangle|^2=\Delta(W)^{1/2}$, so the reciprocal formal degree is obtained from Hua's integral~\cite{hua1963harmonic}:
\begin{equation}
\int_{\mathbb D_n}\Delta(W)^{s/2-(n+1)}dW
=\frac{2^n\pi^{n(n+1)/2}}{\prod_{i\leq j}(s-i-j)},
\qquad s>2n.
\label{eq:ab-hua}
\end{equation}
At $s=r$, this proves \cref{eq:ab-centered-degree}.
Since $0<\Delta(W)\leq1$, the integrand increases as $s$ decreases; the pole at $s=2n$ therefore implies divergence for every $s\leq2n$.
The passive transformation $C_r$ factors the displaced vacuum overlap into one displaced copy and $r-1$ centered copies.
By \cref{eq:ab-displacement-integral} and Tonelli's theorem, its integral is
\begin{equation}
\int|\langle0|D(\alpha)\Omega_W\rangle|^{2r}d\nu_{\mathrm{bg}}
=\left(\frac\pi r\right)^n\int_{\mathbb D_n}\Delta(W)^{(r-1)/2-(n+1)}dW,
\label{eq:ab-displaced-degree-reduction}
\end{equation}
Hua's integral now applies with $s=r-1$, giving the displaced degree and the threshold $r\geq2n+2$.
In these ranges, \cref{thm:structured_optimal_cloning} supplies the normal canonical channels, their covariant optimality and their constant fidelity $d_N/d_M$.
For centered states, substitution gives
\begin{equation}
f_{\mathrm{can}}^{\mathrm{sq}}(N,M)=\prod_{1\leq i\leq j\leq n}\frac{N-i-j}{M-i-j},
\qquad N\geq2n+1,
\label{eq:ab-centered-canonical}
\end{equation}
and the full-family formulas \cref{eq:ab-full-factorization,eq:ab-full-canonical} follow analogously.
By \cref{eq:canonical-cloner-exact-trace-error}, the canonical trace error equals one minus its fidelity.
\end{proof}

\paragraph{The affine model.}
\label{app:ab-affine-model}

We identify the copy spans for $r>n-1$ with one representation, turning cloning into an overlap of vacuum vectors.
The subgroup $S_n$ defined in \cref{eq:ab-affine-schrodinger} has multiplication law $(A,B)(A',B')=(AA',B+AB'A^{\mathsf T})$.
Its action on the vacuum gives a Gaussian with complex symmetric $n\times n$ matrix $Z=AA^{\mathsf T}-iB$, whose Bargmann parameter is $W=(\bI-Z)(\bI+Z)^{-1}$.
The bijection between $\Re Z>0$ and $\mathbb D_n$, together with the unique upper-triangular factorization $\Re Z=AA^{\mathsf T}$, proves simple transitivity on squeezed vacua.
The group is amenable because it is solvable.

For the common model, use the cone $\Omega_n$ of positive-definite real symmetric $n\times n$ matrices $X$, with coordinate measure $dX=\prod_{i\leq j}dX_{ij}$.
Put $d\mu_n(X)=\det(X)^{-(n+1)/2}dX$ and $\cK_n=L^2(\Omega_n,d\mu_n)$.
The congruence Jacobian $(\det A)^{n+1}$ makes $\mu_n$ invariant.
Hence $R(A,B)f(X)=e^{i\Tr(BX)}f(A^{\mathsf T}XA)$ defines a unitary representation.
Normalize the vacuum vector $\xi_s$ using the cone gamma integral~\cite{FarautKoranyi1994} for $s>(n-1)/2$:
\begin{equation}
\int_{\Omega_n}e^{-\Tr(YX)}\det(X)^s\,d\mu_n(X)=\Gamma_n(s)\det(Y)^{-s},\qquad
\xi_s(X)=\left[\frac{2^{ns}}{\Gamma_n(s)}\right]^{1/2}\det(X)^{s/2}e^{-\Tr X},
\label{eq:ab-cone-model}
\end{equation}
where $\Gamma_n(s)=\pi^{n(n-1)/4}\prod_{j=1}^n\Gamma(s-(j-1)/2)$, and $Y=Y^{\mathsf T}$ is a complex $n\times n$ matrix with $\Re Y>0$.
The determinant branch is continuous and positive for real $Y$, and $Y=2\bI$ gives $\|\xi_s\|=1$.

\begin{lemma}[Translations and the common copy model]
\label{lem:ab-cone-commutant}
\label{lem:ab-common-copy-model}
The translations $R(\bI,B)$ generate the multiplication algebra $L^\infty(\Omega_n)$, every almost-everywhere nonzero vector is cyclic for them, and $R$ is irreducible.
For every integer $r>n-1$, the assignment
\begin{equation}
\Xi_r:\cH_r^{\mathrm{sq}}\longrightarrow\cK_n,\qquad
\Xi_r U_g^{\otimes r}\ket0^{\otimes r}=R(g)\xi_{r/2}
\label{eq:ab-gram-unitary}
\end{equation}
extends to an intertwining unitary.
\end{lemma}
\begin{proof}
The translations multiply by the characters $e^{i\Tr(BX)}$.
If $h\in L^1(\mu_n)$ integrates to zero against every character, the finite measure $h\,d\mu_n$, extended by zero to all real symmetric matrices, has vanishing Fourier transform.
Fourier uniqueness gives $h=0$, so the characters span a weak-star dense subspace of $L^\infty(\Omega_n)$.
If $v\perp R(\bI,B)f$ and $f\ne0$ almost everywhere, Cauchy--Schwarz gives $\overline vf\in L^1(\mu_n)$, and the same argument proves $v=0$.
The multiplication algebra is maximal abelian, so its commutant also consists of multiplication operators.
Commutation with congruences makes the multiplier invariant: in Cholesky coordinates $X=C^{\mathsf T}C$, congruence is right translation of the triangular group, with $\mu_n$ giving right Haar measure.
Fubini's theorem, applied to $(C,A)\mapsto(C,CA)$, makes the invariant multiplier constant almost everywhere.
This proves irreducibility.

For $g=(A,B)$, put $Y_g=\bI+AA^{\mathsf T}-iB$.
Integrating the Schr\"odinger vacuum $\pi^{-n/4}e^{-|q|^2/2}$ and using \cref{eq:ab-cone-model} give
\begin{equation}
\langle0|U_g|0\rangle=\left[\frac{2^n\det A}{\det Y_g}\right]^{1/2},\qquad
\langle\xi_s,R(g)\xi_s\rangle=\left[\frac{2^n\det A}{\det Y_g}\right]^s.
\label{eq:ab-cone-coefficient}
\end{equation}
At $s=r/2$, the cone coefficient is the $r$th power of the physical one, giving identical Gram matrices because every inner product depends only on $g^{-1}h$.
The physical orbit spans $\cH_r^{\mathrm{sq}}$ by transitivity, giving an intertwining isometry in \cref{eq:ab-gram-unitary}.
Its closed range is all of $\cK_n$, since the positive vector $\xi_{r/2}$ is cyclic for translations.
\end{proof}

Only the vacuum vector depends on the copy number, so the affine cloner is implemented by $V_{\mathrm{aff}}=\Xi_M^{-1}\Xi_N$.
Its constant fidelity is $|\langle\xi_{M/2},\xi_{N/2}\rangle|^2$, giving \cref{eq:ab-affine-fidelity} by the gamma integral and the trace error by purity.
A competing covariant isometry satisfies $R(g)^\dagger VR(g)=\chi(g)V$ for a character $\chi$, because isometries defining the same channel differ by a phase.
Translations are commutators, since conjugating $(\bI,B)$ by $(\sqrt2\,\bI,0)$ doubles $B$.
Thus $\chi$ is trivial on translations and $V$ is multiplication by a phase, which cannot increase the overlap of the positive vacuum vectors by the triangle inequality.

\begin{lemma}[Uniform affine estimates]
\label{lem:ab-affine-uniform}
Let $N\geq n$, $K\geq1$, $t=nK^2/(8N)$, and $\bar N=N-n+1$.
At $(N,N+K)$,
\begin{align}
\frac{t}{1+K/N}&\leq-n\log f_{\mathrm{aff},1}\leq-\log f_{\mathrm{aff},n}\leq t\left(\frac N{\bar N}+\frac{2N}{\bar N^2}\right),
\label{eq:ab-affine-log-bound}\\
\frac{t}{1+K/N+t}&\leq1-f_{\mathrm{aff},1}^n\leq1-f_{\mathrm{aff},n}\leq\min\left\{1,t\left(\frac N{\bar N}+\frac{2N}{\bar N^2}\right)\right\}.
\label{eq:ab-affine-uniform}
\end{align}
For $N\geq2n$, every $f$ between $f_{\mathrm{aff},n}$ and $f_{\mathrm{aff},1}^n$ satisfies
\begin{equation}
1-f=t+R,\qquad |R|\leq\frac{n^2K^4}{N^2},\qquad
|R|\leq\frac{n^2K^2}{N^2}\quad(K\leq4).
\label{eq:ab-affine-asymptotic}
\end{equation}
\end{lemma}
\begin{proof}
Logarithms reduce the finite-copy errors to curvature estimates for each gamma factor.
Put $x_j=(N-j+1)/2$, $h=K/4$, and $\psi_1=(\log\Gamma)''$.
The gamma product becomes a sum of second differences:
\begin{equation}
-\log f_{\mathrm{aff},n}=\sum_{j=1}^n\delta(x_j),\qquad
\delta(x)=\int_0^h\!\int_0^h\psi_1(x+u+v)\,du\,dv.
\label{eq:ab-affine-second-difference}
\end{equation}
The series $\psi_1(x)=\sum_{k\geq0}(x+k)^{-2}$ implies monotonicity and $x^{-1}\leq\psi_1(x)\leq x^{-1}+x^{-2}$.
Hence $\delta(x_1)\geq K^2/[8(N+K)]$, $\delta(x_j)\leq h^2(x_j^{-1}+x_j^{-2})$, and $\delta(x_j)\geq\delta(x_1)$.
Summing and using $2x_j\geq\bar N$ proves \cref{eq:ab-affine-log-bound}.
Applying $L/(1+L)\leq1-e^{-L}\leq\min\{1,L\}$ to the logarithmic losses gives \cref{eq:ab-affine-uniform}, whose outer bounds also control every intervening fidelity $f$.
For $N\geq2n$, $N/\bar N+2N/\bar N^2-1\leq8n/N$; subtracting $t$ from the infidelity bounds gives $-t(K/N+t)\leq R\leq8nt/N$.
The inequalities $K/N\leq8t$ and $8nt/N\leq64t^2$ give the first remainder bound.
For $K\leq4$, the sharper estimate $K/N+t\leq(2n+4)/N\leq8n/N$ gives the second.
\end{proof}

At $N=2n+1$ and $K=1$, the canonical product telescopes to $2^{-n}$, while \cref{eq:ab-affine-uniform} gives $1-f_{\mathrm{aff},n}\leq n(n+4)/[8(n+2)^2]<1/8$.
Thus the canonical trace error is at least $1/2$, whereas the affine fidelity exceeds $7/8$ and its trace error is below $1/\sqrt8$.

\subsection{One-mode optimality}

\begin{lemma}[Covariant reduction]
\label{lem:ab-covariant-reduction}
Let an amenable locally compact group $G$ act through strongly continuous unitary representations on separable input and output spaces, with unit-vector orbits $\psi_g=U_{\mathrm{in}}(g)\psi_0$ and $\varphi_g=U_{\mathrm{out}}(g)\varphi_0$.
Every normal completely positive trace-nonincreasing map $\Phi$ admits a contraction $V:\cH_{\mathrm{in}}\to\cH_{\mathrm{out}}\otimes\cE$ and a strongly continuous unitary representation $W$ on a separable environment such that
\begin{equation}
VU_{\mathrm{in}}(g)=(U_{\mathrm{out}}(g)\otimes W(g))V,
\qquad
\| (\bra{\varphi_0}\otimes\bI)V\psi_0\|^2
\geq\inf_g\langle\varphi_g,\Phi(\ketbra{\psi_g})\varphi_g\rangle.
\label{eq:ab-one-mode-intertwiner}
\end{equation}
An isometric intertwiner defines a channel with this overlap at every orbit point.
\end{lemma}
\begin{proof}
Average $\Phi^*$ with an invariant mean as in the proof of \cref{thm:structured_optimal_cloning}.
The bounded continuous integrands give a covariant completely positive subunital map $\overline\Lambda$, whose target expectation is at least the worst orbit fidelity.
This averaging need not preserve normality, so we cannot directly regard $\overline\Lambda$ as the adjoint of a channel.
Instead, restrict it to compact operators, the norm closure of finite-rank operators, which retain the target projector and therefore the fidelity test we need.

Every nondegenerate representation of this algebra is a multiple of its defining representation, so the minimal Stinespring dilation has the form $\overline\Lambda(X)=V^\dagger(X\otimes\bI)V$ on a separable environment.
Finite-rank projections $P_j\uparrow\bI$ give $V^\dagger V=\operatorname*{s-lim}_j\overline\Lambda(P_j)\leq\bI$, proving the contraction property without normality of $\overline\Lambda$.
A contraction accommodates trace loss, as occurs when we compress the output to the target copy span.

The covariant Stinespring construction~\cite{Scutaru1979} turns covariance of the restricted map into the intertwining relation $VU_{\mathrm{in}}(g)=(U_{\mathrm{out}}(g)\otimes W(g))V$.
The representation $W$ is strongly continuous because the input action is strongly continuous and the output action is norm-continuous on compact operators.
Evaluating the dilation at the target projector proves the fidelity inequality.
Conversely, intertwining gives $(\bra{\varphi_g}\otimes\bI)V\psi_g=W(g)(\bra{\varphi_0}\otimes\bI)V\psi_0$, so an isometry gives constant fidelity.
\end{proof}

\begin{proof}[Proof of \cref{thm:ab-one-mode}]
In the one-mode common model, $\cK=L^2(\bR_+,dx/x)$ and $\xi_s(x)=c_sx^{s/2}e^{-x}$ with $c_s^2=2^s/\Gamma(s)$.
Set $p=N/2$, $q=M/2$, $r=(p+q)/2$ and $f_0=\Gamma(r)^2/[\Gamma(p)\Gamma(q)]$.
For $K\leq4$, we resolve the environment into frequency sectors; for $K>4$, an orbit resolution bounds all environments at once.
With $a=A^2$ and $b=B$, define the positive and negative frequency representations
\begin{equation}
(U_\pm(a,b)f)(x)=e^{\pm ibx}f(ax).
\label{eq:ab-one-mode-model}
\end{equation}
The copy spans carry $U_+$.
Output compression preserves target fidelities and gives a trace-nonincreasing map, so amenability and \cref{lem:ab-covariant-reduction} reduce the upper bound to contractions intertwining $U_+$ with $U_+\otimes W$.

\emph{The branch $K\leq4$.}
Write $W(1,b)=e^{ibP}$.
Since $W(a,0)PW(a,0)^\dagger=aP$ with $a>0$, the positive, negative and zero spectral subspaces of $P$ are invariant under the group.
The affine-group representation classification gives $W_\pm=U_\pm\otimes\bI_{\cM_\pm}$, while $W_0(a,b)=a^{iT}$ for a self-adjoint $T$~\cite{GelfandNaimark1947,folland2016course}.

Put $d\nu(t)=dt/[t(1-t)]$ on $(0,1)$.
In the nonzero sectors use $(z,t)=(x+y,x/(x+y))$ and $(z,t)=(x-y,(x-y)/x)$, respectively; the input's positive frequency requires $x>y$ in the negative sector.
The measures become $(dz/z)d\nu(t)$ and $(dz/z)t\,d\nu(t)$, and the group acts only on $z$.
Covariance therefore fixes the input dependence, leaving only the splitting profiles displayed below.
Irreducibility of $U_+$ and Schur's lemma consequently give
\begin{equation}
\begin{aligned}
(V_+f)(x,y)&=f(x+y)h_+(x/(x+y)),\\
(V_-f)(x,y)&=\boldsymbol1_{\{x>y\}}\sqrt{x/(x-y)}\,f(x-y)h_-((x-y)/x),\\
(V_0f)(x)&=x^{-iT}f(x)\omega,
\end{aligned}
\label{eq:ab-sector-intertwiners}
\end{equation}
with $h_\pm\in L^2(d\nu)\otimes\cM_\pm$, $\omega\in\cE_0$ and $V^\dagger V=(\|h_+\|^2+\|h_-\|^2+\|\omega\|^2)\bI\leq\bI$.
The square-root factor compensates for the negative sector's additional density $t$.
For $A_\sigma=(\bra{\xi_q}\otimes\bI)V_\sigma\xi_p$, substituting $x=yt/(1-t)$ or $x=y/(1-t)$ gives
\begin{equation}
A_\pm(y)=c_pc_qy^r\int_0^1t^{\gamma_\pm}(1-t)^{-r}
 e^{-y(1+t)/(1-t)}h_\pm(t)d\nu(t),
\quad \gamma_+=q/2,\quad\gamma_-=(p+1)/2.
\label{eq:ab-sector-amplitude}
\end{equation}
Integrating its squared norm over $dy/y$, using $\int_0^\infty y^{2r-1}e^{-\kappa y}dy=\Gamma(2r)\kappa^{-2r}$ and $\kappa=2(1-tu)/[(1-t)(1-u)]$, yields $\|A_\pm\|^2=\langle h_\pm,(J_{\gamma_\pm}\otimes\bI)h_\pm\rangle$.
The calculation first holds for profiles supported away from the endpoints and extends by the kernel bound and density.
The kernel is
\begin{equation}
J_\gamma(t,u)=\frac{\Gamma(2r)}{\Gamma(p)\Gamma(q)}
\frac{(tu)^\gamma[(1-t)(1-u)]^r}{(1-tu)^{2r}}.
\label{eq:ab-beta-kernel}
\end{equation}
Choose the positive weight $w(t)=t^{r-\gamma}$ so that the powers of $u$ combine into Euler's beta integral $\int_0^1u^{r-1}(1-u)^{r-1}(1-tu)^{-2r}du=B(r,r)(1-t)^{-r}$, where $B(a,b)=\Gamma(a)\Gamma(b)/\Gamma(a+b)$.
This gives $J_\gamma w=f_0t^{2\gamma-r}w$.
Since $0<t<1$, the condition $2\gamma\geq r$ makes the remaining factor at most one.
The weighted Schur test converts $J_\gamma w\leq f_0w$ into $\|J_\gamma\|\leq f_0$; the weight itself need not lie in $L^2$.
The positive sector satisfies this condition for every $K$, while the negative sector requires $K\leq4$; hence $\|A_\pm\|^2\leq f_0\|h_\pm\|^2$ in this branch.
The zero sector satisfies
\begin{equation}
\|A_0\|^2=\frac{\|\Gamma(r-iT)\omega\|^2}{\Gamma(p)\Gamma(q)}\leq f_0\|\omega\|^2,
\label{eq:ab-zero-sector-bound}
\end{equation}
because Euler's integral gives $|\Gamma(r-i\tau)|\leq\Gamma(r)$.
Orthogonality and the contraction condition give the upper bound $f_0$, attained by the identity isometry on $\cK$ since $|\langle\xi_q,\xi_p\rangle|^2=f_0$.

\emph{The branch $K>4$.}
We now bound the target effect directly, independently of the environment's frequency sectors.
Let $D$ be multiplication by $x^{-1}$ and use left Haar measure $dg=da\,db/a^2$.
Plancherel in $b$ followed by the substitution $u=ax$ gives
\begin{equation}
\int_{S_1}|\langle f,U_+(g)v\rangle|^2dg
=2\pi\|f\|^2\|D^{1/2}v\|^2.
\end{equation}
For bounded $B\geq0$, let $\tau_D(B)=\sum_j\|D^{1/2}B^{1/2}e_j\|^2$ in an orthonormal basis, allowing infinity.
Tonelli's theorem extends the identity to $\int U_+(g)BU_+(g)^\dagger dg=2\pi\tau_D(B)\bI$ as positive forms.
Take $B=V^\dagger(\ketbra{\xi_q}\otimes\bI)V$, the effect whose expectation is the fidelity.
Covariance cancels the environmental unitaries against $\bI$, and $\|D^{1/2}\xi_q\|^2=2/(q-1)$ gives $2\pi\tau_D(B)\bI=[4\pi/(q-1)]V^\dagger V$.
The contraction condition therefore bounds its weighted trace by $\tau_D(B)\leq2/(q-1)$.
Insert $D^{1/2}$ and $D^{-1/2}$ in each inner product of $\langle\xi_p,B\xi_p\rangle=\sum_j|\langle B^{1/2}e_j,\xi_p\rangle|^2$ and apply Cauchy--Schwarz to obtain
\begin{equation}
\langle\xi_p,B\xi_p\rangle
\leq\|D^{-1/2}\xi_p\|^2\tau_D(B)
\leq\frac p2\frac2{q-1}=\frac p{q-1}.
\end{equation}
All weighted-domain conditions follow from the finite trace and the moment $\|D^{-1/2}\xi_p\|^2=p/2$.
For attainment, set $s=q-p-2>0$ and define
\begin{equation}
(V_-f)(x,y)=\boldsymbol1_{\{x>y\}}f(x-y)
\frac{(y/x)^{s/2}(1-y/x)^{p/2}}{\sqrt{B(s,p+1)}}.
\label{eq:ab-negative-channel}
\end{equation}
This intertwines $U_+$ with $U_+\otimes U_-$.
The coordinates $z=x-y$, $u=y/x$ give $dx\,dy/(xy)=(dz/z)(du/u)$, so its norm is fixed by $\int_0^1u^{s-1}(1-u)^pdu=B(s,p+1)$ and $V_-$ is an isometry.
The target projection sends $f$ to $[c_q/(c_p\sqrt{B(s,p+1)})]y^{s/2}e^{-y}\langle x\xi_p,f\rangle$, since substituting $z=x-y$ cancels the powers of $x$.
Integrating this gamma density gives the rank-one effect
\begin{equation}
V_-^\dagger(\ketbra{\xi_q}\otimes\bI)V_-
=\frac4{p(q-1)}\ketbra{x\xi_p}.
\end{equation}
Since $\langle\xi_p,x\xi_p\rangle=p/2$, discarding $y$ and returning through $\Xi_M^{-1}$ attains fidelity $p/(q-1)=N/(M-2)$.
Together the two branches prove the theorem and the arbitrary-environment estimate
\begin{equation}
\| (\bra{\xi_q}\otimes\bI)V\xi_p\|^2
\leq f_{\mathrm{sq},1}^\star(N,M)\|V\|^2,
\label{eq:ab-one-mode-dilation-bound}
\end{equation}
where scaling gives the same statement for bounded intertwiners.
\end{proof}

The metric comparison \cref{eq:ab-metric-comparison} bounds every channel's trace error below by its infidelity and bounds the fidelity-optimal channel's trace error above by $\sqrt{1-f_{\mathrm{sq},1}^\star}$.
For $N\geq3$, the canonical channel supplies the additional upper bound $K/(M-2)$, proving \cref{eq:ab-one-mode-trace-bracket}.
At $K=1$ and $N\geq3$, \cref{eq:ab-affine-uniform} gives $1/(8N+9)\leq\epsilon_{\mathrm{sq},1}^\star(N,N+1)\leq1/(N-1)$, proving $\Theta(1/\epsilon)$ sample complexity as $\epsilon\to0$.

\paragraph{Multimode optimality and sample bounds.}
\label{app:ab-multimode}

To prove \cref{eq:ab-centered-upper}, restrict an arbitrary channel to the product orbit of $S_1^{\times n}$ and compress its output to this orbit's copy span.
Restriction can only increase worst-case fidelity; compression preserves the orbit fidelities.
The spans are $\cK^{\otimes n}$, with vacuum vectors $\xi_p^{\otimes n}$ and $\xi_q^{\otimes n}$, where $p=N/2$, $q=M/2$, and $\cK=\cK_1$.
By \cref{lem:ab-covariant-reduction}, it suffices to bound a contraction $V$ intertwining the product representations with an arbitrary environment representation $W$.
Write $f_1=f_{\mathrm{sq},1}^\star(N,M)$ and set $V_1v=(\bra{\xi_q}\otimes\bI)V(\xi_p\otimes v)$.
For fixed $v$, $f\mapsto V(f\otimes v)$ has norm at most $\|V\|\|v\|$ and intertwines the first factor, including the remaining output modes in its environment.
The one-mode estimate \cref{eq:ab-one-mode-dilation-bound} therefore gives $\|V_1\|^2\leq f_1\|V\|^2$.
Contracting the first-mode vectors preserves the remaining intertwinings, so induction yields
\begin{equation*}
\bigl\|(\bra{\xi_q}^{\otimes n}\otimes\bI)V\xi_p^{\otimes n}\bigr\|^2\leq f_1^n\|V\|^2\leq f_1^n.
\end{equation*}
For $K\leq4$ and $N\geq n$, the one-mode optimum is affine, so $f_{\mathrm{aff},n}\leq f_{\mathrm{sq},n}^\star\leq f_{\mathrm{aff},1}^n$.
Applying \cref{lem:ab-affine-uniform} gives the finite bracket in \cref{eq:ab-affine-uniform} and, for $N\geq2n$, \cref{eq:ab-centered-optimal-asymptotic}.

For a canonical product $f=\prod_{\ell=1}^p(N-a_\ell)/(N+K-a_\ell)$, put $t=pK/N$ and assume $0\leq a_\ell\leq A<N$.
Expanding $f^{-1}$ and using $1-\prod(1-y_\ell)\leq\sum y_\ell$ give
\begin{align}
\frac{t}{1+t}&\leq1-f\leq\frac{pK}{N+K-A},
\label{eq:ab-canonical-uniform}\\
-t^2&\leq1-f-t\leq4t^2\quad(A\leq2p,\ t\leq1/4).
\label{eq:ab-canonical-uniform-remainder}
\end{align}
Indeed, $f^{-1}\geq1+t$ and $y_\ell=K/(N+K-a_\ell)$; the remainder uses $A/N\leq2t$ and $t/(1-2t)\leq t+4t^2$.
For centered states, $p=n(n+1)/2$; combining this expansion with \cref{eq:ab-centered-optimal-asymptotic} shows that, for $1\leq K\leq4$ and $n^2/N\to0$, the unrestricted optimal infidelity is smaller than the canonical one by a factor $4(n+1)/K\,[1+o(1)]$.
Since the canonical trace error equals $1-f$, at $K=1$ and $0<\epsilon<1$ these give necessary and sufficient conditions, respectively.
Inverting them gives
\begin{equation}
\frac{p(1-\epsilon)}\epsilon\leq N_{\mathrm{can}}(\epsilon)\leq\left\lceil\frac p\epsilon+A-1\right\rceil.
\label{eq:ab-canonical-sample-range}
\end{equation}
The sufficient integer exceeds $A$, since $p/\epsilon>1$.
Substituting $(p,A)=(n(n+1)/2,2n)$ or $(n(n+3)/2,2n+1)$ proves the centered or displaced canonical upper bound, respectively.
For the full family, $A/p\leq3/2$, so $N_{\mathrm{can}}=p(1+\theta\epsilon)/\epsilon$ with $-1\leq\theta\leq3/2$, uniformly in $n$.

For centered states, trace error at most $\epsilon$ requires fidelity at least $1-\epsilon$.
At $K=1$, the one-mode case of \cref{eq:ab-affine-log-bound}, valid for every $N\geq1$, and \cref{eq:ab-centered-upper} give $1-\epsilon\leq e^{-n/[8(N+1)]}$ and hence $N\geq n/[8(-\log(1-\epsilon))]-1$.
For $\epsilon\leq1/2$, combining $-\log(1-\epsilon)\leq2\epsilon$ with $N\geq1$ gives $N\geq n/(32\epsilon)$.
For $N\geq2n$, $N/\bar N+2N/\bar N^2\leq6$ bounds the affine trace error by $\sqrt{3n/(4N)}$.
Choosing $N=\lceil\max\{2n,3n/(4\epsilon^2)\}\rceil$ proves its sample upper bound.
Together with the canonical bound, these prove \cref{eq:ab-centered-sample-bounds}; restriction to coherent states gives the full-family lower bound in \cref{eq:ab-gaussian-sample-complexity}.

For arbitrary pure Fock-space states, discard $M-N-1$ outputs and restrict to a $D$-dimensional subspace.
Output compression, completed by preparing a fixed state on the complementary outcome, cannot decrease fidelity.
The $N\to N+1$ fidelity bound $(N+1)/(N+D)$~\cite{werner1998optimal} then vanishes as $D\to\infty$, so \cref{eq:ab-metric-comparison} gives trace error one.

%% file: AppDimensionLawProofs.tex
\section{Proofs and verification of the dimension law}
\label{app:dimension-law-proofs}

\subsection{Proof of the dimension law}
 
\begin{proof}[Proof of \cref{thm:uniform_dimension_law_fN}]
We record two consequences of $N\geq N_0$ used below: first, $N\geq256\kappa^2m$; second,
\begin{align}
\label{eq:N_threshold}
    u_0N\;\geq\;8\kappa m+8\log(2N).
\end{align}
Indeed, at $N=N_0$ we have $u_0N_0=128\kappa m$, while $8\log(2N_0)=8\log(8192\kappa^2m)\leq120\kappa m$ for $\kappa,m\geq1$; and $N\mapsto u_0N-8\log(2N)$ is increasing for $N\geq8/u_0=256\kappa$, which $N_0$ exceeds.
 
The proof has three steps. We first confine the integrals to the chart by a radial cutoff, and show that what is cut off is negligible (using \ref{assum:A3}, \ref{assum:A4}). We then derive the Gaussian moment identity \cref{eq:gaussian_moment} with correction terms, by integrating a total divergence, and bound the corrections (using \ref{assum:A1}, \ref{assum:A2}). Finally, we pass from the mean energy to the infidelity.
 
\emph{Step 1: cutoff and tail.}
On the chart $B_R$ we use the radial cutoff function
\begin{align}
\label{eq:cutoff}
    \chi(\vec z)\defeq
    \begin{cases}
        1, & r^2\leq R^2/2,\\
        2\lrp{1-r^2/R^2}, & R^2/2<r^2<R^2,
    \end{cases}
\end{align}
which is continuous, vanishes on the boundary of $B_R$, and satisfies $\vec z\cdot\nabla\chi=-4r^2/R^2\in[-4,0]$ on the transition region $R^2/2<r^2<R^2$ where it is not constant. It is only piecewise $C^1$, which causes no difficulty in the integrations by parts below: $\chi$ is continuous across the sphere $r^2=R^2/2$, so the interface terms from the two pieces cancel. Let
\begin{align}
\label{eq:Jk}
    q_N(\vec z)\defeq e^{-Nh(\vec z)}w(\vec z),
    \qquad
    J_k\defeq\int_{B_R}\chi\,r^{2k}\,q_N\,\dd^m\vec z\quad(k=0,1,2),
\end{align}
so that $J_0$ is the weight of the cut-off distribution and $J_1/J_0$, $J_2/J_0$ are its moments of $r^2$. Everything outside the support of $\chi$ is well separated from $\psi_0$: on the transition region, $h\geq\f78r^2>\f78\cdot\f{R^2}{2}=\f74u_0$, and outside $U$ we have $h\geq u_0$ by \ref{assum:A3}. Since $L'\leq\kappa m/u_0=32\kappa^2m\leq N$, we can write $F^N=F^{N-L'}F^{L'}\leq e^{-u_0(N-L')}F^{L'}$ on these regions, whence both the discarded weight and the weight of the transition region are bounded by
\begin{align}
\label{eq:discarded_mass}
    0\;\leq\;I_N-J_0\;\leq\;A_N,
    \qquad
    \int_{R^2/2<r^2<R^2}q_N\,\dd^m\vec z\;\leq\;A_N,
    \qquad
    A_N\defeq e^{-u_0(N-L')}I_{L'}.
\end{align}
We compare $A_N$ with $J_0$. On the ball $r^2\leq u_0$ we have $\chi=1$, $h\leq(1+\kappa u_0)r^2=\f{33}{32}r^2$ by \ref{assum:A1}, and $w(\vec z)\geq w(\vec 0)e^{-\kappa mr^2/2}$ by \ref{assum:A2} integrated along the ray. Therefore, with $b\defeq\f{33}{32}N+\f{\kappa m}{2}$,
\begin{align}
\label{eq:J0_lower}
    J_0\;\geq\;w(\vec 0)\int_{r^2\leq u_0}e^{-br^2}\,\dd^m\vec z\;\geq\;\f12\,w(\vec 0)\,\pi^{m/2}b^{-m/2},
\end{align}
where the second inequality says that the ball $r^2\leq u_0$ carries at least half of the Gaussian weight: on $r^2>u_0$, $e^{-br^2}\leq e^{-bu_0/2}e^{-br^2/2}$, so the fraction of the weight outside the ball is at most $e^{-bu_0/2}\int e^{-br^2/2}/\int e^{-br^2}=2^{m/2}e^{-bu_0/2}<\f12$, because $bu_0\geq u_0N\geq8\kappa m\geq8m$ by \cref{eq:N_threshold}.
Combining \cref{eq:A4,eq:discarded_mass,eq:J0_lower} and $b\leq\f{17}{16}N$ (which follows from $N\geq16\kappa m$), the ratio $\tau_N\defeq A_N/J_0$ obeys
\begin{align}
    \log\tau_N
    \;\leq\;\log2-u_0N+2\kappa m+\f m2\log\f{2u_0b}{m}
    \;\leq\;\log2+m\lrb{-x+2\kappa+\f12\log\f{17x}{8}}
    \;\leq\;\log2-\f{u_0N}{2},
\end{align}
where $x\defeq u_0N/m\geq8\kappa$ by \cref{eq:N_threshold}, and in the last step we used $2\kappa\leq x/4$ and $\f12\log(17x/8)\leq x/4$ for $x\geq8$. By the logarithmic part of \cref{eq:N_threshold}, $e^{-u_0N/2}\leq(2N)^{-4}$, and hence
\begin{align}
\label{eq:tau}
    \tau_N\;=\;\f{A_N}{J_0}\;\leq\;\f{1}{N^2}.
\end{align}
That is, the weight that the cutoff discards, relative to the weight it keeps, is negligible on the scale $m/N$ of the effect we are after.
 
\emph{Step 2: the Gaussian moment identity with corrections.}
For $k\in\{0,1\}$ we integrate the divergence of the vector field $\chi\,\vec z\,r^{2k}q_N$ over $B_R$. There is no boundary term, since $\chi$ vanishes on $\partial B_R$, and using $\nabla\cdot(\vec z\,r^{2k})=(m+2k)r^{2k}$ and $\vec z\cdot\nabla q_N=(-N\,\vec z\cdot\nabla h+\vec z\cdot\nabla\log w)\,q_N$ we obtain
\begin{align}
\label{eq:ibp}
    N\int\chi\,r^{2k}\,(\vec z\cdot\nabla h)\,q_N
    \;=\;
    (m+2k)J_k
    +\int\chi\,r^{2k}\,(\vec z\cdot\nabla\log w)\,q_N
    +\int r^{2k}\,(\vec z\cdot\nabla\chi)\,q_N .
\end{align}
For $h=r^2$, $w=1$, $\chi=1$ and $k=0$ this reads $2N\int r^2e^{-Nr^2}=m\int e^{-Nr^2}$, i.e., it is the Gaussian moment \cref{eq:gaussian_moment}, evaluated by integrating the total divergence $\nabla\cdot(\vec z\,e^{-Nr^2})$ rather than directly. The point of this detour is that, unlike the Gaussian integral, the divergence identity makes sense for any $h$ and $w$, and it then produces one correction term for each effect neglected in the Gaussian approximation: the radial derivative $\vec z\cdot\nabla h$ replaces $2h$ (anharmonicity), the varying density $w$ contributes the second term on the right (density drift), and the cutoff contributes the third. By \ref{assum:A1}, \ref{assum:A2} and $\kappa R^2=\f18$, on $B_R$ we have
\begin{align}
\label{eq:score_bounds}
    \vec z\cdot\nabla h\;\geq\;\f{15}{8}r^2,
    \qquad
    |\vec z\cdot\nabla\log w|\;\leq\;\kappa mr^2,
    \qquad
    |2h-\vec z\cdot\nabla h|\;\leq\;3\kappa r^4 .
\end{align}
We first extract the moments of $r^2$ from \cref{eq:ibp}. The cutoff term is nonpositive, so dropping it and inserting the first two bounds of \cref{eq:score_bounds} gives $\lrp{\f{15}{8}N-\kappa m}J_{k+1}\leq(m+2k)J_k$. As $\kappa m\leq N/256$, this yields
\begin{align}
\label{eq:moment_bounds}
    \Ex_\chi[r^2]=\f{J_1}{J_0}\;\leq\;\f{4m}{7N},
    \qquad
    \Ex_\chi[r^4]=\f{J_2}{J_0}\;\leq\;\f{16\,m(m+2)}{49\,N^2}\;\leq\;\f{48\,m^2}{49\,N^2},
\end{align}
where $\Ex_\chi$ denotes the average with respect to the cut-off distribution $\chi q_N/J_0$. Now we take $k=0$ in \cref{eq:ibp}, divide by $J_0$, and replace $N\vec z\cdot\nabla h$ by $2Nh$ using the third bound in \cref{eq:score_bounds}; the cutoff term is at most $4A_N$ in absolute value by \cref{eq:discarded_mass}. This gives the Gaussian moment identity with controlled corrections:
\begin{align}
\label{eq:mean_h}
    \abs{\,2N\,\Ex_\chi[h]-m\,}
    \;\leq\;\kappa m\,\Ex_\chi[r^2]+3\kappa N\,\Ex_\chi[r^4]+4\tau_N
    \;\leq\;\lrp{\f47+\f{144}{49}}\f{\kappa m^2}{N}+\f{4}{N^2}
    \;\leq\;\f{4\kappa m^2}{N},
\end{align}
by \cref{eq:moment_bounds,eq:tau}.
 
\emph{Step 3: from the mean energy to the infidelity.}
Since $h\geq0$, we have $0\leq h-(1-e^{-h})\leq h^2/2$, while $h\leq\f98r^2$ on $B_R$ by \ref{assum:A1}, so that $\Ex_\chi[h^2]\leq\f{81}{64}\Ex_\chi[r^4]<\f54\f{m^2}{N^2}$ by \cref{eq:moment_bounds}. Together with \cref{eq:mean_h},
\begin{align}
\label{eq:cutoff_answer}
    \abs{\,\f{2N}{m}\,\Ex_\chi\big[1-e^{-h}\big]-1\,}
    \;\leq\;\f{4\kappa m}{N}+\f{2N}{m}\cdot\f12\Ex_\chi[h^2]
    \;\leq\;\f{4\kappa m}{N}+\f{5m}{4N}\;\leq\;\f{6\kappa m}{N}.
\end{align}
It remains to remove the cutoff. By \cref{eq:tilted_expectation}, $1-f_N$ is the average of $1-F$ under $p_N$, whereas \cref{eq:cutoff_answer} concerns the average $a_N\defeq\Ex_\chi[1-e^{-h}]$ under the cut-off distribution. Writing $D_N\defeq I_N-J_0$ for the discarded weight and $B_N$ for the contribution of the discarded region to $\int(1-F)F^N\dd\nu$, we have $0\leq B_N\leq D_N\leq A_N$ (as $0\leq1-F\leq1$) and $0\leq a_N\leq1$, hence
\begin{align}
    \abs{(1-f_N)-a_N}
    \;=\;\f{\abs{B_N-a_ND_N}}{J_0+D_N}
    \;\leq\;\f{D_N}{J_0}\;\leq\;\tau_N\;\leq\;\f{1}{N^2}.
\end{align}
Combining this with \cref{eq:cutoff_answer} and using $m,\kappa\geq1$ proves \cref{eq:main_sandwich}: $\f{2N}{m}\cdot\f1{N^2}=\f{2}{mN}\leq\f{2\kappa m}{N}$.
\end{proof}
 
\begin{proof}[Proof of \cref{cor:uniform_dimension_law_complexity}]
The trace-distance error of the cloner $T_N$ is exactly $1-f_N$: the target $\psi^{\otimes N+1}$ is an eigenvector of the output $T_N(\psi^{\otimes N})$ with eigenvalue $f_N$, and the rest of the output is supported orthogonally to it. Thus $T_N$ is $\epsilon$-approximate if and only if $1-f_N\leq\epsilon$. Moreover, $f_N$ is nondecreasing in $N$: writing $F^{N+1}=F^{N/2}F^{(N+2)/2}$, the Cauchy--Schwarz inequality in $L^2(\nu)$ gives $I_{N+1}^2\leq I_NI_{N+2}$, i.e.\ $f_N\leq f_{N+1}$. Hence $\{N: 1-f_N\leq\epsilon\}$ is an interval $[N_\epsilon,\infty)$.
 
By \cref{eq:main_sandwich}, for $N\geq N_0$ we have $\f{m}{2N}(1-\f{8\kappa m}{N})\leq1-f_N\leq\f{m}{2N}(1+\f{8\kappa m}{N})$. Let $N^\pm$ be the solutions of $\f{m}{2N}(1\pm\f{8\kappa m}{N})=\epsilon$; then $N^\pm=\f{m}{2\epsilon}\big(1\pm8\kappa m/N^\pm\big)=\f{m}{2\epsilon}\big(1\pm16\kappa\epsilon+O(\kappa^2\epsilon^2)\big)$. Let $\epsilon\leq\epsilon_0\defeq m/(8N_0)=2^{-15}\kappa^{-2}$. Then $N^-\geq N_0$, and $1-f_N\leq\epsilon$ fails for all $N<N^-$ (for $N_0\leq N<N^-$ because the lower bound is decreasing in $N$ there, and for $N<N_0$ by monotonicity of $f_N$ together with $1-f_{N_0}\geq\f{m}{4N_0}=2\epsilon_0$) and holds for all $N\geq N^+$; so $N^-\leq N_\epsilon\leq N^++1$, which is \cref{eq:Nmin}. Finally, $N_{\mr{min},\cS}(\epsilon)\leq N_\epsilon$ because $T_{N_\epsilon}$ is an $\epsilon$-approximate cloner, and if $T_N$ is optimal for every $N$ then no $\epsilon$-approximate cloner exists for $N<N_\epsilon$.
\end{proof}

\subsection{Verification of the assumptions for the Gaussian families}
\label{sec:verifications}
 
We now verify \ref{assum:A0}--\ref{assum:A4} for the three Gaussian families of this work, $\Sffq$, $\Sff$ and $\Sbg$, and also for the unstructured family $\cS_\cH$ of \cite{werner1998optimal}, which fits the same scheme. All four admit the same kind of explicit description: a chart by a complex matrix $Z$, in which the overlap with the reference state is a power of the determinant $\det(\one\pm Z^\dag Z)$, and so is the measure $\nu$. We isolate this structure, because a single lemma then verifies the assumptions for all families at once.
 
\begin{definition}[Determinantal chart]\label{def:det_chart}
Let $\cS$, $\nu$ and $\psi_0$ be as in \cref{sec:dimension_law_setup}. A \emph{determinantal chart of type $(\sigma,\zeta,\gamma,k)$}, with a sign $\sigma\in\{+1,-1\}$, $\zeta,\gamma>0$ and an integer $k\geq0$, consists of a complex linear space $V$ of (in general rectangular) matrices with $\dim_\bR V=m-2k$, and a bijection $(Z,\vec\beta)\mapsto\psi_{Z,\vec\beta}$ from
\begin{align}\label{eq:det_domain}
    \cD\times\bC^k,
    \qquad
    \cD\defeq\{Z\in V:\;\one+\sigma Z^\dag Z\succ0\},
\end{align}
onto the set $\{\psi\in\cS:F(\psi)>0\}$ of states not orthogonal to $\psi_0$, with $\psi_{0,\vec 0}=\psi_0$, such that
\begin{align}\label{eq:det_overlap}
    h(\psi_{Z,\vec\beta})=\sigma\zeta\log\det\lrp{\one+\sigma Z^\dag Z}+|\vec\beta|^2,
    \qquad
    d\nu(\psi_{Z,\vec\beta})=\det\lrp{\one+\sigma Z^\dag Z}^{-\gamma}\,\dd^{m-2k}\vec z\,\dd^{2k}\vec\beta .
\end{align}
Here $\vec z\in\bR^{m-2k}$ are real linear coordinates on $V$, normalized so that $|\vec z|^2=\zeta\Tr Z^\dag Z$, and $\vec\beta\in\bC^k$ is identified with $\bR^{2k}$.
\end{definition}
In words: the overlap is a power of a determinant in $Z$ and, if $k>0$, exactly Gaussian in the additional $k$ complex coordinates $\vec\beta$ (which will be the displacements of the bosonic Gaussian states); the invariant measure is another power of the same determinant. For $\sigma=+1$ the domain is $\cD=V$, for $\sigma=-1$ it is the matrix ball $\|Z\|_{\mr{op}}<1$. Note that $\{F=0\}$ need not be small: for the free fermions it contains a whole connected component of $\cS$ (see below). This does no harm, as $\{F=0\}$ contributes to no $I_N$. The normalization $|\vec z|^2=\zeta\Tr Z^\dag Z$ of the coordinates is chosen so that $h=|\vec z|^2+|\vec\beta|^2+O(|\vec z|^4)$, as \ref{assum:A1} requires; the overall normalization of $\nu$ is chosen so that $w(\vec 0)=1$, a constant factor in $\nu$ cancels from $f_N$ anyway.
 
\begin{lemma}[Determinantal charts satisfy the assumptions]\label{lem:det_chart}
Let $\cS$ have a determinantal chart of type $(\sigma,\zeta,\gamma,k)$, and let $\kappa\geq1$ satisfy
\begin{align}\label{eq:K_conditions}
    \kappa\;\geq\;\f{4}{\zeta}
    \qquad\text{and}\qquad
    \kappa\,m\;\geq\;\f{3\gamma}{\zeta}\,.
\end{align}
Then \ref{assum:A0}--\ref{assum:A4} hold with constant $\kappa$ and reserve exponent $L'=16\kappa m-\sigma\gamma/\zeta$, provided $L'\geq L$, the threshold above which the moments $I_t$ are assumed finite (\cref{sec:dimension_law_setup}).
\end{lemma}
 
\begin{proof}
Write $r_Z\defeq|\vec z|$ and $r_\beta\defeq|\vec\beta|$, so that $r^2=r_Z^2+r_\beta^2$, and let $h_Z\defeq\sigma\zeta\log\det(\one+\sigma Z^\dag Z)$ be the $Z$-part of $h$. Take $U\defeq\{\psi_{Z,\vec\beta}:r<R\}$ with $R^2=1/(8\kappa)$. Since $r_Z^2=\zeta\Tr Z^\dag Z$, on the ball $B_R$ the operator norm $a\defeq\|Z^\dag Z\|_{\mr{op}}$ satisfies
\begin{align}\label{eq:A_small}
    a\;\leq\;\Tr Z^\dag Z\;=\;\f{r_Z^2}{\zeta}\;<\;\f{R^2}{\zeta}\;=\;\f{1}{8\kappa\zeta}\;\leq\;\f1{32},
\end{align}
so in particular $B_R$ is inside the domain of the chart, and \ref{assum:A0} holds with $w=\det(\one+\sigma Z^\dag Z)^{-\gamma}$.
 
\emph{\ref{assum:A1}.} By $\log\det=\Tr\log$ and the series of $\log(\one+\sigma Z^\dag Z)$, which converges on $B_R$ since $a<1$, $h_Z=\zeta\sum_{j\geq1}\f{(-\sigma)^{j-1}}{j}\Tr (Z^\dag Z)^j$, whose $j=1$ term is $r_Z^2$; the $\vec\beta$-part of $h$ is exactly $r_\beta^2$, so $h-r^2=h_Z-r_Z^2$. In terms of the eigenvalues $a_i\in[0,a]$ of $Z^\dag Z$, $\Tr(Z^\dag Z)^j=\sum_ia_i^j\leq a^{j-2}\sum_ia_i^2\leq a^{j-2}(\Tr Z^\dag Z)^2$ for $j\geq2$; summing the geometric series,
\begin{align}
    \big|h-r^2\big|\;\leq\;\f{\zeta(\Tr Z^\dag Z)^2}{2(1-a)}\;=\;\f{r_Z^4}{2\zeta(1-a)}\;\leq\;\f{16}{31}\cdot\f{r_Z^4}{\zeta}\;\leq\;\f{4}{31}\,\kappa\,r^4,
\end{align}
by $1/\zeta\leq\kappa/4$. For the radial derivative, differentiate along the ray $t\mapsto(tZ,t\vec\beta)$, which multiplies $Z^\dag Z$ by $t^2$:
\begin{align}
    \vec z\cdot\nabla h_Z=\f{\dd}{\dd t}h_Z(tZ)\Big|_{t=1}&=2\zeta\Tr\lrb{Z^\dag Z(\one+\sigma Z^\dag Z)^{-1}},\\
    \vec z\cdot\nabla h_Z-2r_Z^2&=-2\zeta\sigma\Tr\lrb{(Z^\dag Z)^2(\one+\sigma Z^\dag Z)^{-1}},\label{eq:radial_derivative}
\end{align}
while the $\vec\beta$-part contributes exactly $2r_\beta^2$. Hence $|\vec z\cdot\nabla h-2r^2|\leq\f{2\zeta\Tr(Z^\dag Z)^2}{1-a}\leq\f{2r_Z^4}{\zeta(1-a)}\leq\f{16}{31}\kappa r^4$. Both bounds are stronger than \cref{eq:A1}.
 
\emph{\ref{assum:A2}.} By \cref{eq:det_overlap}, $\log w=-\gamma\log\det(\one+\sigma Z^\dag Z)=-\f{\sigma\gamma}{\zeta}h_Z$, so \cref{eq:radial_derivative} gives $|\vec z\cdot\nabla\log w|=\f\gamma\zeta|\vec z\cdot\nabla h_Z|\leq\f\gamma\zeta\lrp{2r_Z^2+\kappa R^2r_Z^2}=\f{17}{8}\f{\gamma}{\zeta}r_Z^2\leq\kappa m\,r^2$ by the second condition in \cref{eq:K_conditions}.
 
\emph{\ref{assum:A3}.} For $\sigma=-1$, $-\log(1-a_i)\geq a_i$ for each eigenvalue $a_i$ of $Z^\dag Z$ gives $h_Z\geq r_Z^2$ on all of $\cD$, so $h\geq r^2\geq R^2=4u_0$ outside $U$. For $\sigma=+1$, $\det(\one+Z^\dag Z)\geq1+\Tr Z^\dag Z$ gives $h_Z\geq\zeta\log(1+r_Z^2/\zeta)$, increasing in $r_Z$. Outside $U$, $r_Z^2+r_\beta^2\geq R^2$, so either $r_\beta^2\geq R^2/2$, and $h\geq r_\beta^2\geq2u_0$; or $r_Z^2\geq R^2/2$, and $h\geq h_Z\geq\zeta\log\big(1+\f{R^2}{2\zeta}\big)\geq\f{R^2}{2}\big(1-\f{R^2}{4\zeta}\big)\geq\f{127}{64}u_0$, by $\log(1+x)\geq x-x^2/2$ and $R^2/\zeta\leq\f1{32}$. In both cases $h\geq u_0$, i.e.\ $F\leq e^{-u_0}$ outside $U$, and $F=0$ on the part of $\cS$ not covered by the chart.
 
\emph{\ref{assum:A4}.} By \cref{eq:det_overlap}, $F^{L'}\nu=e^{-(L'+\sigma\gamma/\zeta)h_Z}e^{-L'r_\beta^2}\dd^{m-2k}\vec z\,\dd^{2k}\vec\beta$: the measure contributes $\sigma\gamma/\zeta$ to the exponent of the overlap, and only the total exponent $L'+\sigma\gamma/\zeta$ matters; hence $L'$ is chosen as $t_0-\sigma\gamma/\zeta$ with $t_0\defeq16\kappa m=\f{m}{2u_0}$. The first condition in \cref{eq:A4} holds since $u_0L'=\f m2-\sigma u_0\gamma/\zeta\leq\f m2+\f{m}{96}\leq\kappa m$, using $\gamma/\zeta\leq\kappa m/3$. For the second, the $\vec\beta$-integral is exactly $(\pi/L')^k$, and $L'\geq t_0(1-\f{1}{48})$ since $\gamma/\zeta\leq\kappa m/3=t_0/48$. For the $Z$-integral, if $\sigma=-1$ then $h_Z\geq r_Z^2$ gives directly $\int_\cD e^{-t_0h_Z}\dd\vec z\leq\pi^{(m-2k)/2}t_0^{-(m-2k)/2}$. If $\sigma=+1$, then $h_Z\geq\zeta\log(1+r_Z^2/\zeta)$ gives, in polar coordinates with $s=r_Z^2/\zeta$ and $\ell\defeq(m-2k)/2$ (an integer, as $V$ is a complex space),
\begin{align}
    \int_\cD e^{-t_0h_Z}\dd\vec z\;\leq\;\int_{\bR^{2\ell}}\lrp{1+\f{r_Z^2}{\zeta}}^{-\zeta t_0}\dd\vec z
    \;=\;\pi^{\ell}\zeta^{\ell}\,\f{\Gamma(\zeta t_0-\ell)}{\Gamma(\zeta t_0)}
    \;=\;\pi^{\ell}\prod_{j=1}^{\ell}\lrp{t_0-\f{j}{\zeta}}^{-1}
    \;\leq\;\pi^{\ell}t_0^{-\ell}\lrp{\f{128}{127}}^{\ell},
\end{align}
where the middle step is the Beta integral $\int_0^\infty s^{\ell-1}(1+s)^{-\zeta t_0}\dd s=B(\ell,\zeta t_0-\ell)$ together with $\f12|S^{2\ell-1}|\Gamma(\ell)=\pi^\ell$, and the last step uses $j/\zeta\leq\ell/\zeta\leq\kappa m/8=t_0/128$. Altogether $I_{L'}\leq\pi^{m/2}t_0^{-m/2}\lrp{\f{128}{127}}^{\ell}\lrp{\f{48}{47}}^{k}\leq\pi^{m/2}t_0^{-m/2}e^{\kappa m}$, which is \cref{eq:A4}.
\end{proof}

It remains to exhibit a determinantal chart for our families of interest: coordinates $Z$ (and, for $\Sbg$, $\vec\beta$) covering the set $\{F>0\}$, in which both the overlap and the invariant measure take the form \cref{eq:det_overlap}, and to read off $(\sigma,\zeta,\gamma,k)$ and $m$. In each case, the chart is given by an unnormalized representative $\ket{\td\psi_Z}$ of the state, with $\braket{\psi_0}{\td\psi_Z}=1$, so that $F=\braket{\td\psi_Z}{\td\psi_Z}^{-1}$. In the compact cases---unstructured states, Slater states and free fermions---we also use that the group acts transitively on the family, so that the invariant probability measure is unique \cite{folland2016course}. The results are collected in \cref{tab:det_charts}.
 
\begin{table}[t]
\caption{Determinantal charts of the four families. The domain of $Z$ is all of $V$ for $\sigma=+1$ and the matrix ball $\|Z\|_{\mr{op}}<1$ for $\sigma=-1$.}
\label{tab:det_charts}
\begin{center}
\small
\begin{tabular}{llcccccc}
\hline
family & reference $\psi_0$ & $Z$ & $k$ & $\sigma$ & $\zeta$ & $\gamma$ & $m$\\
\hline
unstructured $\cS_{\cH}$ & $\ket{\psi_0}\in\cH=\bC^d$ & $\bC^{d-1}$ (column) & $0$ & $+$ & $1$ & $d$ & $2(d-1)$\\
Slater states $\Sffq$ & $a_1^\dag\cdots a_q^\dag\vac$ & $\bC^{(n-q)\times q}$ & $0$ & $+$ & $1$ & $n$ & $2q(n-q)$\\
free fermions $\Sff$ & $\vac$ & $\bC^{n\times n}$ antisym. & $0$ & $+$ & $\tfrac12$ & $n-1$ & $n(n-1)$\\
bosonic Gaussian $\Sbg$ & $\vac$ & $\bC^{n\times n}$ sym. & $n$ & $-$ & $\tfrac12$ & $n+\tfrac32$ & $n(n+3)$\\
\hline
\end{tabular}
\end{center}
\end{table}
 
\emph{Unstructured states.} For $\cH=\bC^d$ with reference state $\ket{\psi_0}$ and orthonormal basis $\ket{\psi_j}$, $0\leq j\leq d-1$, the chart is $\ket{\td\psi_{Z}}=\ket{\psi_0}+\sum_{j\geq1}Z_{j}\ket{\psi_j}$, $Z\in\bC^{d-1}$, which covers all states not orthogonal to $\ket{\psi_0}$; its squared norm is $1+|Z|^2=\det(\one+Z^\dag Z)$ for the column $Z$, so $F=\det(\one+Z^\dag Z)^{-1}$. It remains to find $\nu$, which is the unique $\mathrm U(d)$-invariant probability measure on $\cS_\cH$, in the coordinates $Z$. Writing $U=\begin{pmatrix}a&b\\c&D\end{pmatrix}$ with $a\in\bC$, $b$ a row, $c$ a column and $D\in\bC^{(d-1)\times(d-1)}$, one observes that its action in coordinates $Z$ is fractional-linear,
\begin{align}
\label{eq:mobius_unstructured}
    U\ket{\td\psi_Z}\propto\ket{\td\psi_{Z'}},\qquad Z'=(c+DZ)(a+bZ)^{-1}.
\end{align}
Under \cref{eq:mobius_unstructured}, $1+|Z'|^2=(1+|Z|^2)/|a+bZ|^2$ (the norm of $U\ket{\td\psi_Z}$ is that of $\ket{\td\psi_Z}$) while $\dd\vec z'=|a+bZ|^{-2d}\dd\vec z$, so $(1+|Z|^2)^{-d}\dd\vec z$ is invariant. By uniqueness of the invariant measure, $\nu\propto(1+|Z|^2)^{-d}\dd\vec z$. From the form of the overlap and the measure, we conclude that the chart $Z$ is indeed determinantal, with $(\sigma,\zeta,\gamma,k)=(1,1,d,0)$ and $m=2(d-1)$.
Thus Werner's cloner falls within the scope of the dimension law, which explains its sample complexity $\simeq d/\epsilon$ without reference to the symmetric subspace. Comparing with Werner's exact fidelity $f_N=\f{N+1}{N+d}$, i.e.\ $1-f_N=\f{m/2}{N+d}$: the leading term $m/(2N)$ of \cref{eq:main_sandwich} is exact, and the relative correction $-d/(N+d)$ lies within the bound $8\kappa m/N$.
 
\emph{Slater states} (\cref{prop:slater_orbit}). For $A\in\bC^{n\times q}$, let $\ket{\td\psi_A}\defeq b_1^\dag\cdots b_q^\dag\vac$ with $b_k^\dag\defeq\sum_{i=1}^nA_{ik}a_i^\dag$; then $\braket{\td\psi_A}{\td\psi_{A'}}=\det(A^\dag A')$. Take the reference state $\ket{\psi_0}=\ket{\td\psi_{A_0}}$, $A_0=\binom{\one_q}{0}$. Every Slater state not orthogonal to $\ket{\psi_0}$ is proportional to $\ket{\td\psi_Z}\defeq\ket{\td\psi_A}$ with $A=\binom{\one_q}{Z}$, for a unique $Z\in\bC^{(n-q)\times q}$; its overlap with $\ket{\psi_0}$ is $\det(A_0^\dag A)=1$ and its squared norm is $\det(A^\dag A)=\det(\one+Z^\dag Z)$, so $F=\det(\one+Z^\dag Z)^{-1}$. It remains to find $\nu$, the unique $\rU(n)$-invariant probability measure on $\Sffq$, in the coordinates $Z$. Writing $U\in\rU(n)$ as $\begin{pmatrix}a&b\\c&D\end{pmatrix}$ with $a\in\bC^{q\times q}$ and $D\in\bC^{(n-q)\times(n-q)}$, and using $Ub_k^\dag U^\dag=\sum_i(UA)_{ik}a_i^\dag$, one finds that the action of $U$ in the coordinates $Z$ is fractional-linear,
\begin{align}
\label{eq:mobius_slater}
    U\ket{\td\psi_Z}=\det(a+bZ)\ket{\td\psi_{Z'}},\qquad Z'=(c+DZ)(a+bZ)^{-1}.
\end{align}
Under \cref{eq:mobius_slater}, $\det(\one+Z'^\dag Z')=\det(\one+Z^\dag Z)/|\det(a+bZ)|^2$ while $\dd\vec z'=|\det(a+bZ)|^{-2n}\dd\vec z$, so $\det(\one+Z^\dag Z)^{-n}\dd\vec z$ is invariant. By uniqueness of the invariant measure, $\nu\propto\det(\one+Z^\dag Z)^{-n}\dd\vec z$. From the form of the overlap and the measure, we conclude that the chart $Z$ is determinantal, with $(\sigma,\zeta,\gamma,k)=(1,1,n,0)$ and $m=2q(n-q)$.
 
\emph{Free fermions} (\cref{prop:free_fermion_+_orbit}). The family $\Sff=\Sffp\cup\Sffm$ has two connected components, of the same dimension $n(n-1)$, exchanged by the parity-flipping element of $\Pin(2n)$. The integral \cref{eq:finite_degree} over $\Pin(2n)$ is the sum of the integrals over its two cosets of $\Spin(2n)$, and the coset which maps the even reference state $\psi_0=\vac$ into $\Sffm$ contributes nothing, since $\Sffm\perp\psi_0$. Thus $I_N$ is (one half of) the integral over $\Sffp$ against the $\Spin(2n)$-invariant measure, and in the assumptions the odd component is simply part of $\{F=0\}$. On $\Sffp$, by Thouless' theorem \cite{ring2004nuclear}, every state with $\braket{0}{\psi}\neq0$ is proportional to $\ket{\td\psi_Z}\defeq\exp\big(\tfrac12\sum_{ij}Z_{ij}a^\dag_ia^\dag_j\big)\vac$ for a unique antisymmetric $Z\in\bC^{n\times n}$, with $\braket{0}{\td\psi_Z}=1$ and, by Onishi's formula \cite{ring2004nuclear}, $\braket{\td\psi_Z}{\td\psi_Z}=\det(\one-\bar ZZ)^{1/2}=\det(\one+Z^\dag Z)^{1/2}$, using $\bar Z=-Z^\dag$; so $F=\det(\one+Z^\dag Z)^{-1/2}$ and $\zeta=\f12$. With this $\zeta$, $|\vec z|^2=\f12\Tr Z^\dag Z=\sum_{i<j}|Z_{ij}|^2$, so $\vec z$ consists of the real and imaginary parts of the independent entries $Z_{ij}$, $i<j$, and $m=n(n-1)$. It remains to find $\nu$, the unique $\Spin(2n)$-invariant probability measure on $\Sffp$, in the coordinates $Z$. Conjugation by a Gaussian unitary $\phi(g)$ acts on the Majorana operators by a rotation $R\in\mathrm{SO}(2n)$ (\cref{app:ff_background}); in terms of $\hat c\defeq(a_1,\dots,a_n,a_1^\dag,\dots,a_n^\dag)^\top$ this reads $\phi(g)\,\hat c\,\phi(g)^\dag=M\hat c$ with $M=\begin{pmatrix}A&B\\ \bar B&\bar A\end{pmatrix}$ unitary. Using this transformation law, one finds that the action of $\phi(g)$ in the coordinates $Z$ is fractional-linear,
\begin{align}
\label{eq:mobius_ff}
    \phi(g)\ket{\td\psi_Z}\propto\ket{\td\psi_{Z'}},\qquad Z'=(A-Z\bar B)^{-1}(Z\bar A-B).
\end{align}
Under \cref{eq:mobius_ff} and using Onishi's formula for the overlap \cite{ring2004nuclear}, $$\det(\one+Z'^\dag Z')^{1/2}=\det(\one+Z^\dag Z)^{1/2}/|\det(A-Z\bar B)|,$$ while $\dd\vec z'=|\det(A-Z\bar B)|^{-2(n-1)}\dd\vec z$, so $\det(\one+Z^\dag Z)^{-(n-1)}\dd\vec z$ is invariant. By uniqueness of the invariant measure, $\nu\propto\det(\one+Z^\dag Z)^{-(n-1)}\dd\vec z$. From the form of the overlap and the measure, we conclude that the chart $Z$ is determinantal, with $(\sigma,\zeta,\gamma,k)=(1,\f12,n-1,0)$ and $m=n(n-1)$.

\emph{Bosonic Gaussian states} (\cref{thm:ab-resolving}). The states of $\Sbg$ are $D(\vec\alpha)\ket{\Omega_W}$, where $\ket{\Omega_W}$ is the squeezed vacuum with Bargmann function $\det(\one-W^\dag W)^{1/4}\exp(\tfrac12z^\top Wz)$, $W$ symmetric with $\|W\|_{\mr{op}}<1$, and $D(\vec\alpha)$ a displacement by $\vec\alpha\in\bC^n$; the chart $(W,\vec\alpha)$ is global. The overlap with the vacuum is
\begin{align}\label{eq:gaussian_overlap}
    F(W,\vec\alpha)\;=\;\det(\one-W^\dag W)^{1/2}\,e^{-Q_W(\vec\alpha)},
    \qquad
    Q_W(\vec\alpha)\defeq|\vec\alpha|^2-\mathrm{Re}\lrp{\vec\alpha^\dag W\bar{\vec\alpha}},
\end{align}
by $\braket{0}{D(\vec\alpha)\Omega_W}=\braket{-\vec\alpha}{\Omega_W}=e^{-|\vec\alpha|^2/2}\,\Omega_W(-\bar{\vec\alpha})$ in the Bargmann representation. The quadratic form $Q_W$ is positive definite, since $|\vec\alpha^\dag W\bar{\vec\alpha}|\leq\|W\|_{\mr{op}}|\vec\alpha|^2$. The invariant measure of the Jacobi group, in the coordinates $(W,\vec\alpha)$, is $\nu\propto\det(\one-W^\dag W)^{-(n+1)}\dd W\,\dd^{2n}\vec\alpha$, the Siegel-disk measure times the flat measure on displacements (see \cref{eq:ab-orbit-measures}). The overlap \cref{eq:gaussian_overlap} is not yet of the form \cref{eq:det_overlap}, because $Q_W$ depends on $W$; but the $W$-dependent linear change of coordinates $\vec\beta\defeq Q_W^{1/2}\vec\alpha$ (the positive square root of the real $2n\times2n$ matrix of $Q_W$) makes $Q_W(\vec\alpha)=|\vec\beta|^2$ exact, at the cost of a Jacobian: in the Takagi basis, $Q_W=\sum_k\big[(1-d_k)x_k^2+(1+d_k)y_k^2\big]$ with $\vec\alpha_k=x_k+iy_k$, so $\dd^{2n}\vec\alpha=\prod_k(1-d_k^2)^{-1/2}\dd^{2n}\vec\beta=\det(\one-W^\dag W)^{-1/2}\dd^{2n}\vec\beta$. From the form of the overlap and the measure in the coordinates $(W,\vec\beta)$, we conclude that this chart is determinantal, with $(\sigma,\zeta,\gamma,k)=(-1,\f12,n+\f32,n)$ and $m=n(n+1)+2n=n(n+3)$; here $F=0$ nowhere. The reserve exponent of \cref{lem:det_chart} is $L'=16\kappa m+2n+3$, above the resolving threshold $L=2n+2$ of \cref{thm:ab-resolving}.

\emph{Applicability of \cref{lem:det_chart}.} Finally, we check \cref{eq:K_conditions}. The first condition requires $\kappa\geq4$ for $\cS_\cH$ and $\Sffq$, and $\kappa\geq8$ for $\Sff$ and $\Sbg$. For the second, $\f{3\gamma}{\zeta m}$ equals $\f{3d}{2(d-1)}\leq3$ for unstructured states ($d\geq2$); $\f{3n}{2q(n-q)}\leq3$ for Slater states, since $q(n-q)\geq n-1$ for $1\leq q\leq n-1$; $\f6n\leq3$ for free fermions ($n\geq2$); and $\f{3(2n+3)}{n(n+3)}\leq\f{15}4$ for bosonic Gaussian states ($n\geq1$). Hence
\begin{align}\label{eq:K_value}
    \kappa\;=\;8
\end{align}
satisfies \ref{assum:A0}--\ref{assum:A4} for every family in the table, uniformly in the number of modes $n$, in the particle number $q$ and in $d$.
By \cref{thm:uniform_dimension_law_fN,cor:uniform_dimension_law_complexity}, the canonical cloner for each of these families requires $\dim\cS/(2\epsilon)$ copies up to a relative correction $O(\epsilon)$, with constants independent of the dimension.
This is the unrestricted sample complexity when the canonical cloner is optimal, as it is for the compact families in the table.
For the full bosonic Gaussian family $\Sbg$, it gives an upper bound; whether the quadratic dependence on $n$ is necessary remains open.

%% file: AppGaussianHowe.tex
\section{Gaussian symmetry groups and Howe duality}
\label{app:ff_background}
\label{sec:howe-background}

This appendix collects the representation theory behind the fermionic and bosonic Gaussian families.
\Cref{app:spin_pin,app:slater} treat a single copy of the fermionic families.
We describe the role of $\Spin(2n)$ in describing fermionic Gaussian unitaries, and the $\Pin(2n)$ group, of which the free-fermion states form a single orbit. We also show that the number-conserving fermionic Gaussian unitaries realise $\rU(n)$, with the Slater states as an orbit.
These facts and the Cartan product theorem of \cref{app:cartan} underlie the proofs of \cref{prop:slater_orbit,prop:free_fermion_+_orbit}.
The remaining subsections treat $r$ copies, where Howe duality pairs the physical symmetry group with a group that mixes the copies: \cref{subsec:fermionic-skew-howe} for the fermionic families and \cref{subsec:bosonic-howe} for the bosonic ones (Bose condensates and squeezed vacua).
In each of the four cases of \Cref{tab:howe-summary}, the $r$-copy span is one irreducible representation of the physical group times a one-dimensional vector space on the copy side.
This picture is the Howe-duality view of the resolving-orbit property mentioned in the overview of results, and it supplies the data for the circuits of \cref{subsec:howe-cg-circuit}.

\subsection{Gaussian unitaries, parity sectors, and the groups \texorpdfstring{$\Spin(2n)$}{Spin(2n)} and \texorpdfstring{$\Pin(2n)$}{Pin(2n)}}
\label{app:spin_pin}

We will denote the mode operators $a_j\defeq\f12(c_{2j-1}-ic_{2j})$, so that $a^\dag_ja^{\phantom\dag}_j=n_j$. The fermionic Gaussian unitaries are those of the form $U_A=e^{\sum^{2n}_{\alpha,\beta=1} A_{\alpha\beta} c_{\alpha} c_\beta}$ with $A$ real antisymmetric; being even in the Majorana operators, they commute with the parity operator $\Par$. Conjugation by $U_A$ rotates the Majorana operators,
\begin{align}
\label{eq:majorana_conjugation}
    U_A^{\phantom\dag}\, c_\alpha\, U_A^\dag=\sum_\beta R_{\beta\alpha}\,c_\beta,\qquad R=e^{4A}\in\mathrm{SO}(2n).
\end{align}
The unitaries $\{U_A\}$ form a compact connected group, and the correspondence $U\mapsto R$ of \cref{eq:majorana_conjugation} is a two-to-one homomorphism onto $\mathrm{SO}(2n)$; this identifies the group with $\Spin(2n)$, the double cover of $\mathrm{SO}(2n)$, and we will use the name $\Spin(2n)$ for the group of unitaries itself. It acts on $\Hf$ by the $2^n$-dimensional \emph{spinor representation}, which we denote $\phi$ \cite{FultonHarris}\cite{lawson1989spin}\cite{helsen2022matchgate,chen2026gaussian}.
The parity subspaces $\Hfp$ and $\Hfm$ are invariant under $\Spin(2n)$; they carry the two \emph{half-spinor} representations of $\Spin(2n)$, which are irreducible and mutually inequivalent \cite{FultonHarris}.

The base states $\vac$ and $\ket{e}\defeq a^\dag_n\vac=\ket{0\cdots01}$ are free-fermion states of even and odd parity, respectively. Moreover, each fixed-parity family of free-fermion states is a single orbit of $\Spin(2n)$ \cite{chen2026gaussian},
\begin{align}
\label{eq:sector_orbits}
    \Sffpm=\left\{\phi(g)\,\chi_\pm\,\phi(g)^\dag \,:\, g\in\Spin(2n)\right\},\qquad \chi_+\defeq\ketbra{0},\quad\chi_-\defeq\ketbra{e}.
\end{align}
Since every Fock basis state $\ket{s}$ is itself a free-fermion state of parity $(-1)^{\sum_js_j}$ --- the non-degenerate ground state of a Hamiltonian $\sum_j\pm ic_{2j-1}c_{2j}$ with signs set by the occupations --- the two families span the entire parity sectors:
\begin{align}
\label{eq:sector_spans}
    \Span\lrp{\Sffpm}=\Hfpm.
\end{align}
To obtain the full family $\Sff=\Sffp\cup\Sffm$ as a single orbit of a group, we consider the group generated by $\Spin(2n)$ together with one additional, parity-flipping, unitary $\mu$. Specifically, we choose $\mu\defeq c_{2n-1}$, which is unitary since $\mu=\mu^\dag$ and $\mu^2=\bI$.
Being odd in the Majorana operators, $\mu$ anticommutes with $\Par$ and thus exchanges the parity sectors. It preserves the free-fermion family, mapping $\Sffp$ bijectively onto $\Sffm$; in particular, $\mu\vac=\ket{e}$, since $c_{2n-1}=a_n+a^\dag_n$.
It also normalizes $\Spin(2n)$ (conjugation by $\mu$ maps quadratic exponents to quadratic exponents), so that
\begin{align}
\label{eq:pin_group}
    \Pin(2n)\defeq\langle \Spin(2n),\mu\rangle=\Spin(2n)\,\cup\,\mu\, \Spin(2n)
\end{align}
is again a compact group of unitaries on $\Hf$. The name is justified since this group is a realization of the double cover $\Pin(2n)$ of the full orthogonal group $\rO(2n)$: conjugation by $\mu$ acts on the Majorana operators as the reflection with respect to the axis no.~$2n-1$, namely $\mu\,c_\alpha\,\mu=-c_\alpha$ for $\alpha\neq2n-1$ and $\mu\,c_{2n-1}\,\mu=c_{2n-1}$, which extends \cref{eq:majorana_conjugation} to a two-to-one homomorphism from this group onto $\rO(2n)$ \cite{lawson1989spin}.
We keep the name $\phi$ for the resulting representation of $\Pin(2n)$. Finally, combining \cref{eq:sector_orbits} with $\mu\vac=\ket e$, the full family $\Sff$ is a single orbit of $\Pin(2n)$ --- e.g., of the vacuum.

\subsection{Number-conserving Gaussian unitaries and Slater states}
\label{app:slater}

We now specialize to the fixed-particle-number family $\Sffq$ of \cref{eq:slater_def}, where $\hat N\defeq\sum^n_{j=1}n_j$ is the total particle number. These are exactly the \emph{Slater states}, and they form a single orbit of the number-conserving Gaussian unitaries, as we now explain.

The number-conserving (passive) fermionic Gaussian unitaries are those commuting with $\hat N$; they are of the form $U_h=e^{i\sum_{j,l} h_{jl}a^\dag_j a^{\phantom\dag}_l}$ with $h$ Hermitian, and conjugation by them rotates the modes,
\begin{align}
\label{eq:mode_rotation}
    U_h^{\phantom\dag}\, a^\dag_j\, U_h^\dag=\sum_k u_{kj}\, a^\dag_k,\qquad u=e^{ih}\in\rU(n).
\end{align}
Each fixed-number subspace $\cH_q\subset\Hf$, spanned by the Fock states $\ket s$ with $\sum_js_j=q$, is invariant. Identifying $\cH_q$ with the antisymmetric tensor power $\Lambda^q\bC^n$ via
\begin{align}
\label{eq:wedge_identification}
    \ket{s}\ \longleftrightarrow\ e_{j_1}\wedge\cdots\wedge e_{j_q},\qquad j_1<\cdots<j_q \text{ the occupied sites of } s,
\end{align}
the unitary $U_h$ acts on $\cH_q$ simply as $u^{\wedge q}$; in other words, the number-conserving Gaussian unitaries realize the representation $u\mapsto\Gamma(u)\defeq\bigoplus_q u^{\wedge q}$ of $\rU(n)$ on $\Hf$.
In this language, the Slater states are the normalized decomposable vectors $v_1\wedge\cdots\wedge v_q$, and
\begin{align}
\label{eq:slater_orbit}
    \Sffq=\left\{\Gamma(u)\,\ketbra{\psi_0}\,\Gamma(u)^\dag \,:\, u\in\rU(n)\right\},\qquad \ket{\psi_0}\defeq\ket{1\cdots10\cdots0}\ \longleftrightarrow\ e_{1}\wedge\cdots\wedge e_{q},
\end{align}
where $\ket{\psi_0}$ has the first $q$ sites occupied. Indeed, a decomposable vector depends (up to normalization) only on the span of $v_1,\dots,v_q$, so we may take them orthonormal, and any two orthonormal $q$-frames are related by some $u\in\rU(n)$; every Slater state belongs to $\Sffq$ (it is the non-degenerate ground state of a number-conserving quadratic Hamiltonian, obtained by rotating $\sum_j \eps_j\, n_j$ with $\eps_1=\dots=\eps_q<\eps_{q+1}=\dots=\eps_n$); and, conversely, every state in $\Sffq$ is a Slater state \cite{bravyi2005lagrangian,surace2022fermionic}.
Since every Fock basis state of $\cH_q$ is itself a Slater state, \cref{eq:wedge_identification}, the family spans the entire sector,
\begin{align}
\label{eq:slater_span}
    \Span\lrp{\Sffq}=\cH_q\cong\Lambda^q\bC^n.
\end{align}
As a representation of $\rU(n)$, $\Lambda^q\bC^n$ is irreducible --- it is the $q$-th \emph{fundamental representation} \cite{FultonHarris}.

\subsection{Double commutants, Schur--Weyl duality, and \texorpdfstring{$r$}{r}-copy spans}
\label{subsec:double-commutant-background}

We now move on to the setting of multiple copies and describe a dual relation between the groups which act on physical and on copy indices. We start with an abstract setting, later moving on to those needed for our fermionic and bosonic families. On $r$ copies, the diagonal action of a physical group commutes with a group that mixes the copies, and in all cases below each of the two generates the commutant of the other.
For a finite-dimensional Hilbert space $\cK$, two unital $*$-algebras $\cA,\cB\subseteq\operatorname{End}(\cK)$ form a \emph{mutual-commutant pair} if $\cA'=\cB$ and $\cB'=\cA$, where $\cX'$ denotes the operators commuting with every element of $\cX$.
By the double-commutant theorem and the structure theorem for $*$-algebras, the space and the two algebras then decompose, up to a unitary, as
\begin{align}
\cK
&\cong
\bigoplus_{\lambda\in\Lambda}
V_\lambda\otimes M_\lambda,
\nonumber\\
\cA
&\cong
\bigoplus_{\lambda\in\Lambda}
\operatorname{End}(V_\lambda)
\otimes
\bC\,\bI_{M_\lambda},
\nonumber\\
\cB
&\cong
\bigoplus_{\lambda\in\Lambda}
\bC\,\bI_{V_\lambda}
\otimes
\operatorname{End}(M_\lambda),
\label{eq:double-commutant-normal-form}
\end{align}
where the $V_\lambda$ are pairwise inequivalent irreducible representations of $\cA$, and the $M_\lambda$ of $\cB$, so no pair occurs twice.

The prototype is Schur--Weyl duality.
Here and below, $r$ denotes a generic number of copies; we reserve $N$ and $N+K$ for the input and output copy numbers of a cloner.
On $(\bC^d)^{\otimes r}$, the diagonal action of $\rU(d)$ and the action of $\rS_r$ that permutes the tensor factors generate mutual commutants, so
\begin{align}
(\bC^d)^{\otimes r}
\cong
\bigoplus_{\substack{\lambda\vdash r\\\ell(\lambda)\leq d}}
V_\lambda^{\rU(d)}
\otimes
M_\lambda^{\rS_r},
\label{eq:schur-weyl-duality}
\end{align}
where $V_\lambda^{\rU(d)}$ is the irreducible polynomial representation of $\rU(d)$ with highest weight $\lambda$, $M_\lambda^{\rS_r}$ is the Specht module labeled by the same partition $\lambda$ of $r$, and $\ell(\lambda)$ is the number of rows of $\lambda$.
The $r$-copy span of all pure states on $\bC^d$ is the symmetric subspace, the single block $\lambda=(r)$:
\begin{align}
\Sym^r(\bC^d)
&=
\lrp{(\bC^d)^{\otimes r}}^{\rS_r}
=
\Span\left\{\ket{\psi}^{\otimes r}:\ket{\psi}\in\bC^d\right\}
\nonumber\\
&\cong
V_{(r)}^{\rU(d)}
\otimes
\bC\ket{\mathbf 1},
\label{eq:schur-weyl-symmetric-sector}
\end{align}
where $\ket{\mathbf 1}$ spans the trivial representation of $\rS_r$; for $r=N+K$, $\pi_{N+K}$ in \cref{eq:optimal_cloner} projects onto this block, and \cref{thm:structured_optimal_cloning} reproduces Werner's cloner~\cite{werner1998optimal}.

Howe duality extends this pattern to the Gaussian families.
In each case below, a \emph{physical group} $G$ acts on every copy by $\phi$, and a \emph{copy group} $G'_r$ mixes the $r$ copies and commutes with $\phi^{\otimes r}$; the same separation of mode and copy indices appears in recent work on fermionic Gaussian tomography~\cite{chen2026gaussian}.
The $r$-copy space then decomposes as
\begin{align}
\cH^{\otimes r}
\cong
\bigoplus_{\lambda\in\Lambda_r^{\mathrm{occ}}}
V_\lambda^{G}
\otimes
M_{\vartheta_r(\lambda)}^{G'_r},
\label{eq:abstract-howe-decomposition}
\end{align}
where $V_\lambda^{G}$ and $M_{\vartheta_r(\lambda)}^{G'_r}$ are irreducible representations of $G$ and of $G'_r$, $\Lambda_r^{\mathrm{occ}}$ is the set of physical labels that occur, and the \emph{Howe correspondence} $\vartheta_r$ is a bijection between the physical and the copy-side labels that occur~\cite{Howe1989,FultonHarris}.
We write $V_\lambda$ for physical and $M_\eta$ for copy-side irreducible representations, and $\bigoplus$ denotes the Hilbert direct sum when infinitely many summands occur.

Multiplicity freeness determines the $r$-copy span of an orbit (\cref{def:ncopy_span}) as soon as the span is irreducible.

\begin{lemma}[Copy-side line]
\label{lem:copy-side-line}
Let $\cS$ be an orbit of pure states under a representation $\phi$ of $G$, and suppose that $\cH^{\otimes r}$ decomposes as in \cref{eq:abstract-howe-decomposition}, with $G$ acting by $\phi^{\otimes r}$.
If $G$ acts irreducibly on the $r$-copy span $\cH^{(r)}_{\cS}$, then there are a unique label $\lambda_r\in\Lambda_r^{\mathrm{occ}}$ and a unique line $\bC w_r\subset M^{G'_r}_{\vartheta_r(\lambda_r)}$ such that
\begin{align}
\cH^{(r)}_{\cS}
=
V_{\lambda_r}^{G}\otimes\bC w_r
\cong
V_{\lambda_r}^{G}.
\label{eq:copy-side-line}
\end{align}
Moreover, every $h\in G'_r$ that maps $\cH^{(r)}_{\cS}$ to itself acts on $\bC w_r$ by a scalar, so the subgroup of such $h$ acts on $\bC w_r$ by a character.
\end{lemma}

\begin{proof}
The span $\cH^{(r)}_{\cS}$ is invariant and irreducible under $G$, so it is contained in a single $G$-isotypic component of \cref{eq:abstract-howe-decomposition}.
Since the decomposition is multiplicity-free, this component is a single summand $V^G_\lambda\otimes M^{G'_r}_{\vartheta_r(\lambda)}$, which fixes $\lambda_r=\lambda$.
Write $M=M^{G'_r}_{\vartheta_r(\lambda_r)}$.
By Schur's lemma, $\operatorname{Hom}_G\lrp{V^G_{\lambda_r},V^G_{\lambda_r}\otimes M}\cong M$, the intertwiner attached to $w\in M$ being $v\mapsto v\otimes w$.
Hence every irreducible $G$-invariant subspace of $V^G_{\lambda_r}\otimes M$ has the form $V^G_{\lambda_r}\otimes\bC w$ for a unique line $\bC w\subset M$, which proves \cref{eq:copy-side-line}.
For the last claim, let $h\in G'_r$ map $\cH^{(r)}_{\cS}$ to itself.
Since $h$ commutes with $\phi^{\otimes r}(G)$, its restriction to $\cH^{(r)}_{\cS}$ lies in $\operatorname{End}_G(\cH^{(r)}_{\cS})=\bC\,\bI$ by Schur's lemma.
As $h$ acts on the summand as $\bI_{V_{\lambda_r}}\otimes h|_M$, it therefore multiplies $w_r$ by a scalar.
\end{proof}

\Cref{tab:howe-summary} lists $\lambda_r$ and the copy-side line for the four families treated below; in each case, we find the line from the action of $G'_r$ on $\ket{\psi_0}^{\otimes r}$, which generates $\cH^{(r)}_{\cS}$ by \cref{rmk:orbit_span}.

\begin{table}[t]
\centering
\small
\renewcommand{\arraystretch}{1.2}
\setlength{\tabcolsep}{8pt}
\begin{tabular}{@{}lllll@{}}
\hline
Family & $G$ & $G'_r$ & $\cH^{(r)}_{\cS}$ & Copy-side line $\bC w_r$\\
\hline
Slater states $\Sffq$ & $\rU(n)$ & $\rU(r)$ & $V^{\rU(n)}_{(r^q)}$ & the character $\det^q$\\
Free-fermion states $\Sffpm$ & $\Spin(2n)$ & $\rO(r)$ & $V_{r\alpha}$ and $V_{r\beta}$ & trivial ($+$) and $\det$ ($-$)\\
Bose condensates $\cS_{\mr{cond}}$ & $\rU(n)$ & $\rU(r)$ & $\Sym^{rq}\bC^n$ & $\bC\,x_1^q\cdots x_r^q\subset\Sym^{rq}\bC^r$\\
Squeezed vacua $\cS_{\mathrm{sq}}$ & $\operatorname{Mp}(2n,\bR)$ & $\rO(r)$ & $\Theta_{n,r}(\varnothing)$ & trivial\\
\hline
\end{tabular}
\caption{Howe-duality data of the $r$-copy spans, derived in \cref{subsec:fermionic-skew-howe,subsec:bosonic-howe}: the dual pair $(G,G'_r)$, the span $\cH^{(r)}_{\cS}$ as a representation of $G$, and the action of $G'_r$ on the line $\bC w_r$ of \cref{lem:copy-side-line}. Number-conserving symmetries pair with a unitary copy group, and symmetries that create pairs of particles with an orthogonal one.}
\label{tab:howe-summary}
\end{table}

\subsection{Fermionic skew Howe dualities}
\label{subsec:fermionic-skew-howe}

For fermions, the number-conserving fermionic Gaussian unitaries pair with a unitary copy group, $\rU(n)\times\rU(r)$, and the general ones with an orthogonal copy group, $\Spin(2n)\times\rO(r)$.
One point requires care: $\Hf^{\otimes r}$ is an ordinary tensor product, whereas its identification with the Fock space $\Lambda^\bullet(\bC^n\otimes\bC^r)$ of $nr$ modes respects only the parity grading.
Exchanging two copies of parities $p$ and $p'$ therefore produces the Koszul sign $(-1)^{pp'}$, and the copy generators carry Jordan--Wigner strings, \cref{eq:jw-copy-generators}; this sign decides whether $\bC w_r$ carries the trivial or the sign character of the copy permutations.

\paragraph{Unitary skew Howe duality.}
By \cref{eq:wedge_identification}, $\Hf\cong\Lambda^\bullet\bC^n$, and the graded identification $\Hf^{\otimes r}\cong\Lambda^\bullet(\bC^n\otimes\bC^r)$ gives the skew Howe decomposition
\begin{align}
\Hf^{\otimes r}
\cong
\Lambda^\bullet(\bC^n\otimes\bC^r)
\cong
\bigoplus_{\lambda\subseteq(r^n)}
V_{\lambda}^{\rU(n)}\otimes M_{\lambda^{\mathsf T}}^{\rU(r)},
\label{eq:fermionic-skew-howe}
\end{align}
where $\rU(n)$ acts by $\Gamma(u)^{\otimes r}$ (\cref{app:slater}), $(a^b)$ is the rectangular partition with $b$ rows of length $a$, $\lambda^{\mathsf T}$ is the transposed partition, and $V_\lambda^{\rU(n)}$ and $M_{\lambda^{\mathsf T}}^{\rU(r)}$ are the irreducible polynomial representations with these highest weights.
For the Slater states, $\ket{\psi_0}^{\otimes r}$ is a highest-weight vector of weight $r\omega_q$, as in the proof of \cref{prop:slater_orbit}, so $\lambda_r=(r^q)$, $\lambda_r^{\mathsf T}=(q^r)$, and
\begin{align}
\cH^{(r)}_{\Sffq}
&\cong
V_{(r^q)}^{\rU(n)},
&\qquad
M_{(q^r)}^{\rU(r)}
&\cong
\det\nolimits^q.
\label{eq:slater-howe-selection}
\end{align}
Thus the copy-side line is the whole one-dimensional representation $\det^q$, and $\cH^{(r)}_{\Sffq}$ is the $\det^q$-isotypic component of $\Hf^{\otimes r}$, of dimension given by \cref{eq:slater_dimension} at $N=r$.
On permutations, $\det^q$ restricts to $\operatorname{sgn}^q$, in agreement with the Koszul sign $(-1)^{q^2}=(-1)^q$ for exchanging two copies of parity $q\bmod2$.

\paragraph{Spin--orthogonal skew Howe duality.}
The general fermionic Gaussian unitaries act by the spinor representation $\phi$ of $\Spin(2n)$ (\cref{app:spin_pin}).
We state the highest-weight parametrization of type $D_n$ for $n\geq3$; the cases $n=1,2$ need separate conventions.
Let $\ell=\lfloor r/2\rfloor$.
For $r\geq3$, the copy labels form the set $\cO_{n,r}$ of integral dominant weights $\gamma=(\gamma_1,\ldots,\gamma_\ell)$ of $\mathfrak{so}(r)$ with $\gamma_1\leq n$, that is, of the integer vectors with
\begin{align}
\gamma_1\geq\cdots\geq\gamma_\ell\geq0,
\quad \gamma_1\leq n,
&\qquad r=2\ell+1,
\nonumber\\
\gamma_1\geq\cdots\geq\gamma_{\ell-1}\geq|\gamma_\ell|,
\quad \gamma_1\leq n,
&\qquad r=2\ell.
\label{eq:orthogonal-copy-labels}
\end{align}
For $r=1$ and $r=2$ we set $\cO_{n,1}=\{0\}$ and $\cO_{n,2}=\{\gamma\in\bZ:|\gamma|\leq n\}$.
For odd $r$, put $\eta(\gamma)=\gamma$ and $\delta_r=\tfrac12$, with $\eta(0)=\varnothing$ for $r=1$; for even $r$, put $\eta(\gamma)=(\gamma_1,\ldots,\gamma_{\ell-1},|\gamma_\ell|)$ and $\delta_r=0$.
The highest weight $\nu(\gamma)$ of $\Spin(2n)$ paired with $\gamma$, and the reflection $\tau$ of type $D_n$, are
\begin{align}
\nu_i(\gamma)
&=
\ell-\eta(\gamma)^{\mathsf T}_{n+1-i}+\delta_r,
\qquad i=1,\ldots,n,
\nonumber\\
\tau(\nu_1,\ldots,\nu_n)
&=
(\nu_1,\ldots,\nu_{n-1},-\nu_n),
\label{eq:spin-howe-label-map}
\end{align}
where $\eta^{\mathsf T}$ is the transposed partition, padded with zeros to $n$ parts.
In particular, $\nu(0)=\tfrac r2(1,\ldots,1)$, and for even $r$
\begin{align}
\nu_n(\gamma)=0
\iff
\gamma_\ell\neq0,
\label{eq:nu-n-vanishing}
\end{align}
because $\eta(\gamma)^{\mathsf T}_1$, the number of nonzero parts of $\eta(\gamma)$, equals $\ell$ exactly when $\gamma_\ell\neq0$.
Let
\begin{align}
\cV_\gamma^{\Spin(2n)}
\defeq
V_{\nu(\gamma)}^{\Spin(2n)}
\mathbin{\oplus_\tau}
V_{\tau\nu(\gamma)}^{\Spin(2n)},
\end{align}
where $\oplus_\tau$ keeps a single summand if $\tau\nu(\gamma)=\nu(\gamma)$, that is, if $\nu_n(\gamma)=0$.
Orthogonal skew Howe duality states that, as a representation of $\mathfrak{so}(2n)\oplus\mathfrak{so}(r)$ with the physical action integrated to $\Spin(2n)$,
\begin{align}
\Hf^{\otimes r}
&\cong
\bigoplus_{\gamma\in\cO_{n,r}}
\cV_\gamma^{\Spin(2n)}
\otimes
M_\gamma^{\mathfrak{so}(r)},
\label{eq:spin-orthogonal-howe}
\end{align}
where $V_\nu^{\Spin(2n)}$ and $M_\gamma^{\mathfrak{so}(r)}$ are the irreducible representations with the indicated highest weights~\cite{Hasegawa1989,Wenzl2012,Wenzl2020}.

The copy group is generated by the bilinears
\begin{align}
L_{ab}
\defeq
\tfrac i2\sum_{\alpha=1}^{2n}
\td c^{(a)}_\alpha\td c^{(b)}_\alpha,
\qquad
\td c^{(a)}_\alpha
\defeq
\Par^{\otimes(a-1)}\otimes c_\alpha\otimes\bI^{\otimes(r-a)}.
\label{eq:jw-copy-generators}
\end{align}
Their Jordan--Wigner strings make the $\td c^{(a)}_\alpha$, $a=1,\ldots,r$, satisfy the Majorana algebra on $\Hf^{\otimes r}$, so the $L_{ab}$ realize $\mathfrak{so}(r)$, which integrates to $\mathrm{SO}(r)$ because the weights $\gamma$ are integral.
The $L_{ab}$ commute with every $U_A^{\otimes r}$, since $U_A$ commutes with $\Par$ and rotates each $\td c^{(a)}_\alpha$ by the matrix $R$ of \cref{eq:majorana_conjugation}.
The extension to $\rO(r)$ is generated by the parity of the first copy,
\begin{align}
\Par^{(1)}\defeq\Par\otimes\bI^{\otimes(r-1)},
\label{eq:copy-reflection}
\end{align}
which commutes with $\phi^{\otimes r}(\Spin(2n))$ and implements the reflection $\mathrm{diag}(-1,1,\ldots,1)\in\rO(r)$, since it flips the sign of $\td c^{(1)}_\alpha$ and fixes $\td c^{(a)}_\alpha$ for $a\geq2$.

In \cref{eq:spin-orthogonal-howe}, one copy label $\gamma$ can pair with two representations of $\Spin(2n)$, and one representation of $\Spin(2n)$ with the two labels $\gamma$ and $\bar\gamma\defeq(\gamma_1,\ldots,\gamma_{\ell-1},-\gamma_\ell)$.
The full copy group $\rO(r)$ resolves both ambiguities:
\begin{align}
\Hf^{\otimes r}
\cong
\bigoplus_{\sigma}
V_{\kappa_{n,r}(\sigma)}^{\Spin(2n)}
\otimes
M_\sigma^{\rO(r)},
\label{eq:spin-orthogonal-disconnected-copy}
\end{align}
where $\sigma$ runs over the irreducible representations of $\rO(r)$ that occur and $\kappa_{n,r}(\sigma)$ is the paired highest weight.
If $\nu_n(\gamma)\neq0$, the two extensions of $M_\gamma^{\mathfrak{so}(r)}$ to $\rO(r)$ pair with $V_{\nu(\gamma)}$ and $V_{\tau\nu(\gamma)}$, in an order fixed by a convention; if $\nu_n(\gamma)=0$ and $\gamma_\ell>0$, the irreducible representation of $\rO(r)$ that restricts to $M_\gamma\oplus M_{\bar\gamma}$ pairs with $V_{\nu(\gamma)}$.

The parity-flipping unitary $\mu=c_{2n-1}$ of \cref{app:spin_pin} acts on $r$ copies as $\mu^{\otimes r}$, which normalizes the copy group but does not commute with it:
\begin{align}
\mu^{\otimes r}\,\td c^{(a)}_\alpha\,\mu^{\otimes r}
&=
(-1)^{a-1}\epsilon_\alpha\,\td c^{(a)}_\alpha,
\nonumber\\
\mu^{\otimes r}\,\Par^{(1)}\,\mu^{\otimes r}
&=
-\Par^{(1)},
\label{eq:pin-normalizes-copy}
\end{align}
with $\epsilon_\alpha=+1$ for $\alpha=2n-1$ and $\epsilon_\alpha=-1$ otherwise, so that $\mu^{\otimes r}L_{ab}\,\mu^{\otimes r}=(-1)^{a+b}L_{ab}$.
Hence the dual pair is $\lrp{\Spin(2n),\rO(r)}$, and there is no dual pair with $\Pin(2n)$ on the physical side.

The decomposition of $\Hf^{\otimes r}$ under $\Pin(2n)$ nevertheless follows from \cref{eq:spin-orthogonal-howe}.
Conjugation by $\mu$ acts on weights by $\tau$, so $\mu^{\otimes r}$ maps the $V_\nu$-isotypic component onto the $V_{\tau\nu}$-isotypic component.
If $\nu_n(\gamma)\neq0$, these two components share the multiplicity space $M_\gamma^{\mathfrak{so}(r)}$, and together they form $m_\gamma\defeq\dim M_\gamma^{\mathfrak{so}(r)}$ copies of the irreducible representation $V_{\nu(\gamma)}\oplus V_{\tau\nu(\gamma)}$ of $\Pin(2n)$.
If $\nu_n(\gamma)=0$, the $V_{\nu(\gamma)}$-isotypic component has multiplicity $2m_\gamma$ and splits into the two extensions $V^\pm_{\nu(\gamma)}$ of $V_{\nu(\gamma)}$ to $\Pin(2n)$, on whose highest-weight vectors $\mu$ acts by $\pm1$.
Each extension occurs $m_\gamma$ times, because $\Par^{(1)}$ commutes with $\Spin(2n)$ and anticommutes with $\mu^{\otimes r}$, and therefore exchanges the two isotypic components.
In both cases, the representation is induced from $V_{\nu(\gamma)}$, so that, as a representation of $\Pin(2n)$,
\begin{align}
\Hf^{\otimes r}
\cong
\bigoplus_{\gamma\in\cO^{+}_{n,r}}
\operatorname{Ind}_{\Spin(2n)}^{\Pin(2n)}\lrp{V^{\Spin(2n)}_{\nu(\gamma)}}
\otimes
\bC^{m_\gamma},
\label{eq:pin-decomposition}
\end{align}
where $\cO^{+}_{n,r}$ keeps, of each pair $\gamma,\bar\gamma$, the label with $\gamma_\ell\geq0$ (for odd $r$, $\cO^{+}_{n,r}=\cO_{n,r}$).
The multiplicities are those of \cref{eq:spin-orthogonal-howe}, but $\rO(r)$ does not act on the multiplicity spaces, since it does not commute with $\mu^{\otimes r}$.

The free-fermion states sit in the term $\gamma=0$.
With $\alpha,\beta$ the half-spinor highest weights of \cref{eq:ff_highest_weights}, $\nu(0)=r\alpha$ and $\tau\nu(0)=r\beta$.
By the proof of \cref{prop:free_fermion_+_orbit}, $\cH^{(r)}_{\Sffp}$ and $\cH^{(r)}_{\Sffm}$ are irreducible under $\Spin(2n)$, with highest weights $r\alpha$ and $r\beta$, respectively.
Since $\nu_n(\gamma)\geq0$ for all $\gamma$, and $\nu(\gamma)=\tfrac r2(1,\ldots,1)$ only for $\gamma=0$, they are the two summands of the block $\gamma=0$ of \cref{eq:spin-orthogonal-howe}.
The copy-side factor $M_0^{\mathfrak{so}(r)}$ of this block is trivial.
Hence the full $r$-copy span is the $\mathfrak{so}(r)$-invariant subspace,
\begin{align}
\cH^{(r)}_{\Sff}
=
\cH^{(r)}_{\Sffp}\oplus\cH^{(r)}_{\Sffm}
=
\lrp{\Hf^{\otimes r}}^{\mathfrak{so}(r)},
\end{align}
and in \cref{eq:pin-decomposition} it is the single copy ($m_0=1$) of the irreducible representation $V_{r\alpha}\oplus V_{r\beta}$ of $\Pin(2n)$, in agreement with \cref{eq:ncopy_split} and \cref{prop:free_fermion_+_orbit}.
On the copy side, $\Par^{(1)}\ket{\psi}^{\otimes r}=\pm\ket{\psi}^{\otimes r}$ for $\kpsi\in\Sffpm$, so $\cH^{(r)}_{\Sffp}$ carries the trivial and $\cH^{(r)}_{\Sffm}$ the determinant character of $\rO(r)$, the two extensions of $M_0^{\mathfrak{so}(r)}$.

\subsection{Bosonic Howe dualities}
\label{subsec:bosonic-howe}

For bosons, symmetric powers replace exterior powers, and again passive bosonic Gaussian unitaries pair with a unitary copy group and active ones with an orthogonal copy group; no grading is involved, so the copy permutations act without signs.

\paragraph{Passive bosons.}
\label{subsec:passive-bosonic-howe}
Identifying the Fock space of $n$ bosonic modes with $\Hb\cong\Sym^\bullet\bC^n$, the $r$-copy space has the symmetric Howe decomposition
\begin{align}
\Hb^{\otimes r}
\cong
\Sym^\bullet(\bC^n\otimes\bC^r)
\cong
\bigoplus_{\lambda:\,\ell(\lambda)\leq\min(n,r)}
V_{\lambda}^{\rU(n)}\otimes M_{\lambda}^{\rU(r)},
\label{eq:bosonic-symmetric-howe}
\end{align}
where $V_\lambda^{\rU(n)}$ and $M_\lambda^{\rU(r)}$ are the irreducible polynomial representations with highest weight $\lambda$.
For the Bose condensate states $\cS_{\mr{cond}}$ of \cref{sec:other_resolving_orbits}, the line of \cref{lem:copy-side-line} is a weight line inside a larger representation of the copy group.
The diagonal torus of $\rU(r)$ counts the bosons in each copy, and its weight-$(q,\ldots,q)$ subspace is $\lrp{\Sym^q\bC^n}^{\otimes r}$, which contains $\cH^{(r)}_{\cS_{\mr{cond}}}$; restricting \cref{eq:bosonic-symmetric-howe} to it gives
\begin{align}
\lrp{\Sym^q\bC^n}^{\otimes r}
\cong
\bigoplus_{\substack{\lambda\vdash rq\\\ell(\lambda)\leq\min(n,r)}}
V_\lambda^{\rU(n)}
\otimes
\lrp{M_\lambda^{\rU(r)}}_{(q^r)},
\end{align}
where $(\cdot)_{(q^r)}$ denotes the weight-$(q,\ldots,q)$ subspace.
The vector $\ket{\psi(v,q)}^{\otimes r}=(v^{\vee q})^{\otimes r}$ generates the Cartan component $\Sym^{rq}\bC^n$, so $\lambda_r=(rq)$.
Realizing $M_{(rq)}^{\rU(r)}=\Sym^{rq}\bC^r$ as polynomials in $x_1,\ldots,x_r$, its weight space $\lrp{M_{(rq)}^{\rU(r)}}_{(q^r)}=\bC\,x_1^q\cdots x_r^q$ is one-dimensional, and
\begin{align}
\cH^{(r)}_{\cS_{\mr{cond}}}
\cong
V_{(rq)}^{\rU(n)}
=
\Sym^{rq}(\bC^n),
\qquad
\dim
=
\binom{n+rq-1}{n-1},
\label{eq:condensate-howe-selection}
\end{align}
which recovers the count of \cref{sec:other_resolving_orbits}.

\paragraph{Active bosons.}
\label{subsec:active-bosonic-howe}
Allowing pair creation and annihilation enlarges the physical group to the metaplectic group $\operatorname{Mp}(2n,\bR)$ of active bosonic Gaussian unitaries and shrinks the copy group to $\rO(r)$; the two form the oscillator dual pair $\lrp{\operatorname{Mp}(2n,\bR),\rO(r)}$ inside $\operatorname{Mp}(2nr,\bR)$.
The corresponding Kashiwara--Vergne decomposition is indexed by the set $\cP_{n,r}$ of partitions, padded with zeros as needed, that satisfy
\begin{align}
\cP_{n,r}
\defeq
\left\{\lambda:\ \ell(\lambda)\leq n,
\ \lambda^{\mathsf T}_1+\lambda^{\mathsf T}_2\leq r\right\}.
\label{eq:oscillator-howe-labels}
\end{align}
For $\lambda\in\cP_{n,r}$, let $M_\lambda^{\rO(r)}$ be the irreducible tensor representation of $\rO(r)$ labeled by $\lambda$; then
\begin{align}
\Hb^{\otimes r}
\cong
\left(\Sym^\bullet\bC^n\right)^{\otimes r}
\cong
\bigoplus_{\lambda\in\cP_{n,r}}
\Theta_{n,r}(\lambda)
\otimes
M_\lambda^{\rO(r)},
\label{eq:metaplectic-orthogonal-howe}
\end{align}
where $\Theta_{n,r}(\lambda)$ is the irreducible unitary lowest-weight representation of $\operatorname{Mp}(2n,\bR)$ whose lowest $\rU(n)$-type has highest weight $\lambda+(\tfrac r2,\ldots,\tfrac r2)$~\cite{KashiwaraVergne1978,Howe1989}.

Writing $U_g$ for the metaplectic representation on $\Hb$, let $\cS_{\mathrm{sq}}\defeq\{U_g\ketbra{0}{0}U_g^\dagger:g\in\operatorname{Mp}(2n,\bR)\}$ be the orbit of centered squeezed vacua.
We determine its copy-side action directly from the squeezed-vacuum wavefunctions in \cref{eq:ab-squeezed-vector}.
Let $Z=(z^{(1)},\ldots,z^{(r)})\in\bC^{n\times r}$, where $z^{(c)}$ contains the Bargmann variables of copy $c$, and set $Y=ZZ^{\mathsf T}$.
Taking the tensor product of the $r$ normalized wavefunctions gives
\begin{align*}
F_{W,r}(Z)
&\defeq\prod_{c=1}^r\Omega_W(z^{(c)})\\
&=\Delta(W)^{r/4}\exp\!\left(\frac12\sum_{c=1}^r(z^{(c)})^{\mathsf T}Wz^{(c)}\right)\\
&=\Delta(W)^{r/4}\exp\!\left(\frac12\operatorname{tr}(WY)\right).
\end{align*}
The copy group $\rO(r)$ acts by mixing the columns of $Z$: for $R\in\rO(r)$, we use the vacuum-preserving action $(\mathcal U_Rf)(Z)\defeq f(ZR)$.
Since $\operatorname{tr}((ZR)^\dagger(ZR))=\operatorname{tr}(Z^\dagger Z)$ and the real Jacobian is $|\det R|^{2n}=1$, this change of variables preserves the product Bargmann inner product defined in \cref{sec:active-bosonic-cloning}.
It also fixes the vacuum function $F_{0,r}=1$.
Orthogonality now gives the action on every squeezed-vacuum tensor power explicitly:
\begin{align*}
(ZR)(ZR)^{\mathsf T}
&=ZRR^{\mathsf T}Z^{\mathsf T}=Y,\\
(\mathcal U_RF_{W,r})(Z)
&=\Delta(W)^{r/4}\exp\!\left(\frac12\operatorname{tr}\lrp{W(ZR)(ZR)^{\mathsf T}}\right)
=F_{W,r}(Z).
\end{align*}
This calculation applies also when $\det R=-1$, so reflections act by $+1$ and the copy-side character is trivial.
Thus the closed span $\cH^{(r)}_{\cS_{\mathrm{sq}}}$ is contained in the $\rO(r)$-fixed subspace of $\Hb^{\otimes r}$.

In \cref{eq:metaplectic-orthogonal-howe}, the only copy-side representation with a fixed vector is the trivial representation $M_{\varnothing}^{\rO(r)}=\bC$, so this fixed subspace is $\Theta_{n,r}(\varnothing)$.
The squeezed-vacuum span is nonzero because it contains the vacuum at $W=0$, and it is a closed metaplectic-invariant subspace because it is generated by the vacuum orbit.
Irreducibility of $\Theta_{n,r}(\varnothing)$ therefore gives
\begin{align}
\cH^{(r)}_{\cS_{\mathrm{sq}}}
=
\lrp{\Hb^{\otimes r}}^{\rO(r)}
=
\Theta_{n,r}(\varnothing),
\label{eq:squeezed-howe-selection}
\end{align}
the scalar lowest-weight representation of $\operatorname{Mp}(2n,\bR)$ with lowest weight $\tfrac r2(1,\ldots,1)$.
It is square-integrable exactly when $r/2>n$, so its coherent-state frame has a finite formal degree precisely for $r\geq2n+1$, in agreement with \cref{thm:ab-resolving}.
Displaced Gaussian states also involve the Heisenberg--Weyl displacements and lie outside this decomposition.

%% file: AppPhaseStates.tex
\section{Technical lemmas for phase-state optimality}
\label{app:phase-smoothing}

This appendix proves the two-coordinate smoothing lemma used in \cref{subsec:phase-optimality}. The proof reduces the comparison of two occupation shifts to a one-dimensional statement on each line of input occupations. We first establish the required log-concavity of the line weights.

\begin{lemma}[Line log-concavity]
\label{lem:phase-line-log-concavity-app}
Fix distinct coordinates $i,j$ and a vector $\boldsymbol\eta\in\N^d$ satisfying $\eta_i=\eta_j=0$ and $|\boldsymbol\eta|\leq N$. Set $n=N-|\boldsymbol\eta|$ and
\begin{align}
    \mathbf b(r) \coloneq \boldsymbol\eta + r\mathbf e_i + (n-r)\mathbf e_j, \qquad 0\leq r\leq n. \label{eq:phase-line-app}
\end{align}
Then the sequence
\begin{align}
    w_r \coloneq \sqrt{ \frac{ \binom{N}{\mathbf b(r)} }{ Z_{\mathbf b(r)} } } \label{eq:phase-line-weight-app}
\end{align}
is symmetric and log-concave:
\begin{align}
    w_{n-r} &= w_r, \label{eq:phase-line-symmetry-app}\\
    w_r^2 &\geq w_{r-1}w_{r+1} \qquad (1\leq r\leq n-1). \label{eq:phase-line-log-concavity-app}
\end{align}
\end{lemma}

\begin{proof}
Let $\sigma_{ij}$ be the permutation that exchanges coordinates $i$ and $j$. Since $\mathbf b(n-r)=\sigma_{ij}\mathbf b(r)$, the multinomial coefficient satisfies
\begin{align}
    \binom{N}{\mathbf b(n-r)} = \binom{N}{\mathbf b(r)}.
\end{align}
The set $\cC_{d,K}$ is invariant under $\sigma_{ij}$. Reindexing $\boldsymbol\Gamma=\sigma_{ij}\boldsymbol\Theta$ in the definition of $Z_{\mathbf b}$ gives
\begin{align}
    Z_{\mathbf b(n-r)} &= \sum_{\boldsymbol\Gamma\in\cC_{d,K}} \binom{N+K}{ \sigma_{ij}\mathbf b(r) + \boldsymbol\Gamma } = \sum_{\boldsymbol\Theta\in\cC_{d,K}} \binom{N+K}{ \sigma_{ij} \bigl( \mathbf b(r) + \boldsymbol\Theta \bigr) } = Z_{\mathbf b(r)}.
\end{align}
This proves \cref{eq:phase-line-symmetry-app}.

For log-concavity, define
\begin{align}
    H(r) \coloneq \frac{ Z_{\mathbf b(r)} }{ \binom{N}{\mathbf b(r)} }. \label{eq:phase-H-app}
\end{align}
We show that $H$ is log-convex. For a fixed $\boldsymbol\Gamma\in\cC_{d,K}$, write $p=\Gamma_i$ and $\ell=\Gamma_j$. Its contribution to $H(r)$ has the form
\begin{align}
    h_{\boldsymbol\Gamma}(r) = c_{\boldsymbol\Gamma} \frac{r!}{(r+p)!} \frac{(n-r)!}{(n-r+\ell)!}, \label{eq:phase-h-gamma-app}
\end{align}
where $c_{\boldsymbol\Gamma}>0$ is independent of $r$. Its successive ratio is
\begin{align}
    \frac{ h_{\boldsymbol\Gamma}(r+1) }{ h_{\boldsymbol\Gamma}(r) } = \frac{r+1}{r+p+1} \frac{n-r+\ell}{n-r}. \label{eq:phase-h-ratio-app}
\end{align}
Both factors are nondecreasing in $r$: they can be written as $1-p/(r+p+1)$ and $1+\ell/(n-r)$, respectively. Hence the ratio in \cref{eq:phase-h-ratio-app} is nondecreasing, which is equivalent to
\begin{align}
    h_{\boldsymbol\Gamma}(r)^2 \leq h_{\boldsymbol\Gamma}(r-1) h_{\boldsymbol\Gamma}(r+1). \label{eq:phase-h-log-convex-app}
\end{align}
Thus every $h_{\boldsymbol\Gamma}$ is log-convex. Since $H(r)=\sum_{\boldsymbol\Gamma}h_{\boldsymbol\Gamma}(r)$, \cref{eq:phase-h-log-convex-app} and Cauchy--Schwarz imply
\begin{align}
    H(r) &\leq \sum_{\boldsymbol\Gamma\in\cC_{d,K}} \sqrt{ h_{\boldsymbol\Gamma}(r-1) h_{\boldsymbol\Gamma}(r+1) } \leq \sqrt{ H(r-1)H(r+1) }.
\end{align}
Therefore $H$ is log-convex, and $w_r=H(r)^{-1/2}$ is log-concave.
\end{proof}

\begin{lemma}[Hoggar's convolution theorem]
\label{lem:hoggar-convolution-app}
The convolution of two nonnegative log-concave sequences with no internal zeros is log-concave \cite{Hoggar1974}.
\end{lemma}

Here a sequence has no internal zeros when the indices on which it is positive form an interval. We apply \cref{lem:hoggar-convolution-app} only to finite sequences that are strictly positive on their respective supports.

\begin{proof}[Proof of \cref{lem:phase-smoothing-main}]
Fix coordinates $i,j$ satisfying $\Delta_i\geq\Delta_j+2$, and let
\begin{align}
    \boldsymbol\Delta' = \boldsymbol\Delta-\mathbf e_i+\mathbf e_j.
\end{align}
We partition $\cB_N$ into lines on which only the $i$th and $j$th coordinates vary. Define
\begin{align}
    \mathcal E_{ij} \coloneq \left\{ \boldsymbol\eta\in\N^d: \eta_i=\eta_j=0,\ |\boldsymbol\eta|\leq N \right\}.
\end{align}
For $\boldsymbol\eta\in\mathcal E_{ij}$, put $n_{\boldsymbol\eta}=N-|\boldsymbol\eta|$ and
\begin{align}
    \mathbf b_{\boldsymbol\eta}(r) \coloneq \boldsymbol\eta + r\mathbf e_i + (n_{\boldsymbol\eta}-r)\mathbf e_j, \qquad 0\leq r\leq n_{\boldsymbol\eta}. \label{eq:phase-line-partition-app}
\end{align}
Every $\mathbf b\in\cB_N$ occurs on exactly one of these lines. Using the zero convention from \cref{eq:Lambda-common-domain-main}, define the complete contribution of a line by
\begin{align}
    \Lambda_{\boldsymbol\eta}(\boldsymbol\Delta) \coloneq \sum_{r=0}^{n_{\boldsymbol\eta}} w_{\boldsymbol\eta,r} \binom{N+K}{ \mathbf b_{\boldsymbol\eta}(r) + \boldsymbol\Delta }, \label{eq:phase-line-contribution-app}
\end{align}
where
\begin{align}
    w_{\boldsymbol\eta,r} \coloneq \sqrt{ \frac{ \binom{N}{ \mathbf b_{\boldsymbol\eta}(r) } }{ Z_{\mathbf b_{\boldsymbol\eta}(r)} } }.
\end{align}
The line partition gives
\begin{align}
    \Lambda(\boldsymbol\Delta) = \sum_{\boldsymbol\eta\in\mathcal E_{ij}} \Lambda_{\boldsymbol\eta}(\boldsymbol\Delta), \qquad \Lambda(\boldsymbol\Delta') = \sum_{\boldsymbol\eta\in\mathcal E_{ij}} \Lambda_{\boldsymbol\eta}(\boldsymbol\Delta'). \label{eq:phase-Lambda-line-sum-app}
\end{align}

Fix one line and suppress the subscript $\boldsymbol\eta$. Write
\begin{align}
    n &\coloneq n_{\boldsymbol\eta}, & A &\coloneq \Delta_i, & B &\coloneq \Delta_j, & L &\coloneq n+A+B. \label{eq:phase-line-parameters-app}
\end{align}
If $\eta_\ell+\Delta_\ell<0$ for some $\ell\neq i,j$, then this line contributes zero for both shifts. The same is true if $L<0$, because the total output occupation in coordinates $i,j$ is then negative. We may therefore assume that these quantities are nonnegative. Define
\begin{align}
    C_{\boldsymbol\eta} \coloneq \frac{ (N+K)! }{ L! \prod_{\ell\neq i,j} (\eta_\ell+\Delta_\ell)! }
\end{align}
and let $\beta_L(u)=\binom{L}{u}$, with the binomial coefficient understood to be zero outside $0\leq u\leq L$. For every $0\leq r\leq n$,
\begin{align}
    \binom{N+K}{ \mathbf b(r)+\boldsymbol\Delta } &= C_{\boldsymbol\eta} \beta_L(r+A), \label{eq:phase-line-factorization-delta-app}\\
    \binom{N+K}{ \mathbf b(r)+\boldsymbol\Delta' } &= C_{\boldsymbol\eta} \beta_L(r+A-1). \label{eq:phase-line-factorization-delta-prime-app}
\end{align}
The same $C_{\boldsymbol\eta}$ and $L$ occur in both expressions because smoothing preserves $A+B$ and leaves all remaining coordinates unchanged. Thus
\begin{align}
    \Lambda_{\boldsymbol\eta}(\boldsymbol\Delta) &= C_{\boldsymbol\eta}F_{\boldsymbol\eta}(A), \label{eq:phase-line-F-delta-app}\\
    \Lambda_{\boldsymbol\eta}(\boldsymbol\Delta') &= C_{\boldsymbol\eta}F_{\boldsymbol\eta}(A-1), \label{eq:phase-line-F-delta-prime-app}\\
    F_{\boldsymbol\eta}(a) &\coloneq \sum_{r=0}^{n} w_r\beta_L(r+a). \label{eq:phase-F-definition-app}
\end{align}

By \cref{lem:phase-line-log-concavity-app}, $(w_r)_{r=0}^n$ is symmetric and log-concave. Extend it by zero outside $0\leq r\leq n$, and define the reflected sequence $\check w_s\coloneq w_{-s}$. Reflection preserves log-concavity, and $\check w$ is positive exactly on the consecutive indices $-n,\ldots,0$, so it has no internal zeros. The binomial sequence $\beta_L$ also has no internal zeros and is log-concave because
\begin{align}
    \frac{ \binom{L}{u}^2 }{ \binom{L}{u-1} \binom{L}{u+1} } = \frac{ (u+1)(L-u+1) }{ u(L-u) } \geq1
\end{align}
at every interior index of its support. Therefore \cref{lem:hoggar-convolution-app} gives
\begin{align}
    F_{\boldsymbol\eta}(a) &= \sum_{s\in\mathbb Z} \check w_s \beta_L(a-s) = (\check w*\beta_L)(a), \label{eq:phase-F-convolution-app}
\end{align}
so $F_{\boldsymbol\eta}$ is log-concave.

It is also symmetric. Using $w_{n-r}=w_r$ and $\beta_L(L-u)=\beta_L(u)$,
\begin{align}
    F_{\boldsymbol\eta}(L-n-a) &= \sum_{r=0}^{n} w_r \beta_L(r+L-n-a) = \sum_{s=0}^{n} w_{n-s} \beta_L(L-s-a) \nonumber\\
    &= \sum_{s=0}^{n} w_s \beta_L(s+a) = F_{\boldsymbol\eta}(a). \label{eq:phase-F-symmetry-app}
\end{align}
Hence $F_{\boldsymbol\eta}$ is symmetric about
\begin{align}
    \frac{L-n}{2} = \frac{A+B}{2}.
\end{align}
A symmetric log-concave sequence is nonincreasing to the right of its center. Since $A\geq B+2$,
\begin{align}
    A-1-\frac{A+B}{2} = \frac{A-B-2}{2} \geq0,
\end{align}
and $A-1$ is one step closer to the center than $A$. Therefore
\begin{align}
    F_{\boldsymbol\eta}(A-1) \geq F_{\boldsymbol\eta}(A).
\end{align}
Together with \cref{eq:phase-line-F-delta-app,eq:phase-line-F-delta-prime-app}, this proves
\begin{align}
    \Lambda_{\boldsymbol\eta}(\boldsymbol\Delta') \geq \Lambda_{\boldsymbol\eta}(\boldsymbol\Delta)
\end{align}
on every nonexcluded line. On an excluded line, either a fixed output coordinate is negative for both shifts or the common $i,j$ total $L$ is negative, so both contributions vanish. Summing over the common line partition in \cref{eq:phase-Lambda-line-sum-app} yields
\begin{align}
    \Lambda(\boldsymbol\Delta') - \Lambda(\boldsymbol\Delta) = \sum_{\boldsymbol\eta\in\mathcal E_{ij}} \left( \Lambda_{\boldsymbol\eta}(\boldsymbol\Delta') - \Lambda_{\boldsymbol\eta}(\boldsymbol\Delta) \right) \geq0.
\end{align}
This proves \cref{lem:phase-smoothing-main}.
\end{proof}

\section{Technical proofs for one-copy phase-state cloning}
\label{app:phase-one-copy-asymptotics}

This appendix supplies the detailed variance estimates and sample-complexity inversion for \Cref{subsec:phase-one-copy-asymptotics}. 

\subsection{Upper variance estimate}
\label{app:phase-one-copy-upper}

The first step compares the variances of $\sqrt X$ and $X$. We then apply Efron--Stein twice, keeping track of the reciprocal expressions explicitly.
\begin{proposition}[Upper bound]
\label{prop:phase-one-copy-upper}
For $d\ge2$ and $N\ge2$,
\begin{align}
    \operatorname{Var}(\sqrt X)
    \le \frac{20(d-1)(N+d)}{(N+2)(N+3)(N+4)}.
    \label{eq:phase-one-copy-upper-exact}
\end{align}
In particular, $\operatorname{Var}(\sqrt X)\le30(d-1)/N^2$ for $N\ge2d$.
\end{proposition}
\begin{proof}
For each $\mathbf{b}\in\cB_N$, apply Cauchy--Schwarz to the vectors with entries $\sqrt{b_j+1}$ and $1/\sqrt{b_j+1}$:
\begin{align*}
    d^2
    &=\left(\sum_{j=0}^{d-1}\sqrt{b_j+1}\,\frac1{\sqrt{b_j+1}}\right)^2 \\
    &\le\left(\sum_{j=0}^{d-1}(b_j+1)\right)
    \left(\sum_{j=0}^{d-1}\frac1{b_j+1}\right)
    =(N+d)\sum_{j=0}^{d-1}\frac1{b_j+1}.
\end{align*}
Here $\sum_{j=0}^{d-1}(b_j+1)=N+d$ because $\sum_{j=0}^{d-1}b_j=N$. Dividing by $N+d$ and substituting into the definition of $X$ gives
\begin{align*}
    X(\mathbf{b})
    &=\frac{N+1}{d^2}\sum_{j=0}^{d-1}\frac1{b_j+1} \\
    &\ge\frac{N+1}{d^2}\,\frac{d^2}{N+d}
    =\frac{N+1}{N+d}.
\end{align*}
Thus $X=X(\mathbf{B})\ge(N+1)/(N+d)$ for every realization of $\mathbf{B}$. To compare the variances of $\sqrt X$ and $X$, let $X'$ be independent of $X$ and have the same distribution. Expanding the squares and using independence gives
\begin{align*}
    \E\bigl[(\sqrt X-\sqrt{X'})^2\bigr]
    &=\E[X]+\E[X']-2\E[\sqrt X]\,\E[\sqrt{X'}]
    =2\operatorname{Var}(\sqrt X), \\
    \E\bigl[(X-X')^2\bigr]
    &=\E[X^2]+\E[(X')^2]-2\E[X]\,\E[X']
    =2\operatorname{Var}(X).
\end{align*}
Since both $X$ and $X'$ are at least $(N+1)/(N+d)$, we also have
\begin{align*}
    (\sqrt X+\sqrt{X'})^2
    \ge\left(2\sqrt{\frac{N+1}{N+d}}\right)^2
    =\frac{4(N+1)}{N+d}.
\end{align*}
Therefore,
\begin{align}
    \operatorname{Var}(\sqrt X)
    &=\frac12\E\bigl[(\sqrt X-\sqrt{X'})^2\bigr] \nonumber\\
    &=\frac12\E\left[\frac{(X-X')^2}{(\sqrt X+\sqrt{X'})^2}\right] \nonumber\\
    &\le\frac{N+d}{8(N+1)}\E\bigl[(X-X')^2\bigr]
    =\frac{N+d}{4(N+1)}\operatorname{Var}(X).
    \label{eq:phase-one-copy-lipschitz}
\end{align}

We use Efron--Stein in its replacement form: if $J_1,\ldots,J_m$ are independent, $U=f(J_1,\ldots,J_m)$, and $U^{(i)}$ is obtained by replacing only $J_i$ with an independent copy, then
\begin{align}
    \operatorname{Var}(U)\le\frac12\sum_{i=1}^m\E\bigl[(U-U^{(i)})^2\bigr].
    \label{eq:phase-one-copy-efron-stein}
\end{align}
This standard inequality follows from the tensorization of variance; see \cite{BLM2013}.

Recall that
\begin{align*}
    X(\mathbf{B})=\frac{N+1}{d^2}\sum_{j=0}^{d-1}\frac1{B_j+1},
\end{align*}
where $\mathbf{B}$ is the occupation vector of the independent uniform samples $J_1,\ldots,J_N$.

Let $J_1'$ be another uniform sample, independent of all $J_1,\ldots,J_N$, and define
\begin{align*}
    C_j=\sum_{r=2}^{N}\mathbf1_{\{J_r=j\}},
    \qquad j=0,\ldots,d-1.
\end{align*}
Thus $\mathbf{C}$ counts the remaining $N-1$ samples, and
\begin{align*}
    \mathbf{B}=\mathbf{C}+\mathbf{e}_{J_1},
    \qquad X^{(1)}=X(\mathbf{C}+\mathbf{e}_{J_1'}).
\end{align*}
For each $i$, adding one sample to coordinate $i$ gives
\begin{align*}
    X(\mathbf{C}+\mathbf{e}_i)
    &=\frac{N+1}{d^2}\left(
    \sum_{j=0}^{d-1}\frac1{C_j+1}
    -\frac1{C_i+1}+\frac1{C_i+2}\right) \\
    &=\frac{N+1}{d^2}\left(
    \sum_{j=0}^{d-1}\frac1{C_j+1}
    -\frac1{(C_i+1)(C_i+2)}\right).
\end{align*}
Consequently,
\begin{align*}
    X^{(1)}-X
    =\frac{N+1}{d^2}\left(
    \frac1{(C_{J_1}+1)(C_{J_1}+2)}
    -\frac1{(C_{J_1'}+1)(C_{J_1'}+2)}\right).
\end{align*}
Conditional on $\mathbf{C}$, the two reciprocals on the right are independent and identically distributed. The independent-copy variance identity therefore gives
\begin{align*}
    \E\bigl[(X-X^{(1)})^2\mid\mathbf{C}\bigr]
    =\frac{2(N+1)^2}{d^4}
    \operatorname{Var}\left(\frac1{(C_{J_1}+1)(C_{J_1}+2)}\,\middle|\,\mathbf{C}\right).
\end{align*}
All $N$ sample positions have the same expected squared replacement effect. Applying \cref{eq:phase-one-copy-efron-stein} yields
\begin{align*}
    \operatorname{Var}(X)
    &\le\frac N2\E\bigl[(X-X^{(1)})^2\bigr] \\
    &=\frac{N(N+1)^2}{d^4}
    \E\left[\operatorname{Var}\left(\frac1{(C_{J_1}+1)(C_{J_1}+2)}\,\middle|\,\mathbf{C}\right)\right].
\end{align*}
By the law of total variance, the expected conditional variance is at most the unconditional variance. Moreover, conditional on $J_1=i$, the variable $C_{J_1}=C_i$ counts how many of the other $N-1$ independent samples equal $i$. Hence
\begin{align*}
    T\coloneq C_{J_1}\sim\Bin(N-1,1/d),
\end{align*}
and we obtain
\begin{align}
    \operatorname{Var}(X)
    \le\frac{N(N+1)^2}{d^4}
    \operatorname{Var}\left(\frac1{(T+1)(T+2)}\right).
    \label{eq:phase-one-copy-first-efron}
\end{align}

We next bound this binomial variance. Write
\begin{align*}
    T=\xi_1+\cdots+\xi_{N-1},
    \qquad \xi_r\sim\operatorname{Bernoulli}(1/d)
\end{align*}
independently, and let
\begin{align*}
    U=\xi_2+\cdots+\xi_{N-1}
    \sim\Bin(N-2,1/d).
\end{align*}
Replacing $\xi_1$ by an independent copy changes its value with probability
\begin{align*}
    2\frac1d\left(1-\frac1d\right).
\end{align*}
When it changes, the reciprocal switches between
\begin{align*}
    \frac1{(U+1)(U+2)}
    \quad\text{and}\quad
    \frac1{(U+2)(U+3)}.
\end{align*}
A second application of \cref{eq:phase-one-copy-efron-stein} therefore gives
\begin{align}
    \operatorname{Var}\left(\frac1{(T+1)(T+2)}\right)
    &\le\frac{N-1}{2}\,2\frac1d\left(1-\frac1d\right) \nonumber\\[-2pt]
    &\qquad{}\times\E\left[\left(
    \frac1{(U+2)(U+3)}-\frac1{(U+1)(U+2)}
    \right)^2\right].
    \label{eq:phase-one-copy-binomial-poincare}
\end{align}
For every $k\ge0$,
\begin{align}
    \left(\frac1{(k+2)(k+3)}-\frac1{(k+1)(k+2)}\right)^2
    &=\frac4{(k+1)^2(k+2)^2(k+3)^2} \nonumber\\
    &\le\frac{80}{(k+1)(k+2)\cdots(k+6)}.
    \label{eq:phase-one-copy-difference-bound}
\end{align}
Indeed, the last inequality is equivalent to
\begin{align*}
    4\frac{k+4}{k+1}\frac{k+5}{k+2}\frac{k+6}{k+3}
    \le4\cdot4\cdot\frac52\cdot2=80.
\end{align*}
It follows that
\begin{align*}
    \operatorname{Var}\left(\frac1{(T+1)(T+2)}\right)
    \le\frac{80(N-1)(d-1)}{d^2}
    \E\left[\frac1{(U+1)(U+2)\cdots(U+6)}\right].
\end{align*}
To evaluate the remaining expectation, use
\begin{align*}
    \frac1{(k+1)(k+2)\cdots(k+6)}
    =\frac1{120}\int_0^1u^5(1-u)^k\,du.
\end{align*}
Recall that $U=\xi_2+\cdots+\xi_{N-1}\sim\Bin(N-2,1/d)$ counts successes in $N-2$ independent trials, each succeeding with probability $1/d$. For fixed $u\in[0,1]$, a successful trial contributes a factor $1-u$, while a failed trial contributes a factor $1$. There are exactly $U$ successful trials, so
\begin{align*}
    (1-u)^U=\prod_{r=2}^{N-1}(1-u\xi_r).
\end{align*}
The expected factor from one trial is
\begin{align*}
    \E[1-u\xi_r]
    =\frac1d(1-u)+\left(1-\frac1d\right)\cdot1
    =1-\frac ud.
\end{align*}
By independence, the expectation of the product equals the product of the expectations. There are $N-2$ factors, giving
\begin{align*}
    \E[(1-u)^U]
    &=\E\left[\prod_{r=2}^{N-1}(1-u\xi_r)\right]
    =\prod_{r=2}^{N-1}\E[1-u\xi_r]
    =\left(1-\frac ud\right)^{N-2}.
\end{align*}
\begin{samepage}
Since $U$ takes only finitely many values, its expectation is a finite weighted sum, which can be interchanged with the integral. Using the identity above then gives
\begin{align}
    \E\left[\frac1{(U+1)(U+2)\cdots(U+6)}\right]
    &=\frac1{120}\int_0^1u^5\E[(1-u)^U]\,du \nonumber\\
    &=\frac1{120}\int_0^1u^5\left(1-\frac ud\right)^{N-2}\,du \nonumber\\
    &=\frac{d^6}{120}\int_0^{1/d}v^5(1-v)^{N-2}\,dv \nonumber\\
    &\le\frac{d^6}{120}\int_0^1v^5(1-v)^{N-2}\,dv \nonumber\\
    &=\frac{d^6}{(N-1)N(N+1)(N+2)(N+3)(N+4)}.
    \label{eq:phase-one-copy-beta-integral}
\end{align}
\end{samepage}
Thus
\begin{align*}
    \operatorname{Var}\left(\frac1{(T+1)(T+2)}\right)
    \le\frac{80(d-1)d^4}{N(N+1)(N+2)(N+3)(N+4)}.
\end{align*}
Substituting into our bound on $\operatorname{Var}(X)$, we obtain
\begin{align*}
    \operatorname{Var}(X)
    &\le\frac{N(N+1)^2}{d^4}
    \frac{80(d-1)d^4}{N(N+1)(N+2)(N+3)(N+4)} \\
    &=\frac{80(d-1)(N+1)}{(N+2)(N+3)(N+4)}.
\end{align*}
Finally, combining this with the previously proved comparison of variances in \cref{eq:phase-one-copy-lipschitz} gives
\begin{align*}
    \operatorname{Var}(\sqrt X)
    &\le\frac{N+d}{4(N+1)}\operatorname{Var}(X) \\
    &\le\frac{20(d-1)(N+d)}{(N+2)(N+3)(N+4)}.
\end{align*}
For $N\ge2d$, we have $N+d\le3N/2$, so
\begin{align*}
    \operatorname{Var}(\sqrt X)\le\frac{30(d-1)}{N^2}.
    \qquad\qedhere
\end{align*}
\end{proof}

\subsection{Lower variance estimate}
\label{app:phase-one-copy-lower}

We use the covariance with the count of equal sample pairs to obtain a lower bound on the variance. The calculation then reduces to the second derivative along a transfer of occupation between two coordinates.
\begin{proposition}[Lower bound]
\label{prop:phase-one-copy-lower}
For $d\ge2$ and $N\ge2d$,
\begin{align}
    \operatorname{Var}(\sqrt X)\ge10^{-9}\frac{d-1}{N^2}.
    \label{eq:phase-one-copy-lower}
\end{align}
\end{proposition}
\begin{proof}
Recall that $J_1,\ldots,J_N$ are independent uniform samples from $\{0,\ldots,d-1\}$ whose occupation vector is $\mathbf{B}$. Introduce the centered collision count
\begin{align}
    L=\sum_{1\le a<b\le N}
    \left(\mathbf1_{\{J_a=J_b\}}-\frac1d\right).
\end{align}
Distinct collision indicators are uncorrelated: this is immediate for disjoint pairs, while two pairs sharing one sample satisfy $\mathbb{P}(J_1=J_2=J_3)=d^{-2}$. Hence
\begin{align}
    \operatorname{Var}(L)=\binom N2\frac{d-1}{d^2}.
    \label{eq:phase-one-copy-collision-variance}
\end{align}

Let $\mathbf{C}\sim\Mult(N-2;1/d,\ldots,1/d)$ record the occupations of the remaining $N-2$ samples. By permutation symmetry,
\begin{align*}
    \E\left[\sqrt{X(\mathbf{B})}\,\middle|\,J_1=J_2\right]
    &=\E\left[\sqrt{X(\mathbf{C}+2\mathbf{e}_0)}\right], \\
    \E\left[\sqrt{X(\mathbf{B})}\,\middle|\,J_1\ne J_2\right]
    &=\E\left[\sqrt{X(\mathbf{C}+\mathbf{e}_0+\mathbf{e}_1)}\right].
\end{align*}
Since the equality event has probability $1/d$,
\begin{align}
    \Cov\bigl(\sqrt{X(\mathbf{B})},L\bigr)
    =\binom N2\frac{d-1}{d^2}
    \E\left[\sqrt{X(\mathbf{C}+2\mathbf{e}_0)}-\sqrt{X(\mathbf{C}+\mathbf{e}_0+\mathbf{e}_1)}\right].
    \label{eq:phase-one-copy-collision-covariance}
\end{align}
Since $\E[L]=0$, applying Cauchy--Schwarz to $\sqrt{X(\mathbf{B})}-\E[\sqrt{X(\mathbf{B})}]$ and $L$ gives
\begin{align*}
    \Cov\bigl(\sqrt{X(\mathbf{B})},L\bigr)^2
    &=\left(\E\left[\left(\sqrt{X(\mathbf{B})}-\E[\sqrt{X(\mathbf{B})}]\right)L\right]\right)^2 \\
    &\le\E\left[\left(\sqrt{X(\mathbf{B})}-\E[\sqrt{X(\mathbf{B})}]\right)^2\right]\E[L^2] \\
    &=\operatorname{Var}\bigl(\sqrt{X(\mathbf{B})}\bigr)\operatorname{Var}(L).
\end{align*}
Dividing by $\operatorname{Var}(L)>0$ and substituting \cref{eq:phase-one-copy-collision-variance,eq:phase-one-copy-collision-covariance}, we obtain
\begin{align}
    \operatorname{Var}\bigl(\sqrt{X(\mathbf{B})}\bigr)
    &\ge\frac{\Cov\bigl(\sqrt{X(\mathbf{B})},L\bigr)^2}{\operatorname{Var}(L)} \nonumber\\
    &=\frac{\left(\binom N2\frac{d-1}{d^2}\right)^2}
    {\binom N2\frac{d-1}{d^2}}
    \left(\E\left[\sqrt{X(\mathbf{C}+2\mathbf{e}_0)}-\sqrt{X(\mathbf{C}+\mathbf{e}_0+\mathbf{e}_1)}\right]\right)^2 \nonumber\\
    &=\binom N2\frac{d-1}{d^2}
    \left(\E\left[\sqrt{X(\mathbf{C}+2\mathbf{e}_0)}-\sqrt{X(\mathbf{C}+\mathbf{e}_0+\mathbf{e}_1)}\right]\right)^2.
    \label{eq:phase-one-copy-covariance-lower}
\end{align}

We bound the difference of expectations in \cref{eq:phase-one-copy-covariance-lower} by moving occupation between the first two coordinates. Recall that
\begin{align*}
    X(\mathbf{b})=\frac{N+1}{d^2}\sum_{j=0}^{d-1}\frac1{b_j+1}.
\end{align*}
We use this same formula on nonnegative real coordinates, keeping the prefactor $(N+1)/d^2$ fixed. For fixed $\mathbf{C}$, define
\begin{align}
    h_{\mathbf{C}}(t)
    =\sqrt{X\bigl(\mathbf{C}+\mathbf{e}_0+\mathbf{e}_1+t(\mathbf{e}_0-\mathbf{e}_1)\bigr)},
    \qquad -1\le t\le1.
\end{align}
Writing out $X$ gives
\begin{align*}
    h_{\mathbf{C}}(t)
    =\left[\frac{N+1}{d^2}\left(
    \frac1{C_0+2+t}+\frac1{C_1+2-t}
    +\sum_{j=2}^{d-1}\frac1{C_j+1}
    \right)\right]^{1/2}.
\end{align*}
Set $u_t=C_0+2+t$ and $v_t=C_1+2-t$, so that $u_t'=1$ and $v_t'=-1$. The remaining sum is independent of $t$, and the chain rule gives
\begin{align*}
    h_{\mathbf{C}}'(t)
    &=\frac1{2h_{\mathbf{C}}(t)}\frac{N+1}{d^2}
    \left(-u_t^{-2}u_t'-v_t^{-2}v_t'\right) \\
    &=\frac{N+1}{2d^2}\frac{-u_t^{-2}+v_t^{-2}}{h_{\mathbf{C}}(t)}.
\end{align*}
Differentiating once more and substituting this expression for $h_{\mathbf{C}}'(t)$, we obtain
\begin{align}
    h_{\mathbf{C}}''(t)
    &=\frac{N+1}{2d^2}\left[
    \frac{2u_t^{-3}+2v_t^{-3}}{h_{\mathbf{C}}(t)}
    -\frac{(-u_t^{-2}+v_t^{-2})h_{\mathbf{C}}'(t)}{h_{\mathbf{C}}(t)^2}
    \right] \nonumber\\
    &=\frac{N+1}{d^2}\frac{u_t^{-3}+v_t^{-3}}{h_{\mathbf{C}}(t)}
    -\frac{N+1}{2d^2}\frac{-u_t^{-2}+v_t^{-2}}{h_{\mathbf{C}}(t)^2}
    \left(\frac{N+1}{2d^2}\frac{-u_t^{-2}+v_t^{-2}}{h_{\mathbf{C}}(t)}\right) \nonumber\\
    &=\frac{N+1}{d^2}\frac{u_t^{-3}+v_t^{-3}}{h_{\mathbf{C}}(t)}
    -\frac{(N+1)^2}{4d^4}\frac{(-u_t^{-2}+v_t^{-2})^2}{h_{\mathbf{C}}(t)^3}.
    \label{eq:phase-one-copy-curvature-identity}
\end{align}
To obtain a lower bound on this derivative, we bound the term being subtracted. First, the definition of $h_{\mathbf{C}}$ gives
\begin{align*}
    \frac{d^2h_{\mathbf{C}}(t)^2}{N+1}
    =u_t^{-1}+v_t^{-1}+\sum_{j=2}^{d-1}\frac1{C_j+1}
    \ge u_t^{-1}+v_t^{-1}.
\end{align*}
Cauchy--Schwarz therefore yields
\begin{align*}
    (-u_t^{-2}+v_t^{-2})^2
    &\le(u_t^{-1}+v_t^{-1})(u_t^{-3}+v_t^{-3}) \\
    &\le\frac{d^2h_{\mathbf{C}}(t)^2}{N+1}(u_t^{-3}+v_t^{-3}).
\end{align*}
Multiplying by $(N+1)^2/(4d^4h_{\mathbf{C}}(t)^3)$ gives
\begin{align*}
    \frac{(N+1)^2}{d^4}\frac{(-u_t^{-2}+v_t^{-2})^2}{4h_{\mathbf{C}}(t)^3}
    &\le\frac{(N+1)^2}{4d^4h_{\mathbf{C}}(t)^3}
    \frac{d^2h_{\mathbf{C}}(t)^2}{N+1}(u_t^{-3}+v_t^{-3}) \\
    &=\frac{N+1}{4d^2}\frac{u_t^{-3}+v_t^{-3}}{h_{\mathbf{C}}(t)}.
\end{align*}
Substituting this upper bound for the subtracted term in \cref{eq:phase-one-copy-curvature-identity}, we obtain
\begin{align}
    h_{\mathbf{C}}''(t)
    &\ge\frac{N+1}{d^2}\frac{u_t^{-3}+v_t^{-3}}{h_{\mathbf{C}}(t)}
    -\frac{N+1}{4d^2}\frac{u_t^{-3}+v_t^{-3}}{h_{\mathbf{C}}(t)} \nonumber\\
    &=\frac{3(N+1)}{4d^2}\frac{u_t^{-3}+v_t^{-3}}{h_{\mathbf{C}}(t)}
    \ge0.
    \label{eq:phase-one-copy-curvature}
\end{align}

Consider the event
\begin{align}
    \mathcal G
    =\left\{C_0+C_1\le\frac{8N}{d},
    \quad X(\mathbf{C})\le\frac{4(N+1)}{N-1}\right\}.
\end{align}
Each condition fails with probability at most $1/4$ by Markov's inequality, since
\begin{align}
    \E[C_0+C_1]=\frac{2(N-2)}d,
    \qquad
    \E[X(\mathbf{C})]
    =\frac{N+1}{N-1}
    \left[1-\left(1-\frac1d\right)^{N-1}\right]
    \le\frac{N+1}{N-1}.
\end{align}
Thus $\mathbb{P}(\mathcal G)\ge1/2$. Recall that
\begin{align*}
    u_t=C_0+2+t,
    \qquad v_t=C_1+2-t,
    \qquad X(\mathbf{C})=\frac{N+1}{d^2}\sum_{j=0}^{d-1}\frac1{C_j+1}.
\end{align*}
On $\mathcal G$, uniformly for $-1\le t\le1$,
\begin{align}
    u_t,v_t\le\frac{8N}{d}+3\le\frac{10N}{d},
    \qquad
    h_{\mathbf{C}}(t)^2\le X(\mathbf{C})\le\frac{6(N+1)}{N}.
    \label{eq:phase-one-copy-event-estimates}
\end{align}
Here $N\ge2d\ge4$, and $h_{\mathbf{C}}(t)^2\le X(\mathbf{C})$ because the two added occupations $1+t$ and $1-t$ are nonnegative. Substituting \cref{eq:phase-one-copy-event-estimates} into \cref{eq:phase-one-copy-curvature} gives
\begin{align}
    h_{\mathbf{C}}''(t)\ge\frac{3d}{2000\sqrt6\,N^2}
    \qquad(\mathbf{C}\in\mathcal G).
\end{align}
By \cref{eq:phase-one-copy-curvature}, $h_{\mathbf{C}}''(t)\ge0$ also outside $\mathcal G$. These outcomes therefore contribute nonnegatively to the expected integral against the nonnegative weight $1-|t|$. To lower-bound this expectation, we may retain only the outcomes in $\mathcal G$, which occur with probability at least $1/2$.

Finally, swapping coordinates $0$ and $1$ preserves the distribution of $\mathbf{C}$ and the value of $X$, so
\begin{align*}
    \E\left[\sqrt{X(\mathbf{C}+2\mathbf{e}_0)}\right]
    =\E\left[\sqrt{X(\mathbf{C}+2\mathbf{e}_1)}\right],
    \qquad
    \E[h_{\mathbf{C}}(1)]=\E[h_{\mathbf{C}}(-1)].
\end{align*}
Using the definition of $h_{\mathbf{C}}$ and the integral formula for a second difference gives
\begin{align}
    &\E\left[\sqrt{X(\mathbf{C}+2\mathbf{e}_0)}-\sqrt{X(\mathbf{C}+\mathbf{e}_0+\mathbf{e}_1)}\right] \nonumber\\
    &\qquad=\E\bigl[h_{\mathbf{C}}(1)-h_{\mathbf{C}}(0)\bigr] \nonumber\\
    &\qquad=\frac12\E\bigl[h_{\mathbf{C}}(1)+h_{\mathbf{C}}(-1)-2h_{\mathbf{C}}(0)\bigr] \nonumber\\
    &\qquad=\frac12\E\left[\int_{-1}^1(1-|t|)h_{\mathbf{C}}''(t)\,dt\right] \nonumber\\
    &\qquad\ge\frac12\cdot\frac12\cdot\frac{3d}{2000\sqrt6\,N^2}
    =\frac{3d}{8000\sqrt6\,N^2}.
    \label{eq:phase-one-copy-second-difference}
\end{align}
The last inequality follows by restricting to $\mathcal G$, using $\mathbb{P}(\mathcal G)\ge1/2$ and $\int_{-1}^1(1-|t|)\,dt=1$. The factor $1/4$ comes from multiplying the prefactor $1/2$ by the probability lower bound $1/2$ for the good event; the integral of the weight contributes a factor of $1$.

Since $N\ge2d\ge4$, we have
\begin{align*}
    \binom N2
    =\frac{N(N-1)}2
    =\frac{N^2}{2}\frac{N-1}{N}
    \ge\frac{N^2}{2}\cdot\frac34
    =\frac{3N^2}{8}.
\end{align*}
Combining \cref{eq:phase-one-copy-covariance-lower,eq:phase-one-copy-second-difference} and applying this bound gives
\begin{align}
    \operatorname{Var}(\sqrt X)
    &\ge\binom N2\frac{d-1}{d^2}
    \left(\frac{3d}{8000\sqrt6\,N^2}\right)^2 \nonumber\\
    &\ge\frac{3N^2}{8}\frac{d-1}{d^2}
    \frac{9d^2}{384\,000\,000\,N^4} \nonumber\\
    &=\frac9{1\,024\,000\,000}\frac{d-1}{N^2}
    \ge10^{-9}\frac{d-1}{N^2}.
    \qquad\qedhere
\end{align}
\end{proof}

\subsection{Uniform completion and sample-complexity inversion}
\label{app:phase-one-copy-completion}

The estimates above prove \Cref{lem:phase-one-copy-variance-main} with $c_0=10^{-9}$ and $C_0=30$. Together with the exact decomposition, they give the more precise bound
\begin{align}
    \left(1-\frac1d\right)^{N+1}+10^{-9}\frac{d-1}{N^2}
    \le\delta_{N,d}
    \le\left(1-\frac1d\right)^{N+1}+30\frac{d-1}{N^2},
    \qquad N\ge2d.
    \label{eq:phase-one-copy-dense}
\end{align}
The numerical constants are not optimized. We record the uniform completion with $c=10^{-9}$ and $C=68$, and then the full inversion of these bounds.

\begin{proof}[Proof of \Cref{thm:phase-one-copy-uniform}]
For $N\ge2d$, \Cref{lem:phase-one-copy-decomposition} and \Cref{prop:phase-one-copy-upper,prop:phase-one-copy-lower} give \cref{eq:phase-one-copy-dense}. Since $N\le N+d\le3N/2$, they also give \cref{eq:phase-one-copy-uniform} with $c=10^{-9}$ and $C=68$ in this regime.

If $1\le N<2d$, then
\begin{align}
    \left(1-\frac1d\right)^{N+1}\ge\left(1-\frac1d\right)^{2d}\ge\frac1{16},
    \qquad
    \frac{d-1}{(N+d)^2}\le\frac14.
\end{align}
The middle inequality follows from the increase of $(1-1/d)^d$ for $d\ge2$, whose value at $d=2$ is $1/4$. Thus $(1-1/d)^{N+1}+(d-1)/(N+d)^2\le5(1-1/d)^{N+1}$. Since $\delta_{N,d}\ge (1-1/d)^{N+1}$ by \Cref{lem:phase-one-copy-decomposition} and $\delta_{N,d}\le1\le16(1-1/d)^{N+1}$, the same constants $c,C$ also suffice here.
\end{proof}

\begin{proof}[Proof of \Cref{cor:phase-one-copy-complexity}]
Suppose $\delta_{N,d}\le\delta\le1/4$. The exact lower bound $\delta_{N,d}\ge (1-1/d)^{N+1}$ implies
\begin{align}
    N+1
    \ge\frac{\log(1/\delta)}{-\log(1-1/d)}
    \ge(d-1)\log\frac1\delta.
\end{align}
Using $N+1\le2N$ and $d-1\ge d/2$, we obtain
\begin{align}
    N\ge\frac d4\log\frac1\delta.
    \label{eq:phase-one-copy-log-necessary}
\end{align}
In particular, $N+d\le4N$, because $\log(1/\delta)\ge\log4>4/3$. The rational term in \cref{eq:phase-one-copy-uniform} now gives
\begin{align}
    \delta\ge c\,\frac{d-1}{(N+d)^2}
    \ge\frac{cd}{32N^2},
    \qquad
    N\ge\sqrt{\frac{cd}{32\delta}}.
    \label{eq:phase-one-copy-root-necessary}
\end{align}
Combining \cref{eq:phase-one-copy-log-necessary,eq:phase-one-copy-root-necessary} proves the required lower bound.

For sufficiency, choose
\begin{align}
    N=\left\lceil
    8\left(d\log\frac1\delta+\sqrt{\frac d\delta}\right)
    \right\rceil.
    \label{eq:phase-one-copy-sufficient}
\end{align}
This ensures $N\ge2d$, $(1-1/d)^{N+1}\le e^{-(N+1)/d}\le\delta^8\le\delta/2$, and $30(d-1)/N^2\le30\delta/64<\delta/2$. The upper bound in \cref{eq:phase-one-copy-dense} therefore gives $\delta_{N,d}\le\delta$.
\end{proof}

%% file: CloningBounds.tex
\section{Bounds for non-resolving orbit families}
\label{app:bounds_nonresolving}

In this appendix, for the benefit of the Outlook section (\cref{sec:outlook}), we give an upper and lower bound for the set of states which form a well-behaved, but not necessarily resolving orbit inside a finite-dimensional $\cH$. Throughout this section, we fix $\cS \subseteq \cal{D}(\cH)$ to be a pure state family which is an orbit for some group $G$ as defined in~\Cref{def:resolving_orbit}. That is, there exists a compact group $G$ admitting a unimodular Haar measure and a strongly continuous unitary representation $\phi:G\rarr\rU(\cH)$ such that
\begin{equation*}
    \cS = \{\phi(g)\, \ketbra{\psi}{\psi} \,\phi^\dag(g) : g \in G\}
\end{equation*}
for some pure state $ \ketbra{\psi}{\psi} \in \cal{D}(\cH)$. We define the $N$-th moment of $\cS$ as 
\begin{equation*}
    \tau_{N} := \int_{\ket{\psi} \in \cS} \left(\ketbra{\psi}{\psi}\right)^{\otimes N} \nu(d\ket{\psi})
\end{equation*}
where the integral is taken with respect to the normalized uniform measure on $\cS$. Furthermore, we let $\pi_N$ be the projector onto the support of $\tau_N$. Note that $\pi_{N+K}\le\pi_N\otimes\pi_K$; in particular $\pi_{N+K}(\cI\otimes\pi_K)=\pi_{N+K}$. We have the following general bound with respect to cloning for a general orbit. 

\begin{theorem}
\label{thm:general_cloning_bound}
Let $\cH$ be a finite-dimensional Hilbert space, $N$ and $K$ be two natural numbers and let $\cS\subseteq \cal{D}(\cH)$ be an orbit of a compact $G$ as described above. Then:
\begin{enumerate}
    \item There exists an explicit $N \rightarrow N+K$ cloning channel that achieves the cloning fidelity of at least $\frac{1}{\|\Tr_B(\pi_{N+K})\|_{\infty}} $, where $\Tr_B$ denotes tracing out the last $K$ registers of $\pi_{N+K}$. 
    \item For any $N \rightarrow N+K$ cloning channel, the optimal worse-case cloning fidelity is bounded above by
    \begin{equation*}
        F_{\cS}^{\star}(N,N+K) \leq \| \tau_{N+K} \|_{\infty} \cdot  \operatorname{rank}(\pi_{N})
    \end{equation*}
\end{enumerate}
\end{theorem}

As a corollary of the above theorem, the cloning complexity for any orbit can be bounded by 
\begin{equation*}
        \min\{N: \|\tau_{N+1} \|_{\infty} \cdot \operatorname{rank}(\pi_{N}) \geq 1 - \epsilon  \} \leq N_{\text{min}, \cS}(\epsilon) \leq \min\{N: \frac{1}{\|\Tr_B(\pi_{N+1})\|_{\infty}} \geq 1 - \epsilon^2 \} 
\end{equation*}

We remark that if $\cS$ is a resolving orbit, then $\operatorname{rank}(\pi_{N}) \cdot \tau_N = \pi_N$. Since 
\begin{equation*}
    \|\Tr_B(\pi_{N+K})\|_{\infty} =  \|\Tr_B(\operatorname{rank}(\pi_{N+K}) \cdot \tau_{N+K})\|_{\infty} = \frac{\operatorname{rank}(\pi_{N+K})}{\operatorname{rank}(\pi_{N})},
\end{equation*}
and 
\begin{equation*}
    \| \tau_{N+K} \|_{\infty}  = \frac{1}{\operatorname{rank}(\pi_{N+K})},
\end{equation*}
this shows that~\Cref{thm:general_cloning_bound} matches the bound given in~\Cref{sec:orbits_cloners}. For phase state $\cS$, we have
\begin{align*}
    \tau_{N} &= \frac{1}{d^N} \sum_{\mathbf a\in\cB_N} \binom{N}{\mathbf a} \ketbra{ \Symb{N}{\mathbf a} }{ \Symb{N}{\mathbf a} }.  \\
    \pi_N&=\sum_{\mathbf a\in\cB_N} \ketbra{\Symb{N}{\mathbf a}}{\Symb{N}{\mathbf a}}.
\end{align*}
Here we remark that $\pi_{N}$ projects onto the symmetric subspace. Using the bound given in~\Cref{thm:general_cloning_bound}, we obtain a worst-case lower bound of $\frac{\binom{N +d-1}{N}}{\binom{N +K +d-1}{N +K}}$, the same bound as the general state cloner. 

We are now ready to show~\Cref{thm:general_cloning_bound}, we remark that the proof follows a similar structure to that of~\cite{werner1998optimal}.

\begin{proof}
    For part 1, let $c = \frac{1}{\|\Tr_B(\pi_{N+K})\|_{\infty}}$. We wish to show that the cloning channel $T$ given by
    \begin{equation*}
        T(\rho_{N}) := c \cdot \pi_{N+K} \left(\rho_{N} \otimes \pi_K \right) \pi_{N+K} + \left((\cI - c \cdot \Tr_B(\pi_{N+K}))^{\frac{1}{2}} \rho_{N} (\cI - c \cdot \Tr_B(\pi_{N+K}))^{\frac{1}{2}}\right) \otimes \ketbra{0}{0}
    \end{equation*}
    is a valid quantum channel. Write 
    \begin{equation*}
        \pi_K = \sum_{i \in [\operatorname{rank}(\pi_{K})]} \ketbra{\psi_i}{\psi_i}
    \end{equation*}
    as the spectral decomposition for $\pi_{K}$. For $i \in [\operatorname{rank}(\pi_{K})]$, define the Kraus operator 
    \begin{equation*}
        A_i = \sqrt{c} \cdot \pi_{N+K} (\cI_{N} \otimes \ket{\psi_i})
    \end{equation*}
    and 
    \begin{equation*}
        A_0 = (\cI - c \cdot \Tr_B(\pi_{N+K}))^{\frac{1}{2}} \otimes \ket{0}.
    \end{equation*}
    By simple calculation, one can verify that
    \begin{equation*}
        T(\rho_{N}) = \sum_{i=0}^{\operatorname{rank}(\pi_K)} A_i \rho_{N} A_i^{\dagger},
    \end{equation*}
    and 
    \begin{align*}
        \sum_{i=0}^{\operatorname{rank}(\pi_K)}  A^{\dagger}_i A_i &= \cI. 
    \end{align*}
    Thus showing that $T$ is a valid channel. For any state $\ket{\psi} \in \cS$
    \begin{align*}
        &\Tr\!\left[
    \ketbra{\psi}{\psi}^{\otimes(N+K)}
    T\!\left(\ketbra{\psi}{\psi}^{\otimes N}\right)
\right]  \\
        &\geq  \frac{1}{\|\Tr_B(\pi_{N+K})\|_{\infty}} \cdot \bra{\psi}^{\otimes N+K} \pi_{N+K} \left( \ketbra{\psi}{\psi}^{\otimes N} \otimes \pi_K \right)  \pi_{N+K}\ket{\psi}^{\otimes N+K} = \frac{1}{\|\Tr_B(\pi_{N+K})\|_{\infty}}. 
    \end{align*}
    This proves part 1 of the theorem. To show part 2, for any cloning channel $T'$, since $\ketbra{\psi}{\psi}^{\otimes N} \leq \pi_{N}$, 
    \begin{align*}
        F_{\cS}^{\star}(N,N+K) &\leq \int_{\cS} \Tr(\ketbra{\psi}{\psi}^{\otimes N+K} T'(\ketbra{\psi}{\psi}^{\otimes N})) \nu(d\ket{\psi}) \\
        &\leq \Tr(\int_{\cS}  \ketbra{\psi}{\psi}^{\otimes N+K} T'(\pi_N) \nu(d\ket{\psi}) ) \\
        &\leq  \|\tau_{N+K}\|_{\infty}  \cdot \operatorname{rank}(\pi_{N})
    \end{align*}
    where the last line follows H\"older and the channel being trace preserving. This proves part 2 of the theorem.

\end{proof}